\documentclass[11pt]{article}
\usepackage[T1]{fontenc} 
\usepackage[utf8]{inputenc}
\usepackage{amsthm,amsmath,amssymb}
\usepackage[colorlinks]{hyperref}
\usepackage[nameinlink,capitalise,noabbrev]{cleveref}
\usepackage{mathtools}
\usepackage{dsfont}
\usepackage[dvipsnames]{xcolor}
\usepackage{soul}
\usepackage{enumitem}
\usepackage{tikz}
\usetikzlibrary{topaths,calc}
\usepackage{pgfplots}
\usepackage{subcaption}
\usepackage{array}
\usepackage{rotating}

\usepackage[absolute]{textpos}

\usepackage{hyperref}
\hypersetup{
    colorlinks=true,
    linkcolor=blue,
    citecolor=blue,
    filecolor=magenta,      
    urlcolor=cyan,
}
\usepackage[linesnumbered,algoruled,boxed,lined]{algorithm2e}
\usepackage{complexity} %
\usepackage{todonotes} %
\presetkeys{todonotes}{inline}{}
\usepackage{enumitem}
\setlist{nosep}
\usepackage{thm-restate}

\usepackage{vmargin}
\setmarginsrb{1in}{1in}{1in}{1in}{0pt}{0pt}{0pt}{6mm}

\usepackage{fontawesome5}

\newcommand{\calC}{\mathcal{C}}

\newcommand{\OO}{\mathcal{O}}
\newcommand{\Oh}{\mathcal{O}}

\newcommand{\calR}{\mathcal{R}}

\newcommand{\ETH}{\textsf{ETH}}

\newcommand{\yes}{\textsc{Yes}}

\newtheorem{theorem}{Theorem}[section]
\newtheorem{lemma}[theorem]{Lemma}

\newtheorem{claim}[theorem]{Claim}
	\crefname{claim}{claim}{Claims}
	\crefformat{claim}{#2Claim~#1#3}
	\Crefformat{claim}{#2Claim~#1#3}

\newtheorem{proposition}[theorem]{Proposition}
\newtheorem{reduction rule}{Reduction Rule}[section]
\newtheorem{marking-scheme}{Marking Scheme}[section]
\newtheorem{definition}[theorem]{Definition}

\newcommand{\threedmc}{\textsc{$3$-DMC}}

\newcommand{\dmcthree}{\textsc{$3$-DMC}}
\newcommand{\term}{\texttt{T}}

\newcommand{\rr}{\texttt{r}}
\newcommand{\ff}{\texttt{f}}
\newcommand{\rrr}{\textsf{r}}
\newcommand{\fff}{\textsf{f}}

\newcommand{\N}{\mathds{N}}

\newcommand{\permcsp}{\textsc{Permutation CSP}}

\newcommand\downclosed{downwards-closed\xspace}
\newcommand\card[1]{|#1|}
\newcommand\calD{\mathcal{D}}
\newcommand\calE{\mathcal{E}}
\newcommand\gridrank{\mathsf{gr}}
\newcommand\tww{\mathsf{tww}}
\newcommand{\bN}{\N}
\newcommand\Adj{\mathsf{Adj}}
\newcommand\twwCSP[1]{\textsc{Twin-width-#1 Permutation CSP}}
\newcommand\CSPinst{\mathcal{I}}
\newcommand\myparagraph[1]{\paragraph*{#1}}
\newcommand\col{\mathsf{col}}
\newcommand\inc{\mathsf{inc}}

\newcommand{\corecut}[1]{\mathrm{core}(#1)}
\newcommand{\corecutG}[2]{\mathrm{core}_{#2}(#1)}
\newcommand{\flow}{\mathcal{P}}
\newcommand{\witnessflow}{\hat{\mathcal{P}}}

\definecolor{meikeColour}{rgb}{0.41, 0.16, 0.38}

\definecolor{larspaloma}{HTML}{2E86C1}

\definecolor{TM}{HTML}{FFA500}

\definecolor{MS}{HTML}{228b22}

\definecolor{MP}{HTML}{00C000}

\newcommand\bfA{\mathbf{A}}

\usetikzlibrary{arrows}
\usetikzlibrary{arrows.meta}
\usetikzlibrary{decorations.pathmorphing}
\usetikzlibrary{decorations.markings}
\usetikzlibrary{calc}
\usetikzlibrary{positioning}

\tikzset{
  circ/.style = {circle,draw,fill,inner sep=1.3pt},
  mcirc/.style = {circle,draw,fill,inner sep=1pt},
  circR/.style = {circle,draw=red,fill=red,text=red,inner sep=1.3pt},
  circG/.style = {circle,draw=green,fill=green,text=green,inner sep=1.3pt},
  circB/.style = {circle,draw=blue,fill=blue,text=blue,inner sep=1.3pt},
  circb/.style = {circle,draw=blue,fill=blue,text=blue,inner sep=1.1pt},
  circr/.style = {circle,draw=red,fill=red,inner sep=1pt},
  scirc/.style = {circle,draw,fill,inner sep=.8pt},
  invisible/.style = {draw=none,inner sep=0pt,font=\tiny},
  nonedge/.style={decorate,decoration={snake,amplitude=.3mm,segment length=1mm},draw}
}

\newif\ifcomment
\commentfalse
\commenttrue

\newcommand{\frontpageformat}{arxiv}

\begin{document}
\ifthenelse{\equal{\frontpageformat}{submission}}{%
  \author{anonymous}
  \title{Fixed-parameter tractability of \textsc{Directed Multicut} with three terminal pairs parameterized by the size of the cutset:\\twin-width meets flow-augmentation}
\date{}

\begin{titlepage}
\def\thepage{}
\thispagestyle{empty}
\maketitle
}{%
\author{Meike Hatzel\thanks{National Institute of Informatics, Tokyo, Japan. \texttt{meikehatzel@nii.ac.jp}} 
  \and 
    Lars Jaffke\thanks{University of Bergen, Norway.
    \texttt{lars.jaffke@uib.no}} 
  \and 
    Paloma T. Lima\thanks{IT University of Copenhagen, Denmark.
    \texttt{palt@itu.dk}}
  \and 
    Tom\'a\v{s} Masa\v{r}\'ik\thanks{University of Warsaw, Poland. \texttt{masarik@mimuw.edu.pl}}
  \and 
    Marcin Pilipczuk\thanks{University of Warsaw, Poland. \texttt{m.pilipczuk@mimuw.edu.pl}}
  \and
    Roohani Sharma\thanks{Max Planck Institute for Informatics, Saarland Informatics Campus, Saarbr\"{u}cken, Germany. \texttt{rsharma@mpi-inf.mpg.de}}
  \and
    Manuel Sorge\thanks{TU Wien, Austria. \texttt{manuel.sorge@ac.tuwien.ac.at}}
}
\title{Fixed-parameter tractability of \textsc{Directed Multicut} with three terminal pairs parameterized by the size of the cutset:\\twin-width meets flow-augmentation\thanks{The research leading to the results presented in this paper was
		partially carried out during the Parameterized Algorithms Retreat of the University of Warsaw, PARUW 2022,
		held in B\k{e}dlewo in April 2022.
This research is a part of projects that have received funding from the European Research Council (ERC)
under the European Union's Horizon 2020 research and innovation programme
Grant Agreement 714704 (TM, MP) and 648527 (MH), from the Alexander von Humboldt Foundation (MS), from the Research Council of Norway (LJ), and by a fellowship within the IFI programme of the German Academic Exchange Service (DAAD) (MH).
}}

\date{}

\begin{titlepage}
\def\thepage{}
\thispagestyle{empty}
\maketitle
\begin{textblock}{20}(2, 13.9)
\includegraphics[width=40px]{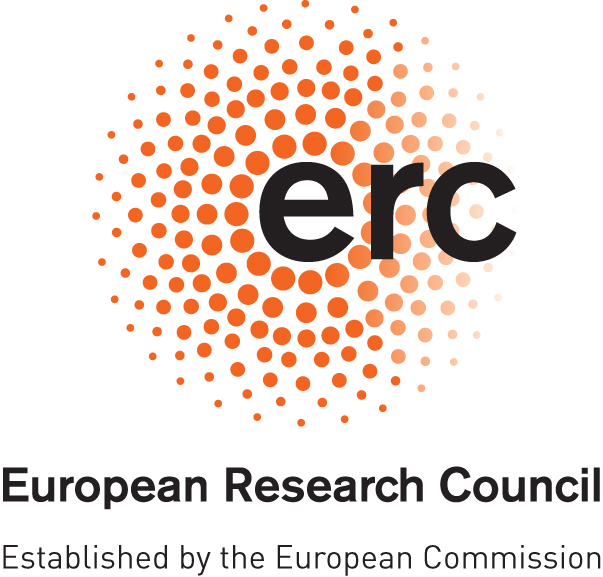}%
\end{textblock}
\begin{textblock}{20}(2, 14.7)
\includegraphics[width=40px]{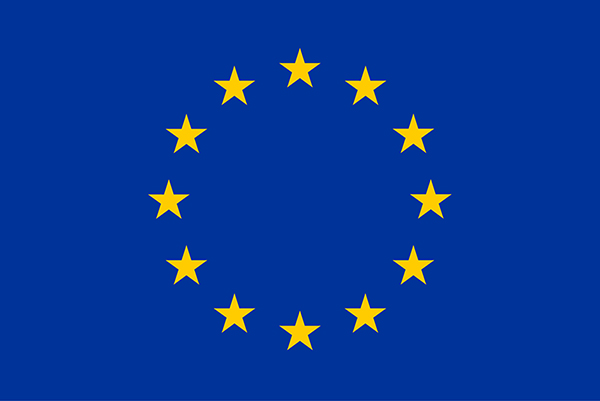}%
\end{textblock}
}

\begin{abstract}

We show fixed-parameter tractability of the \textsc{Directed Multicut} problem with three terminal pairs (with a randomized algorithm).
This problem, given a directed graph $G$, pairs of vertices (called \emph{terminals})
$(s_1,t_1)$, $(s_2,t_2)$, and $(s_3,t_3)$, and an integer $k$, asks to find a set of 
at most $k$ non-terminal vertices in $G$ that intersect all $s_1t_1$-paths, all
$s_2t_2$-paths, and all $s_3t_3$-paths. 
  The parameterized complexity of this case has been open since 
Chitnis, Cygan, Hajiaghayi, and Marx proved fixed-parameter tractability 
of the 2-terminal-pairs case at SODA 2012, and
Pilipczuk and Wahlstr\"{o}m proved the W[1]-hardness of the 4-terminal-pairs case
at SODA 2016. 

On the technical side, we use two recent developments in parameterized algorithms.
Using the technique of \emph{directed flow-augmentation} [Kim, Kratsch, Pilipczuk, Wahlstr\"{o}m, STOC 2022] we cast the problem as a CSP problem with few variables and constraints over a large ordered domain.
We observe that this problem can be in turn encoded as an FO model-checking task over a structure
consisting of a few 0-1 matrices.
We look at this problem through the lenses of twin-width, a recently introduced structural
parameter [Bonnet, Kim, Thomass\'{e}, Watrigant, FOCS 2020]:
By a recent characterization [Bonnet, Giocanti, Ossona de Mendes, Simon, Thomass\'{e}, Toru\'{n}czyk, STOC 2022] the said FO model-checking task can be done in FPT time
if the said matrices have bounded grid rank.
To complete the proof, we show an irrelevant vertex rule: If any of the matrices in
the said encoding has a large grid minor, 
 a vertex corresponding to the ``middle'' box in 
the grid minor can be proclaimed irrelevant --- not contained in the sought solution --- and
thus reduced.
\end{abstract}

\newpage
\tableofcontents
\end{titlepage}

\section{Introduction}\label{sec:introduction}

Parameterized complexity studies the existence of \emph{fixed-parameter algorithms}: 
algorithms with running time bound by $f(k) \cdot n^c$, where $n$ is the size of the input, $c$ is an arbitrary constant, 
$f$ is an arbitrary computable function, and $k$ is the parameter, which is a selected secondary measure of the input that is intended to reflect the hardness of the instance. 

Graph separation problems yield a class of combinatorial problems where the goal is to find a small vertex or edge set in the given graph that satisfies some separation requirements.
For example, \textsc{Multiway Cut} equips the input graph $G$ with a set $T \subseteq V(G)$ of terminals and asks to cut all paths between any two distinct
terminals, whereas \textsc{Subset Feedback Vertex Set} equips the input graph $G$ with a set $R \subseteq V(G)$ of \emph{red vertices} and asks to cut all cycles that contain at least one red vertex. 
The study of graph separation problems,
with the cardinality of the sought cut as the natural choice of the parameter, has been one of the more vivid areas of parameterized complexity in the recent 15 years. 
A number of interesting algorithmic techniques emerged:
important separators~\cite{marx:1,dfvs}, applications of matroid techniques~\cite{ms1,ms2},
shadow removal~\cite{dir-mwc,marx:multicut}, randomized contractions~\cite{rand-contr,lean-decomp}, 
LP-guided branching~\cite{mwc-a-lp,sylvain,IwataWY16}, and treewidth reduction~\cite{tw-red}, among others.

The progress somewhat stalled around 5 years ago in the following state: we understood 
the complexity of the main bulk of graph separation problems in undirected graphs, 
mostly thanks to the wide variety of algorithmic techniques therein.
However, in directed graphs, a number of questions remained widely open. 

The \textsc{Multicut} problem is, given a graph $G$ and a family $\mathcal{T}$ of \emph{pairs}
of vertices (called \emph{terminals}), to delete a minimum number of non-terminal vertices so that for every terminal pair $(s,t) \in \mathcal{T}$, there is no path from $s$ to $t$ in the remaining graph. 
The parameterized complexity of this problem in undirected graphs, after being a long standing open problem
for a while, has been resolved around 2010 independently by two groups of researchers~\cite{marx:multicut,thomasse:multicut}. 
In directed graphs, the problem in full generality was quickly observed to be W[1]-hard~\cite{marx:multicut}. 
However, some restrictions turned out to be tractable: the case of directed acyclic graphs~\cite{multicut-in-DAGs}, 
\textsc{Directed Multiway Cut}, where we are given just a set of terminals and we ask to cut all paths between every pair of distinct terminals, or \textsc{Directed Multicut} with two terminal pairs~\cite{dir-mwc}. Observe that the one-terminal-pair case is just the classic \textsc{Minimum Cut} problem.
In 2015, Pilipczuk and Wahlstr\"{o}m~\cite{PilipczukW18a}
provided a hardness reduction for the four-terminal-pairs
case, leaving the three-terminal-pairs case open until now.

\textsc{Directed Multicut} with three terminal pairs was by far not the only open problem
left in the parameterized complexity of directed graph separation problems. 
Other open problems included the notoriously difficult \textsc{Chain SAT} problem~\cite{ChitnisEM17} and most of the problems in the \emph{weighted setting}.
Here, the deletable objects (edges or vertices) have integer weights, and the question is to find a solution of cardinality at most $k$ and minimum total weight (where $k$ is the given parameter).
One of the reasons for such a state of affairs was a lack of algorithmic techniques in directed
graphs: Among all the aforementioned tools in undirected graphs, only important separators
and shadow removal generalize to directed graphs~\cite{dir-mwc,dsfvs}. 

Very recently, at STOC 2022, a new algorithmic technique for cut problems
in directed graphs has been presented by 
Kim, Kratsch, Pilipczuk, and Wahlstr\"{o}m~\cite{dfl-arxiv}, namely \emph{flow-augmentation}.
This new technique led to fixed-parameter algorithms for \textsc{Chain SAT} and numerous weighted
versions of graph separation problems. 
In this work, we use it to answer the question of the parameterized
complexity of \textsc{Directed Multicut} with three terminal pairs positively.

\begin{restatable}{theorem}{DMCthreeFPTTheorem}\label{thm:dmcthree-fpt}
  \textsc{Directed Multicut} with three terminal pairs is fixed-parameter tractable
  when parameterized by the size of the cutset (with a randomized algorithm).  
\end{restatable}

Flow-augmentation alone is by far not enough to show \cref{thm:dmcthree-fpt}. 
In the basic usage, the main tool of directed flow-augmentation~\cite{dfl-arxiv} can be stated as follows:
\begin{theorem}\label{thm:dfl-basic}
There exists a polynomial-time randomized algorithm that, given a directed graph~$G$,
two distinguished vertices $s,t \in V(G)$, and an integer $k$, outputs
a set $A \subseteq V(G) \times V(G)$ (called \emph{augmentation edges}) such that
for every minimal edge $st$-cut $Z$ of cardinality at most $k$, with probability
$2^{-\Oh(k^4 \log k)}$ the cut $Z$ becomes a minimum edge $st$-cut in $G+A$. 
\end{theorem}
Here, $G+A$ is the graph obtained from $G$ by adding the arcs in $A$.
That is, with good probability, the added arcs $A$ not only do not break the edge cut $Z$ (that
is, connect the $s$-side of the cut to the $t$-side), but also increase the connectivity
of the graph so that $Z$ becomes a minimum-cardinality cut.
We remark that all randomization in \cref{thm:dmcthree-fpt} comes from
\cref{thm:dfl-basic}. That is, if \cref{thm:dfl-basic} were deterministic,
so would be the algorithm of \cref{thm:dmcthree-fpt}.

\looseness=-1
The following point of view on \cref{thm:dfl-basic} turns out to be particularly useful.
Let $G$ be a directed graph, $s,t \in V(G)$, and let $k$ be the size of a minimum edge
$st$-cut.
How does the space of all minimum edge $st$-cuts look like?
Let $\mathcal{P}$ be any maximum $st$-flow, seen as a collection of $k$ edge-disjoint paths
from $s$ to $t$.
Any minimum edge $st$-cut contains exactly one edge from each path in $P$.
Furthermore, for every $P_1,P_2 \in \mathcal{P}$ and every $u_1 \in V(P_1)$, $u_2 \in V(P_2)$,
if $G$ contains a path $Q$ from $u_1$ to $u_2$ that does not contain any edge of $\mathcal{P}$,
then any minimum edge $st$-cut cannot contain an edge of $P_1$ \emph{after $u_1$}
and an edge of $P_2$ \emph{before $u_2$} at the same time.
This motivates the following CSP formulation.
Every path $P \in \mathcal{P}$ becomes 
a variable $x(P)$ with domain $E(P)$, ordered naturally along $P$.
For every tuple $(P_1,P_2,u_1,u_2)$ as above (i.e., $P_1,P_2 \in \mathcal{P}$, $u_1 \in V(P_1)$, $u_2 \in V(P_2)$, $G$ contains a path from $u_1$ to $u_2$ that does not contain any edge of $\mathcal{P}$), we introduce a constraint $(x(P_1) \leq u_1) \vee (x(P_2) \geq u_2)$, where the inequalities have the natural meaning of being before/after the corresponding vertex along the corresponding path.
It is relatively easy to see that the space of all feasible solutions to such a CSP instance
is \emph{exactly} the space of all minimum edge $st$-cuts in $G$.
In this light, \cref{thm:dfl-basic} can be understood as follows: we can subsample
the space of all minimal edge $st$-cuts in $G$ of cardinality at most $k$, so that
every cut is sampled with good probability (i.e., $2^{-\Oh(k^4 \log k)}$) 
and the subsampled set can be described by the aforementioned CSP instance. 

A meticulous reader may observe that flow-augmentation speaks about edge cuts
while \textsc{Directed Multicut} asks for a vertex cut. 
However, in directed graphs there are standard reductions between these two variants of the problem.
Thus, the above framework of a CSP formulation can be easily adapted to minimal vertex $st$-cuts
of cardinality at most $k$ (where $s$ and $t$ are undeletable).

In the context of \textsc{Directed Multicut} with three terminal pairs, we can use 
flow-augmen\-ta\-tion as follows. Let $S$ be an inclusion-wise minimal solution
to an input instance $(G,k,(s_i,t_i)_{i=1,2,3})$. 
Clearly, $S = S_1 \cup S_2 \cup S_3$ where $S_i$ is a minimal vertex $s_it_i$-cut. 
Hence, we can apply flow-augmentation separately to $(G,k,(s_i,t_i))$ for $i=1,2,3$,
obtaining a set of augmentation edges~$A_i$. With good probability, for every $i=1,2,3$
the set $S_i$ becomes a minimum vertex $st$-cut in $G+A_i$. 
Let $\mathcal{P}_i$ be a maximum (vertex-capacitated) $s_it_i$-flow in $G+A_i$
and consider the aforementioned CSP formulation with variables $\{x_i(P)~|~P \in \mathcal{P}_i\}$
and sets of constraints $\mathcal{C}_i$. 

\looseness=-1
The crux of the difficulty of \textsc{Directed Multicut} lies in the fact that
the sets $S_i$ may not be pairwise disjoint; in a sense, we can save on reusing some vertices
in $S$ to separate multiple terminal pairs. 
In the CSP regime, it means that for some $i,j \in \{1,2,3\}$ and $P \in \mathcal{P}_i$
and $Q \in \mathcal{P}_j$, the variables $x_i(P)$ and $x_j(Q)$ describe \emph{the same vertex}
of $S$. 
Note that there is only $2^{\Oh(k \log k)}$ options of which pairs of variables describe the same
vertex; we can exhaustively guess the set $K$ of all pairs $(i,j,P,Q)$ as above. 
Every $(i,j,P,Q) \in K$ induces a constraint that $x_i(P)$ and $x_j(Q)$ is the same vertex;
in the CSP language, this is a permutation constraint, denoted henceforth $C(i,j,P,Q)$, between a subset of the domain of $x_i(P)$ and a subset of the domain of~$x_j(Q)$.

\looseness=-1
It is important to observe that --- assuming the flow-augmentation steps were successful
and in the branching step we made the correct choice of which variables describe the same vertex
--- the final CSP instance is an \emph{equivalent reformulation} of the original 
\textsc{Directed Multicut} instance. That is, every solution to the obtained CSP instance
gives a set of non-terminal vertices that cuts all paths from $s_i$ to $t_i$ for $i=1,2,3$. 
The coincidences guessed in the branching step 
determine the cardinality of the obtained cut, and we can terminate all branches
that lead to cuts larger than $k$. 
Furthermore, every inclusion-wise minimal solution to the original \textsc{Directed Multicut}
instance that is compliant with the flow-augmentation and branching steps yields
a feasible solution to the obtained CSP instance.
Thus, it ``only'' remains to solve the obtained CSP instance.

To this end, we need to understand how complex the permutation constraints $C(i,j,P,Q)$ can be.
Note that the number of variables is small --- bounded by $3k$ --- but the domains can be as large as $|V(G)|$.
A reader experienced in W[1]-hardness reductions may notice at this point that the complexity
of the permutation constraints is crucial:
If one allows arbitrary permutation constraints
(and constraints of the form $(x \leq a) \vee (y \geq b)$ for constants $a,b$ and variables $x,y$, as in the encoding of the space of all minimum cuts), one can easily provide a W[1]-hardness reduction
for the parameterization by the number of variables via the edge-choice gadgets (cf.~\cite{the-book} and \cref{app:edge-choice}). 

Here a recent major milestone in parameterized complexity comes into play: twin-width.
Introduced by Bonnet, Kim, Thomass\'{e}, and Watrigant in 2020~\cite{BonnetKTW22},
this structural complexity measure of graphs and,
more generally, binary structures has turned out to explain and provide a number of fixed-parameter tractability results.
Most importantly, a recent work from STOC 2022~\cite{BonnetGMSTT22} provides a fixed-parameter
algorithm for FO model checking on ordered structures of bounded twin-width and provides
a neat characterization under which conditions a 0-1 matrix gives bounded twin-width in the encoding.

In our case, the crucial notion is the one of a \emph{grid minor} of a matrix.\footnote{Note that this is an entirely different concept to a grid minor of a graph.}
Let $A$ be a $0$-$1$ matrix of dimension $n \times m$.
An \emph{$\ell \times \ell$ grid minor} consists of two 
sequences of indices $0 = i_0 < i_1 < i_2 < \ldots < i_\ell = n$ and 
$0 = j_0 < j_1 < j_2 < \ldots j_\ell = m$ such that for every $1 \leq \alpha,\beta \leq \ell$,
there is at least one cell $A[i,j]$ with value $1$ for some indices $i, j$ with $i_{\alpha-1} < i \leq i_\alpha$
  and $j_{\beta-1} < j \leq j_\beta$. 
  The value $\ell$ is often referred to as the \emph{size} of the grid minor.
One can deduce from the results of Bonnet~et~al.~\cite{BonnetGMSTT22} the following statement:
Our CSP formulation is FPT when parameterized by the number of variables, the number
of permutation constraints, and the maximum size of a grid minor of the matrices of the permutation constraints.
Hence, our task boils down to providing a bound on the size of a grid minor in our permutation constraints $C(i,j,P,Q)$.\footnote{The reader experienced in W[1]-hardness reductions may recall at this point the reduction showing W[1]-hardness parameterized by the number of variables for CSPs with permutation constraints:
  The permutations used therein for the edge-choice gadgets are of the form $\pi : \{0,1,\ldots,n^2-1\} \to \{0,1,\ldots,n^2-1\}$, with $\pi(xn+y) = yn+x$ for $0 \leq x,y < n$.
  These are exactly the permutations with largest possible grid minors of their associated permutation matrices.}

\looseness=-1
To this end and to complete our proof, we prove the following irrelevant-vertex rule:
There exists an integer $\ell$ depending only on $k$ such that if some constraint $C(i,j,P,Q)$ has grid minor of size at least $\ell$, then a vertex corresponding to a $1$ in the ``middle cell'' of the grid minor
(i.e., $A[i,j]$ for $i_{\lceil \ell/2 \rceil-1} < i \leq i_{\lceil \ell/2 \rceil}$,
 $j_{\lceil \ell/2 \rceil-1} < j \leq j_{\lceil \ell/2 \rceil}$)
is \emph{irrelevant}, that is, any so-called \emph{shadowless} solution to the input \textsc{Directed Multicut} instance does not contain the said vertex.
Hence, such a vertex can be reduced in a standard manner (and the analysis of the CSP formulation restarted).

The notion of a \emph{shadowless} solution comes from the technique of \emph{shadow removal},
pivotal for the fixed-parameter algorithms for \textsc{Directed Multiway Cut}~\cite{dir-mwc}
and \textsc{Directed Subset Feedback Vertex Set}~\cite{dsfvs}. 
In the context of \textsc{Directed Multicut}, a solution $Z$ is \emph{shadowless}
if for every non-terminal vertex $v \notin Z$, the graph $G-Z$ features
a path from $v$ to one of the terminals $t_i$ and a path from one of the terminals $s_j$
to $v$ (note that necessarily the $t$-terminal and the $s$-terminal have distinct indices).
In short, the shadow-removal technique~\cite{dir-mwc,dsfvs} allows us 
to focus on the following task:
Given a \textsc{Directed Multicut} instance
$(G,k,(s_i,t_i)_{i=1,2,3})$, find any solution if there exists a shadowless solution
(i.e., the algorithm is allowed to fail if there is a solution, but not a shadowless one).

This last part of the proof --- the irrelevant-vertex rule --- is the only part of the
proof that crucially relies on the fact that we are dealing with only three terminal pairs.
In fact, it is inspired by the reduction for four terminal pairs~\cite{PilipczukW18a}
and our study why this reduction fails for three terminal pairs.

The irrelevant-vertex rule also requires us to look into the details of flow-augmentation
(\cref{thm:dfl-basic}) and extract some extra properties of this tool. 
In short, we need to capture the following intuition: 
A sequence of deletable edges along the same flow path $P \in \mathcal{P}$ in the maximum
flow $\mathcal{P}$ in $G+A$ is in some sense sequentially positioned in the graph,
so one can usually reach any later edge from an earlier one. 
This is not strictly true as stated above, but we prove a variant of this statement
in \cref{sec:soybean}. 

One can ask if the proof of the irrelevant-vertex rule 
crucially needs the assumption of the solution being shadowless. 
In particular, the usage of shadow removal makes our algorithm inherently unweighted
(the shadow-removal step involves a greedy argument that completely breaks down
 in the presence of weights). 
We complement our main result by proving (\cref{sec:hard})
  that \textsc{Weighted Directed Multicut}
is W[1]-hard even with two terminal pairs, so the shadow-removal step seems
necessary.
\begin{theorem}\label{thm:hard}
\textsc{Weighted Directed Multicut}, parameterized by the cardinality of the cutset,
  is W[1]-hard even with two terminal pairs.
\end{theorem}
We remark that the one-terminal-pair case
of \textsc{Weighted Directed Multicut}, or \textsc{Bi-objective $st$-cut},
is proved to be FPT in~\cite{dfl-arxiv} as one of the basic exemplary usages
of flow-augmentation.

Since we rely on the whole meta-algorithmic toolbox of twin-width of~\cite{BonnetKTW22}, 
we cannot state an explicit dependency on the parameter in the running time bound of our algorithm.
Relatedly, we would like to remark that the concept of encoding an instance at hand into a CSP instance with a number of variables that is bounded by a function of the parameter, but over large ordered domains,
appeared also recently in the FPT algorithm for \textsc{Optimal Discretization}~\cite{KratschMMPS21}. 
The encoding there also uses an unbounded number of constraints of the form $(x \leq a) \vee (y \geq b)$ and a bounded-in-parameter number of permutation constraints.
However, the main effort in the proof in~\cite{KratschMMPS21} lies in showing that the used permutation constraints have very simple structure (they are called in~\cite{KratschMMPS21} \emph{segment reversions});
in particular, one can observe that their permutation matrices do not contain a grid minor of size $3$. 
As a result, in~\cite{KratschMMPS21} the authors are able to design an explicit FPT algorithm for the obtained CSP instance (with an explicit single-exponential running time bound).
Although using twin-width meta-algorithms prevents us from stating an explicit running time bound, it allows to claim fixed-parameter tractability of a much wider range of CSP instances:
Permutations of bounded grid minor in their permutation matrices is a much wider class than the aforementioned segment reversions. 
We believe the presented framework of casting a problem into a small number of variables with unbounded ordered domains, bound by permutation constraints, and using twin-width toolbox to solve it,
has a wider future potential in parameterized complexity.

\paragraph{Organization.}
After brief preliminaries, where in particular we state the shadow-removal tool and the extended version of flow-augmentation, we proceed to the main proof (of \cref{thm:dmcthree-fpt}).
\Cref{sec:low-tw-csp-algo} introduces the twin-width toolbox and shows fixed-parameter tractability of CSP instances with permutations only containing bounded grid minors.
\Cref{sec:csp-formulation} contains the main proof, deferring the proof
of the irrelevant-vertex rule to~\cref{sec:tw-red}. 
\Cref{sec:shadow-removal} contains a (standard) proof of the used shadow-removal
statement, while \cref{sec:soybean} contains a proof of the used extension
of flow-augmentation. 
Finally, the proof of \cref{thm:hard} can be found in \cref{sec:hard}.

\ifthenelse{\equal{\frontpageformat}{submission}}{}{%
  \paragraph{Acknowledgements.}
  The research leading to the results presented in this paper was partially carried out during the Parameterized Algorithms Retreat of the University of Warsaw, PARUW 2022, held in B\k{e}dlewo in April 2022.
  We acknowledge insightful discussions with the twin-width experts at B\k{e}dlewo: \'{E}douard Bonnet, Jakub Gajarsk\'{y}, and Micha\l{} Pilipczuk, as well as later discussions with Szymon Toru\'{n}czyk.}

\section{Preliminaries}\label{sec:prelims}

\looseness=-1
Let $G$ be a directed graph.
We use edge and arc interchangably for the directed edges of $G$.
For two vertices $u,v \in V(G)$, we say that \emph{$u$ reaches $v$} or \emph{$v$ is reachable from $u$} if there exists a directed path from $u$ to $v$ in $G$.
A path starting in a vertex $u$ and ending in a vertex $v$ is also called a $uv$-path.
Let $e$ be a (directed) edge.
Then $s(e)$ ($t(e)$) is start (target) of~$e$.
Analogously, we use $s(P)$ (resp.\ $t(P)$) for first (resp.\ last) vertex of a path~$P$.
For $s,t \in V(G)$, a set $S \subseteq E(G)$ is an \emph{$st$-cut} if $t$ is not reachable from $s$ in $G-S$ and, similarly,
    a set $S \subseteq V(G) \setminus \{s,t\}$ is an \emph{$st$-separator} if $t$ is not reachable from $s$ in $G-S$.
In the latter, by $G-S$ we denote the subgraph of $G$ induced by $V(G) \setminus S$. 
Throughout the paper we use $[k]$ as shorthand for the integer set $\{1, 2, \dots,k\}$.

\looseness=-1
In this work, an instance of \textsc{Directed Multicut} (\dmcthree{} for short) is a tuple $(G,k,\allowbreak (s_i,t_i)_{i \in [3]}, \allowbreak V^{\infty})$ consisting of a directed graph $G$, six distinguished vertices $s_1,s_2,s_3,t_1,t_2,t_3 \in V(G)$, called \emph{terminals}, an integer $k$, and a vertex subset $V^{\infty} \subseteq V(G)$, called \emph{undeletable vertices}.
We sometimes denote the set of terminals as $\term = \{s_1,s_2,s_3,t_1,t_2,t_3\}$. We require $\term \subseteq V^{\infty}$. 
A \emph{solution} is a set $S$ of non-terminal vertices of $G$ such that $S \cap V^{\infty} = \emptyset$ and for every $i=1,2,3$ the vertex $t_i$ is not reachable from the vertex $s_i$ in $G-S$. 
\textsc{Directed Multicut} asks for a solution of cardinality at most $k$. 

\looseness=-1
For an instance $(G,k,(s_i,t_i)_{i=1,2,3}, V^{\infty})$ and a non-terminal vertex $v$, by \emph{bypassing $v$} we mean the following operation.
First, for every edge $(u,v) \in E(G)$ and every edge $(v,w) \in E(G)$, we add an edge $(u,w)$ if it is not already present in the graph.
Finally, we delete the vertex $v$.
For a set~$X$ of nonterminal vertices, by bypassing $X$ we mean bypassing vertices of $X$ in arbitrary order; note that the result does not depend on the order.
The following lemma is immediate and shows that bypassing is a good way to reduce vertices that are provably not in the sought solution.

\begin{lemma}[Chitnis et al.~\cite{dsfvs}, Lemma~$3.11$]\label{lem:torso}
Let $G$ be a directed graph and let $C \subseteq V(G)$. Let $G'$ be obtained from $G$ after bypassing $C$ and let $S \subseteq V(G) \setminus C$. For any $a,b \in V(G) \setminus (C \cup S)$, $G - S$ has an $ab$-path if and only if $G' -S$ has an $ab$-path.
\end{lemma}

\paragraph{Shadows and shadow removal.}
Given an instance $(G,k,(s_i,t_i)_{i \in [3]}, V^{\infty})$ of \dmcthree\ and a set $X \subseteq V(G) \setminus V^{\infty}$,
a vertex $v \in V(G) \setminus (X \cup V^{\infty})$ is in the \emph{forward shadow} of $X$
 if $v$ is reachable from neither $s_1$, $s_2$, nor $s_3$ in $G-X$. 
Symmetrically, $v$ is in the \emph{reverse shadow} of $X$ if
neither $t_1$, $t_2$, nor $t_3$ is reachable from $v$ in $G-X$.
The set of vertices in the forward shadow of $X$ in~$G$, is denoted by $\ff_G(X)$ and the set of vertices in the reverse shadow of $X$ in $G$, is denoted by $\rr_G(X)$.
The vertex $v$ is in the \emph{shadow} of $X$ if it is in the forward shadow or the reverse shadow of $X$. 
A set $X$ is \emph{shadowless} if no vertex is in its shadow, that is, if $\rr_G(X) \cup \ff_G(X) =\emptyset$.

The following statement encapsulates the shadow removal technique in the context of \textsc{Directed Multicut}.
Though it follows directly from~\cite{dir-mwc,dsfvs}, we provide a formal proof in \cref{sec:shadow-removal} for the sake of completeness.

\begin{restatable}{theorem}{ShadowRemovalTheorem}
	\label{thm:shadow-removal}
	Given an instance
	$(G,k,(s_i,t_i)_{i \in [3]}, V^{\infty})$ of \dmcthree,
	there is an algorithm that runs in time
	$2^{\Oh(k^2 \log k)} n^{\Oh(1)}$, and outputs a family $\mathcal{Z}$ of subsets of $V(G) \setminus V^{\infty}$ such that $|\mathcal{Z}| =2^{\Oh(k^2 \log k)} \log^2 n$ and,
	if the input instance is a \yes-instance, then there exists $Z \in \mathcal{Z}$
	such that $(G',k,(s_i,t_i)_{i\in [3]}, V^{\infty} \setminus Z)$ is a \yes-instance that admits a shadowless solution of cardinality at most $k$, where $G'$ is the result of bypassing $Z$ in $G$.
\end{restatable}

\subsection{Twin-width}

Next, we state the definition of twin-width as introduced by Bonnet~et~al.~\cite{BonnetKTW22}.
They make use of the concept of a \emph{trigraph} $G$, which consists of a vertex set $V(G)$, and two disjoint edge sets, one containing \emph{black} edges $E(G)$ and the other containing \emph{red} edges $R(G)$.
In particular, every graph $G$ is a trigraph with only black edges and $R(G)$ being empty.
Let $G$ be a trigraph.
We say that we \emph{contract} two vertices $u,v\in V(G)$ if we merge them into a single vertex $w$, and then possibly color the edges incident to the new vertex $w$.
Every existing edge $wz$ remains black if and only if $uz$ and $vz$ were previously black edges.
All other edges incident to $w$ are colored in red.
A \emph{contraction sequence} of an $n$-vertex (tri)graph $G$ is a sequence of trigraphs $G = G_n, \dots , G_1 = K_1$ such that $G_i$ is obtained from $G_{i+1}$ by contracting two vertices.
A contraction sequence is called a \emph{$d$-sequence} if all trigraphs in it have red degree at most $d$.
The \emph{twin-width} of $G$, denoted by $\tww(G)$, is the minimum integer $d$ such that $G$ admits a $d$-sequence.

Now we turn to matrices.
Here, our main proxy to the twin-width are so-called rank-$k$ divisions and the grid rank of the matrix which are closely related to twin-width graph parameter, see Bonnet et al.~\cite{BonnetGMSTT22}.
Let $M$ be a 0-1 matrix.
A \emph{division} $\calD$ of $M$ is a pair $(\calD^R, \calD^C)$,
where $\calD^R$ and $\calD^C$ are partitions of the rows and columns 
into intervals of consecutive rows and intervals of consecutive columns, respectively.
A \emph{$k$-division} is a division with $\card{\calD^R} = \card{\calD^C} = k$.
If $(\calD^R, \calD^C)$ is a division, $\calD^R = (R_1, R_2, \ldots)$ and $\calD^C = (C_1, C_2, \ldots)$,
then for each pair $R_i$, $C_j$, the (contiguous) submatrix of $R_i \cap C_j$ is called the \emph{$(i, j)$-cell} of $\calD$.
A \emph{rank-$k$ division} of $M$ is a $k$-division $\calD$ of $M$ such that each cell of $\calD$ 
contains at least $k$ distinct rows and at least $k$ distinct columns,
or, equivalently, has combinatorial rank at least $k$.
The maximum integer $k$ such that $M$ admits a rank-$k$ division is called the \emph{grid rank} of $M$, 
denoted by $\gridrank(M)$.
A \emph{$k$-grid minor} of $M$ is a $k$-division $\calD$ of $M$ such that each cell of $\calD$
contains at least one $1$; note that any rank-$k$ division of $M$ for $k \geq 2$ is 
necessarily a $k$-grid minor, too.

Let $G$ be a graph and $\prec$ be a linear order on $V(G)$.
We denote by $\Adj_\prec(G)$ the adjacency matrix of $G$ where rows and columns are ordered according to $\prec$.
We use the following.

\begin{theorem}[Bonnet et al.~\cite{BonnetGMSTT22}]\label{thm:gr:tww}
	There is a computable function $f \colon \bN \to \bN$ such that the following hold.
	Let $G$ be a graph.
	\begin{itemize}
		\item For any total order $\prec$ of $V(G)$, if $\gridrank(\Adj_\prec(G)) \le k$, then $\tww(G) \le f(k)$.
		\item If $\tww(G) \le k$, then there is a total order $\prec$ of $V(G)$ such that $\gridrank(\Adj_\prec(G)) \le f(k)$.
	\end{itemize}
\end{theorem}

\begin{lemma}[Bonnet et al.~\cite{BonnetCKKLT22}]\label{lem:gridrankbound}
	There is a computable function $f \colon \bN \to \bN$ such that the following holds.
	Let $V$ be a set of vertices, let $\prec$ be a linear order on $V$, 
	and let $G_1$ and $G_2$ be two graphs on the vertex set $V$.
	If $\gridrank(\Adj_\prec(G_i)) \le k$ for all $i \in [2]$, 
	then $\gridrank(\Adj_\prec(G_1 \cup G_2)) \le f(k)$.
\end{lemma}

The following theorem can be deduced from the arguments of~\cite{BonnetKTW22},
but is not stated there explicitly; for completeness, we
provide a proof in \cref{app:tw:to:gr}.
\begin{restatable}{theorem}{thmgminorgrank}\label{thm:grid-minor-apx}
	There is a computable function $f$ with $f(k) = 2^{\Oh(k \log k)}$ such that the following holds.
	There is an algorithm that, given a 
	0-1 matrix $\bfA$, in $f(k) n^{\Oh(1)}$ time either 
	\begin{itemize}
		\item finds a $k$-grid minor in $\bfA$, or
		\item certifies that $\gridrank(\bfA) \le f(k)$.
	\end{itemize}
\end{restatable}

\subsection{(Permutation) CSP}\label{ss:permCSP}
\looseness=-1
An instance of a \emph{constraint satisfaction problem} (CSP) 
is a triple $(X, \calD, \calC)$, where $X = \{x_1, \ldots, x_n\}$ is a set of \emph{variables},
$\calD = \{D_1, \ldots, D_n\}$ a set of \emph{domains},
and each $C \in \calC$ a \emph{constraint}.
A constraint $C$ is an $a_C$-tuple $(x_{a_1}, \ldots, x_{a_C})$ of variables
and a relation $R(C) \subseteq D_{a_1} \times \ldots \times D_{a_C}$.
A \emph{valuation} $\alpha$ assigns to each $x_i$ a value $\alpha(x_i) \in D_i$.
A constraint $C \in \calC$ is \emph{satisfied} by $\alpha$ if 
$(\alpha(x_{a_1}), \ldots, \alpha(x_{a_C})) \in R(C)$.
Valuation $\alpha$ \emph{satisfies} $(X, \calD, \calC)$ if it satisfies all constraints.

Let $(D_1, \le_1)$ and $(D_2, \le_2)$ be finite totally ordered sets.
We need two types of relations in $D_1 \times D_2$.
A relation $R \subseteq D_1 \times D_2$ is called \emph{\downclosed} 
if for every $(x_1, x_2) \in R$ and every $x_1' \le_1 x_1$ and $x_2' \le_2 x_2$ it holds that $(x_1', x_2') \in R$.
Let $X_1 \subseteq D_1$ and $X_2 \subseteq D_2$ with $\card{X_1} = \card{X_2}$
and let $\pi \colon X_1 \to X_2$ be a bijection.
We refer to the relation $\{(x,\pi(x))~|~x \in X_1\}$ as the \emph{permutation constraint $\pi$}.
We denote by $\Adj(\pi)$ the 0-1 matrix of dimension $|D_1| \times |D_2|$ associated with $\pi$ as follows.
The rows and columns one-to-one correspond to $D_1$ and $D_2$, respectively, following the orders $\le_1$ and $\le_2$.
The entry associated with $(x,y) \in D_1 \times D_2$ equals $1$ if and only if $x \in D_1$ and $y = \pi(x)$.
We deal with CSPs of the following form.

\begin{definition}\label{def:permCSP}
	A \twwCSP{$w$}\ instance consists of variables $x_1, \ldots, x_k$ with domains $(D_1, \le_1)$, $\ldots$, $(D_k, \le_k)$,
	where for all $i \in [k]$, $(D_i, \le_i)$ is a totally ordered set,
	and each constraint is, for some $i \neq j$, either
	\begin{itemize}
		\item a \downclosed relation $R \subseteq D_i \times D_j$, or 
		\item a permutation constraint $\pi \colon X_i \to X_j$ where $X_i \subseteq D_i$ and $X_j \subseteq D_j$.
	\end{itemize}
	Furthermore, for each permutation constraint $\pi$, we have that $\gridrank(\Adj(\pi)) \le w$.
\end{definition}

\subsection{Flow-augmentation}\label{ss:dfl}

Because the shadow removal technique is much easier to phrase and use in the vertex-deletion
regime, we need to adjust flow-augmentation~\cite{dfl-arxiv} from edge deletions to vertex deletions. This is pretty straightforward via 
the standard reductions between edge- and vertex-deletion regimes in directed graphs. 
More importantly, we need to squeeze an extra connectivity property out of
flow-augmentation, which we formalize below as \emph{soybeans}. 
We present here only the necessary definitions and the main statement that is used
in the algorithm; its proof is deferred to \cref{sec:soybean}.

\looseness=-1
Let $G$ be a directed graph with two distinguished vertices $s,t \in V(G)$.
The vertices of $G$ are partitioned into \emph{deletable} vertices and \emph{undeletable}
vertices; $s$ and $t$ are undeletable. 
A (vertex-based) \emph{$st$-flow} is a collection of $st$-paths that do not share a deletable vertex;
the number of paths is the \emph{value} of the flow.
We use $\lambda_G(s,t)$ for the maximum possible value of an $st$-flow;
a flow of value $\lambda_G(s,t)$ is an \emph{$st$-maxflow}.
As a convention we say that, if $G$ contains an $st$-path consisting of undeletable vertices
only, then the flow containing such a path has value $+\infty$ and $\lambda_G(s,t) = +\infty$.

A set $Z$ of deletable vertices is an \emph{$st$-separation} 
if there is no path from $s$ to $t$ in $G-Z$. 
An $st$-separation $Z$ is \emph{minimal} if no proper subset of $Z$ is an $st$-separation and
\emph{minimum} (or \emph{$st$-mincut}) if it has minimum possible cardinality.
By Menger's theorem, if $\lambda_G(s,t) < +\infty$ then the size of every $st$-mincut is exactly $\lambda_G(s,t)$ and there are no $st$-separations if $\lambda_G(s,t) = +\infty$. 
We drop the subscript if the graph $G$ is clear from the context. 

We say that a set of arcs $A \subseteq V(G) \times V(G)$ is \emph{compatible} with a minimal $st$-separation $Z$ if the following holds:
for every $v \in V(G)$, there is a path from $s$ to $v$ in $G-Z$ if and only if there is a path from $s$ to $v$ in $(G+A)-Z$. 
The pair $(A, \flow)$ is \emph{compatible} with $Z$ if $A$ is compatible with $Z$
and $\flow$ is an $st$-maxflow in $G+A$.

A \emph{soybean} in $G$ is an unordered pair of walks that have the same starting vertex and the same ending vertex. We do not require the walks to be disjoint in any sense; in particular, a pair consisting of the same walk twice is always a soybean. 
Two soybeans $PQ$ and $P'Q'$ are \emph{vertex-disjoint} if $(V(P) \cup V(Q)) \cap (V(P') \cup V(Q')) = \emptyset$.
For two sets of vertices or edges $C,D \subseteq V(G) \cup E(G)$, a soybean $PQ$ is a \emph{$CD$-soybean} if one walk of $PQ$ contains an edge or a vertex of $C$ and the other walk of $PQ$ contains an edge or a vertex of $D$. 
For a path $P$ and two disjoint sets $C,D \subseteq V(P) \cup E(P)$, we say that $C$ and $D$ are \emph{interlaced on $P$} if $|C|=|D|=q$ for some integer $q$, $C$ can be enumerated as $c_1,\ldots,c_q$, $D$ can be enumerated
as $d_1,\ldots,d_q$, and the order of these vertices and edges along $P$ is $c_1,d_1,c_2,d_2,\ldots,c_q,d_q$. 

\begin{restatable}{theorem}{VertexSoybeanTheorem}\label{thm:dfl-upgrade-vertices}
  There exist computable functions $c : \N \to \N$ and $q : \N \times \N \to \N$ such that the following holds.

  There exists a polynomial-time randomized algorithm that, given a directed graph $G$
  (with possibly some vertices marked as undeletable),
  vertices $s,t \in V(G)$, and an integer $k$,
  returns 
  an arc set $A \subseteq V(G) \times V(G)$
  and an $st$-maxflow $\witnessflow$ in $G + A$ such that 
  for every minimal $st$-separator $Z$ of size at most $k$,
  with probability $2^{-\Oh(k^4 \log k)}$, the tuple $(A,\witnessflow)$ is compatible with $Z$.
  
  Additionally, the algorithm returns a partition $\mathcal{B}$ of the deletable
  vertices of $\bigcup_{P \in \witnessflow} V(P)$ into at most $c(k)$ sets
  such that for every $P \in \witnessflow$, every integer $p \in \N$, every $B \in \mathcal{B}$
  and every two disjoint sets $C,D$ of size at least $q(k,p)$,
  consisting of vertices of $B \cap V(P)$ that
  are interlaced on $P$, the graph $G$ contains a family of $p$ pairwise vertex-disjoint $CD$-soybeans.
  
  Finally, one can take $c$ and $q$ such that 
  $c(k) = 2^{\Oh(k^3 \log k)}$ and $q(k,p) = 2^{\Oh(k^3 \log (kp))}$.
\end{restatable}

\section{\permcsp\ with bounded twin-width}\label{sec:low-tw-csp-algo}
In this section we show the following.
\begin{theorem}\label{thm:bounded-twin-width}
\twwCSP{$w$}\ parameterized by the number of constraints plus $w$ is fixed-parameter tractable. 
\end{theorem}
\begin{proof}
  \looseness=-1
	We transform the given \twwCSP{$w$} instance $\CSPinst$ into an ordered, vertex- and edge-colored
	graph $G$
	whose twin-width only depends on $w$ and the number of constraints in $\CSPinst$.
	We then give an \FO-formula $\phi$ such that $G \models \phi$ if and only if $\CSPinst$ is satisfiable.
	We can then apply the \FO\ model checking algorithm on (ordered) graphs 
	that runs in \FPT\ time when parameterized by the twin-width of the input graph plus the length of the formula~\cite{BonnetGMSTT22,BonnetKTW22}.

	\looseness=-1
	We denote the variables of $\CSPinst$ by $x_1, \ldots, x_k$, by $\calR$ the set of \downclosed constraints of~$\CSPinst$, 
	and by $\Pi$ the set of permutation constraints of $\CSPinst$.
	For each $i \in [k]$, let $d^i_1, \ldots, d^i_{n_i}$ denote the elements of $D_i$, ordered according to $\le_i$.
	For each $i \in [k]$, we introduce into $G$ a set $V_i$ of $n_i$ vertices, colored with color $i$ and denoted (and ordered) as $v^i_1, \ldots, v^i_{n_i}$.
	We reuse the symbols ``$\le_i$'' to denote the ordering of the vertices in $V_i$
	in the \FO-formula.
	At the heart of our \FO-formula $\phi$ is an existential guess of one vertex per $V_i$, 
	and choosing $v^i_j$ for $j \in [n_i]$ corresponds to setting the variable $x_i$ to $d^i_j$.
	The vertex colors are $[k]$, and encode membership in the $V_i$'s.
	We assume we have predicates $\col(y) = i$ for a vertex variable $y$ and $i \in [k]$ 
	that verify whether the vertex assigned to $y$ has color $i$ (which in turn means that $y \in V_i$).
	The existential guess is:
	\begin{align*}
		\phi_\exists \equiv \exists y_1 \ldots \exists y_k \bigwedge\nolimits_{i \in [k]} \col(y_i) = i
	\end{align*}

	\myparagraph{Encoding the constraints in \boldmath$(G, \phi)$.}
	Next we add edges to the graph $G$ and color them using the constraints $\calR \cup \Pi$ of $\CSPinst$ as colors,
	thus indicating why an edge was added to $G$.
	Similarly to above, we assume that we have predicates $\col(e) = C$ where $C \in \calR \cup \Pi$,
	that verify whether the edge assigned to $e$ has color $C$.
	Let $R \in \calR$ be a \downclosed constraint with $R \subseteq D_i \times D_j$.
	For each $(d^i_a, d^j_b) \in R$ such that there is no $(a', b') \in [n_i] \times [n_j] \setminus \{(a, b)\}$ with $a' \ge a$, $b' \ge b$, 
	and $(d^i_{a'}, d^j_{b'}) \in R$, 
	we add the edge $v^i_a v^j_b$ colored $R$ to $G$.
	Note that the edges with color $R$ form a matching,
	and that they mark the ``boundary'' of the relation $R$,
	in the following sense.
	For each $(d^i_a, d^j_b) \in D_i \times D_j$,
	we have that $(d^i_a, d^j_b) \in R$ if and only if there is an edge $(v^i_{a'}, v^j_{b'})$ 
	with $a' \ge_i a$, $b' \ge_j b$, and of color $R$ in $G$.
	We construct the following part of $\phi$ which checks this condition, and 
	therefore is true if and only if $R$ is satisfied 
	under the value assignment to $x_1, \ldots, x_k$ corresponding to the choice of the vertices $y_1, \ldots, y_k$.
	\begin{align*}
		\phi_R \equiv \exists e\exists z_i \exists z_j 
			\left(\col(e) = R \land \bigwedge\nolimits_{h \in \{i, j\}} \inc(e, z_h) \land \col(z_h) = h \land y_h \le_h z_h\right)
	\end{align*}
	Now, let $\pi \in \Pi$ with $\pi \colon X_i \to X_j$ for some $X_i \subseteq D_i$ and $X_j \subseteq D_j$ be a permutation constraint.
	Then, for each $d^i_a \in X_i$, we let $d^j_b = \pi(d^i_a)$, and add the edge $v^i_av^j_b$ colored
  $\pi$ to $G$.
	Again, the edges with color $\pi$ form a matching.
	This finishes the construction of $G$, and the last building block of $\phi$ is as follows.
	It checks that whenever $y_i$ (resp.\ $y_j$) is incident with some edge colored~$\pi$, 
	that the other endpoint of that edge is chosen to be $y_j$ (resp.\ $y_i$).
	This part of the formula evaluates to true if and only if the permutation constraint $\pi$ is satisfied 
	under the corresponding choice of values for the variables of $\CSPinst$.
	\begin{align*}
		\phi_\pi \equiv \exists e \left(\col(e) = \pi \land \inc(e, y_i) \land \inc(e, y_j)\right).
	\end{align*}
	
	We now obtain $\phi$ as 
		$$\phi \equiv \phi_\exists \land \bigwedge\nolimits_{R \in \calR} \phi_R \land \bigwedge\nolimits_{\pi \in \Pi} \phi_\pi.$$
	
	The correctness of this transformation immediately follows from the description above:
	\begin{claim}\label{claim:csp:correct}
		$\CSPinst$ is satisfiable if and only if $G \models \phi$.
	\end{claim}
	
	It remains to show that the twin-width of $G$ is bounded by a function of $w$ and the number of constraints in $\CSPinst$.
	To do that, consider the order $V_1 < V_2 < \ldots < V_k$ of $V(G)$,
	where each $V_i$ is ordered according to $\le_i$; 
	denote this order by $\prec$.
	We want to show that $\gridrank(\Adj_\prec(G)) \le h(w, \card{\calR \cup \Pi})$, 
	for some computable function $h$,
	which implies a desired bound on the twin-width of $G$ by \cref{thm:gr:tww}.
	This can be done by repeated application of \cref{lem:gridrankbound}
	over all edge colors of~$G$.
	All permutation constraints have grid rank at most $w$ by assumption, 
	so it remains to show that the grid rank of \downclosed constraints is bounded as well.
	Before we do so, we observe one minor technical detail.
	\begin{claim}\label{claim:embed}
		Let $(V, \prec)$ be an ordered set of vertices and let $X_i$ and $X_j$ be disjoint consecutive subsets of $V$.
		Let $G$ be a graph on the vertex set $V$ that only has edges between $X_i$ and $X_j$.
		For each $k \ge 1$,
		if $\gridrank(\Adj_\prec(G)) \ge k+2$, then $\gridrank(\Adj_\prec(G)[X_i, X_j]) \ge k$.
	\end{claim}
	\begin{proof}
		Let $M = \Adj_\prec(G)$, $M_{i, j} = M[X_i, X_j]$, and $M_{j, i} = M[X_j, X_i]$.
		Let $\calD = (\calD^R, \calD^C)$ be a rank-$(k+2)$ division of $M$.
		Since outside of $M_{i, j}$ and $M_{j, i}$, $M$ is all-$0$,
		for either $M_{i, j}$ or $M_{j, i}$, 
		we may assume (up to renaming $i$ and $j$) that 
		every cell of $\calD$ intersects $M_{i, j}$.
		This implies that $k^2$ cells of $\calD$ are entirely contained in $M_{i, j}$.
		Since each such cell has combinatorial rank at least $k+2$, 
		this gives a rank-$k$ division of $M_{i, j}$.
	\end{proof}
	\begin{claim}\label{claim:gr:downclosed}
		Let $R \in \calR$ be a \downclosed constraint,
		let $E_R \subseteq E(G)$ be the set of edges colored $R$ in $G$,
		and let $G_R = (V(G), E_R)$.
		Then, $\gridrank(\Adj_\prec(G_R)) \le 3$.
	\end{claim}
	\begin{proof}
		Let $i, j \in [k]$ be such that $R \subseteq D_i \times D_j$,
		and let $M_{i,j} = \Adj_\prec(G_R)[V_i, V_j]$.
		Suppose for a contradiction that there is a rank-$2$ division of $M_{i,j}$
		with row intervals $(R_1, R_2)$ and column intervals $(C_1, C_2)$.
		Since each cell of this division has combinatorial rank at least $2$,
		we know that each such cell contains at least one $1$.
		Let $e_1 = u_1v_1$ be the edge corresponding to the $1$ in $R_1 \cap C_1$,
		and let $e_2 = u_2 v_2$ be the edge corresponding to the $1$ in $R_2 \cap C_2$,
		where $u_1, u_2 \in V_i$ and $v_1, v_2 \in V_j$.
		But then, $u_1 <_i u_2$ and $v_1 <_j v_2$, 
		which means that the above construction would not have added the edge $e_1$ with color $R$.
		This implies that $\gridrank(M_{i, j}) \le 1$, 
		so by \cref{claim:embed}, $\gridrank(\Adj_\prec(G_R)) \le 3$.
	\end{proof}
	
	\begin{claim}\label{claim:tww}
		There is a computable function $h \colon \bN \times \bN \to \bN$,
		such that $\tww(G) \le h(w, \card{\calR \cup \Pi})$.
	\end{claim}
	\begin{proof}
		For each $C \in \calR \cup \Pi$, let $E_C \subseteq E(G)$ denote the set of edges colored $C$,
		and let $G_C = (V(G), E_C)$.
		Then, $\calE = \{E_C \mid C \in \calR \cup \Pi\}$ is a partition of $E(G)$.
		Recall that $\prec$ is a linear order on $V(G)$ which lets $V_1 < V_2 < \ldots < V_k$,
		and for each $i \in [k]$, coincides with $\le_i$ on~$V_i$.
		By the assumption of the theorem, and \cref{claim:gr:downclosed},
		we have that $\gridrank(\Adj_\prec(G_C)) \le \max\{3, w\}$ for all~$C$.
		Since $\calE$ is a partition of the edge set of $G$,
		we can repeatedly apply \cref{lem:gridrankbound}
		to conclude that $\gridrank(\Adj_\prec(G)) \le h'(w, \card{\calR \cup \Pi})$ for some computable function $h'$.
		\cref{thm:gr:tww} in turn implies that $\tww(G) \le h(w, \card{\calR \cup \Pi})$ for some computable function $h$.
	\end{proof}
	
	We can now run the \FO\ model checking algorithm~\cite{BonnetGMSTT22,BonnetKTW22} 
	on $(G, \phi)$ and return the same answer.
	Note that by~\cite{BonnetGMSTT22}, 
	we can compute a contraction sequence of $G$ whose twin-width 
	is bounded by a computable function of $\tww(G)$
	in FPT time parameterized by $\tww(G)$,
	using the ordering $\prec$.
	Correctness follows from \cref{claim:csp:correct}, 
	and by \cref{claim:tww} the twin-width of $G$ only depends on $w$ and the number of constraints $s$ in $\CSPinst$.
	It is clear that the length of $\phi$ can be upper bounded by a function of $k$ and $s$, 
	and since we can assume that $k \le s$, the length of $\phi$ can be upper bounded by a function of $s$ alone.
	The algorithm of~\cite{BonnetGMSTT22,BonnetKTW22} is fixed-parameter tractable in $\tww(G) + \card{\phi}$,
	which in our application translates to an \FPT-algorithm for \twwCSP{$w$} parameterized by $w + s$, as desired.
\end{proof}

We would like to remark that \cref{thm:bounded-twin-width} generalizes the fixed-parameter tractability result of 
another type of CSP, called \textsc{Forest CSP}~\cite{KratschMMPS21}.
Moreover, if we drop the twin-width of the permutation constraints as part of the parameter, then the resulting
\permcsp\ problem parameterized by the number of constraints 
is \W[1]-hard (see \cref{app:edge-choice}).

\section{Three-terminal-pair \textsc{Directed Multicut} is fixed-parameter tractable}\label{sec:csp-formulation}

This section is dedicated to proving the main theorem of this paper:
\DMCthreeFPTTheorem*

To prove \cref{thm:dmcthree-fpt}, we show how to reduce \dmcthree\ to \permcsp\ such that the twin-width of each constraint is bounded by some function of the desired separator size~$k$.
One main ingredient in this reduction is the flow-augmentation technique explained in \cref{ss:dfl}.
The outline is as follows.
We first perform flow-augmentation for each of the terminal pairs, giving us an augmented graph and a flow of value $k$ for each terminal pair.
If there is a solution~$S$, then with large-enough probability $S$ is preserved as a separator after the augmentation steps.
The solution thus corresponds to a selection of vertices, one for each of the obtained flow paths.
To obtain a reduction to \permcsp, the idea is then to introduce one variable $x$ for each flow path $P$ where the domain of $x$ is the set of vertices on~$P$.
A set of straightforward constraints ensures that the vertices selected by the variables form a separator for each terminal pair.
One crux with this approach is how to ensure that the variables introduced for flow paths between different terminal pairs select vertices in a consistent way.
We note that, after trying all possibilities of the possible overlaps of selected vertices, the consistency requirement can be modelled as a permutation constraint, we also call these \emph{consistency constraints}.
In this way, we obtain an instance of \permcsp\ with $\Oh(k^{2})$ constraints.
However, as mentioned in \cref{sec:low-tw-csp-algo}, \permcsp\ in general is \W[1]-hard with respect to the number of constraints.
Thus, we need more work to obtain \cref{thm:dmcthree-fpt}. 
We show how to ensure that the obtained constraints are simple, in the sense that they have low twin-width, crucially leveraging our improved version of flow-augmentation from \cref{ss:dfl}.
We then apply the algorithm from \cref{sec:low-tw-csp-algo} for solving \permcsp\ instances of low twin-width.

For use in the remainder of the section, fix an instance $(G, (s_i, t_i)_{i \in [3]}, k)$ of \dmcthree.
We show how to solve this instance in fixed-parameter time with large-enough probability.
The algorithm is partitioned into four main steps; they are reflected in the structure of the remainder of this section, also see \cref{fig:flowchart} for an illustration.
The first step is the shadow-removal technique, which ensures in fixed-parameter time that each vertex is reachable from some terminal and reaches some terminal.
We crucially use this property when bounding the twin-width later on.
The second step is using flow-augmentation and reducing to \permcsp.
The third step is to reduce the twin-width of the constraints in the \permcsp\ instance.
Finally, we solve the \permcsp\ instance using the algorithm from \cref{sec:low-tw-csp-algo}.

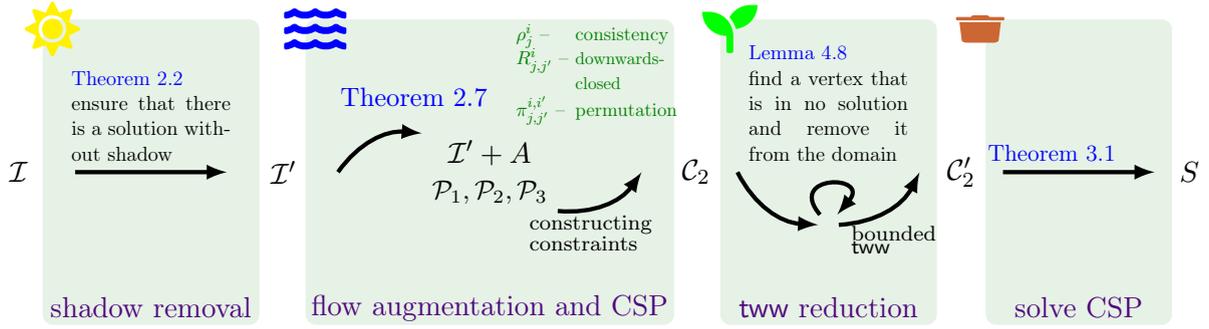
\begin{figure}[t]

\begin{tikzpicture}[scale=0.97]

	\colorlet{myGreen}{green!50!black}
	\definecolor{myRed}{rgb}{0.68, 0.05, 0.0}
	\colorlet{myBlue}{blue!90!black}
	\definecolor{myLightblue}{rgb}{0.54, 0.81, 0.94}
	\colorlet{myViolet}{myBlue!55!myRed}
	\colorlet{myOrange}{yellow!55!myRed}
	
	\tikzstyle{vertex} = [draw, circle, inner sep=.6mm, thick,fill=black]
	\tikzstyle{edge} = [draw=black,line width=1.3pt]
	\tikzstyle{nonedge} = [draw,thick,dashed]
	\tikzstyle{directededge} = [edge,->,-latex]
	\tikzstyle{path} = [decorate,decoration={snake,post length=1mm,pre length=1mm,segment length=17},->,-latex]
	
	\node (input) at (0,0) {$\mathcal{I}$};
	
	\node[fill=myGreen!10!white,rounded corners] (shadow-removal) at ($(input)+(1.8,0)$) {
		\begin{tikzpicture}
                  \useasboundingbox (-1.3,-1.9) rectangle (1.3,1.9);
                  \node[text=yellow] (sun) at (-1.3,1.9) {\huge\faSun};
			\node (sr-center) at (0,0) {};
			\node (sr-label) at ($(sr-center)+(0,-1.8)$) {\textcolor{myViolet}{shadow removal}};
			\draw[directededge,ultra thick] ($(sr-center)+(-1,0)$) -- ($(sr-center)+(1,0)$)
			node[midway,above,align=center] (sr-description) {
				\scalebox{0.7}{\begin{minipage}{3cm}
						\cref{thm:shadow-removal}\\
						ensure that there is a solution without shadow
					\end{minipage}}};
			
	\end{tikzpicture}};

	\node (noshadow-instance) at ($(shadow-removal)+(1.8,0)$) {$\mathcal{I}'$};

	\node[fill=myGreen!10!white,rounded corners] (flow-aug-CSP) at ($(noshadow-instance)+(2.8,0)$) {
		\begin{tikzpicture}
			 
                  \useasboundingbox (-2.3,-1.9) rectangle (2.3,1.9);
                  \node[text=blue] (water) at (-2.3,1.9) {\huge\faWater};
			\node[align=center] (faCSP-center) at (0,0) {$\mathcal{I}' + A$\\
				\small{$\mathcal{P}_1,\mathcal{P}_2,\mathcal{P}_3$}};
			\node (faCSP-label) at ($(faCSP-center)+(0,-1.8)$) {\textcolor{myViolet}{flow augmentation and CSP}};
			\draw[directededge,ultra thick,bend left] ($(faCSP-center)+(-2,0)$) to (faCSP-center); \node[align=center] (fa-description) at ($(faCSP-center)+(-1,1)$) {\small{\cref{thm:dfl-upgrade-vertices}}};

			\draw[directededge,ultra thick,bend right] (faCSP-center) to 
			node[midway,below,align=center] (CSP-description) {
				\begin{minipage}{2cm}
					\linespread{0.5}
					\scriptsize{constructing}\\\scriptsize{constraints}
				\end{minipage}} ($(faCSP-center)+(2,0)$);

			\node[myGreen] (CSP-constraints) at ($(faCSP-center)+(1.6,1.3)$) {
				\scalebox{0.65}{\begin{minipage}{4cm}
						\setlength{\tabcolsep}{2pt}
						\begin{tabular}{ll}
					$\rho^i_j$ -- &consistency\\
					$R^i_{j,j'}$ -- &\small{downwards-}\\
					&\small{closed}\\
					$\pi^{i,i'}_{j,j'}$ -- &permutation
					\end{tabular}
			\end{minipage}}};
	\end{tikzpicture}};

	\node (CSP-instance) at ($(flow-aug-CSP)+(2.8,0)$) {$\mathcal{C}_2$};
	
	\node[fill=myGreen!10!white,rounded corners] (twin-width-reduction) at ($(CSP-instance)+(1.8,0)$) {
		\begin{tikzpicture}
			\useasboundingbox (-1.3,-1.9) rectangle (1.3,1.9);
                        \node[text=green] (soybeans) at (-1.3,1.9) {\huge\faSeedling};
			\node (tw-center) at (0,0) {};
			\node (tw-label) at ($(tw-center)+(0,-1.8)$) {\textcolor{myViolet}{$\tww$ reduction}};
			\node (tw-vertex) at ($(tw-center)+(0,-0.7)$) {};
			\draw[directededge,ultra thick,bend right] ($(tw-center)+(-1.2,0)$) to (tw-vertex);
			\draw[directededge,ultra thick,bend right] (tw-vertex) to ($(tw-center)+(1.2,0)$);
			
			\node[align=center] (bounded-tww-label) at ($(tw-vertex)+(1.3,-0.2)$) {
				\begin{minipage}{2cm}
					\linespread{0.5}
					\scriptsize{bounded}\\\scriptsize{$\tww$}
			\end{minipage}};
			
			\draw[directededge,ultra thick,in=50,out=130,loop] (tw-vertex) to (tw-vertex)
			node[midway,above] (tw-description) {
				\scalebox{0.7}{\begin{minipage}{3cm}
					\cref{lem:tw-red}\\
				find a vertex that
				is in no solution and
				remove it from the domain
				\end{minipage}}};
	\end{tikzpicture}};

	\node (bounded-tw-instance) at ($(twin-width-reduction)+(1.8,0)$) {$\mathcal{C}'_2$};
	
	\node[fill=myGreen!10!white,rounded corners] (solve-CSP) at ($(bounded-tw-instance)+(1.6,0)$) {
		\begin{tikzpicture}
                  \useasboundingbox (-1.1,-1.9) rectangle (1.1,1.9);
                  \node[text=brown!80!red] (cook) at (-1.3,1.9) {\scalebox{1}[0.5]{\huge\faTrash}};
			\node (csp-center) at (0,0) {};
			\node (csp-label) at ($(csp-center)+(0,-1.8)$) {\textcolor{myViolet}{solve CSP}};
			\draw[directededge,ultra thick] ($(csp-center)+(-1,0)$) -- ($(csp-center)+(1,0)$)
			node[midway,above,align=center] (csp-description) {
				\scalebox{0.8}{\begin{minipage}{3cm}
						\cref{thm:bounded-twin-width}
			\end{minipage}}};
	\end{tikzpicture}};

	\node (solution) at ($(solve-CSP)+(1.5,0)$) {$S$};
\end{tikzpicture}
\caption{This flowchart gives an overview on the structure of the algorithm and where the results from the other sections are used.}
\label{fig:flowchart}
\end{figure}

\paragraph{Shadow removal.}
The first step in the algorithm is to remove vertices from the graph in order to ensure that if there is a solution, then there is also one without a shadow.
Recall the definition of being shadowless from \cref{sec:prelims} and recall \cref{thm:shadow-removal} which we restate for convenience below and prove in \cref{sec:shadow-removal}.

\ShadowRemovalTheorem*

We apply the algorithm of \cref{thm:shadow-removal}, yielding the family $\mathcal{Z}$.
We then iterate over all instances of \dmcthree\ resulting from bypassing a set $Z \in \mathcal{Z}$ in the input graph $G$.
For each such constructed instance we continue with the remainder of the algorithm as described below.
For simplicity, we call the instance of the current iteration $(G, k, (s_i, t_i)_{i \in [3]}, V^{\infty})$.
Note that, from now on, it is enough to find a shadowless solution.
This is not immediately relevant, but we crucially use this property when bounding the twin-width later~on.

\paragraph{Flow-augmentation and reduction to \permcsp.}
We continue with the reduction to \permcsp, however, without at first bounding the twin-width of all constraints.
(Recall that the definition of \permcsp{} can be found in \cref{ss:permCSP} 
 as \cref{def:permCSP}.)

Recall that $(G, k, (s_i, t_i)_{i \in [3]}, V^{\infty})$ is the instance of \dmcthree\ that we are working on.
The reduction to \permcsp\ works as follows.
For each terminal pair $(s_i, t_i)$, $i \in [3]$, run the algorithm from \cref{thm:dfl-upgrade-vertices} with input $s = s_i$, $t = t_i$, the graph $G$, and $k$.
We obtain a triple $(A_i, \mathcal{P}_i,\mathcal{B}_i)$ consisting 
of an arc set $A_i \subseteq V(G) \times V(G)$,
an $s_i t_i$-maxflow $\mathcal{P}_i$ in $G + A_i$,
and a partition $\mathcal{B}_i$ of the deletable arcs on $\mathcal{P}_i$. 
For each $\mathcal{P}_i$ let $k_i$ be its flow value and fix an arbitrary ordering $P_1^i, P_2^i, \ldots, P_{k_i}^i$ of the paths in $\mathcal{P}_i$.
Note that not necessarily $k_i \leq k$, but this is the case if $A_i$ is in a sense compatible with a solution, see the notion of safe augmentation below.
Hence, at this point if for some $i \in [3]$ we have $k_i > k$ we stop and return a failure symbol.
Otherwise we continue and denote $G_i \coloneqq G + A_i$.

In order to define a \permcsp\ instance, for each $i \in [3]$ and $j \in [k_i]$ we introduce a variable $x_j^i$ with domain $D^i_j \coloneqq V(P_j^i) \setminus (\{s_i, t_i\} \cup V^{\infty})$ and a variable $x_j'^i$ with domain $D'^i_j \coloneqq V(P_j^i) \setminus (\{s_i, t_i\} \cup V^{\infty})$.
The ordering $\leq_j^i$ of $D^i_j$ corresponds to a traversal of the path $P_j^i$ from $s_i$ to~$t_i$ and the ordering $\leq'^i_j$ of $D'^i_j$ corresponds to a traversal of the path $P_j^i$ from $t_i$ to~$s_i$.
(Intuitively, $x_{j}^{i}$ and $x'^{i}_{j}$ always choose the same vertex, but we need both orderings of their domain in order to define downwards-closed constraints below.)
Below we let $X$ denote the set of variables $x_{j}^{i}$ and $X'$ the set of variables $x'^{i}_{j}$.

As to the constraints, first, for each $i \in [3]$ and $j \in [k_i]$ we introduce a permutation constraint~$\rho_j^i \colon D^i_j \to D'^i_j$ that ensures that $x_j^i$ and $x'^i_j$ are the same, that is, for each $u \in D^i_j$ let $\rho_j^i(u) = u$.

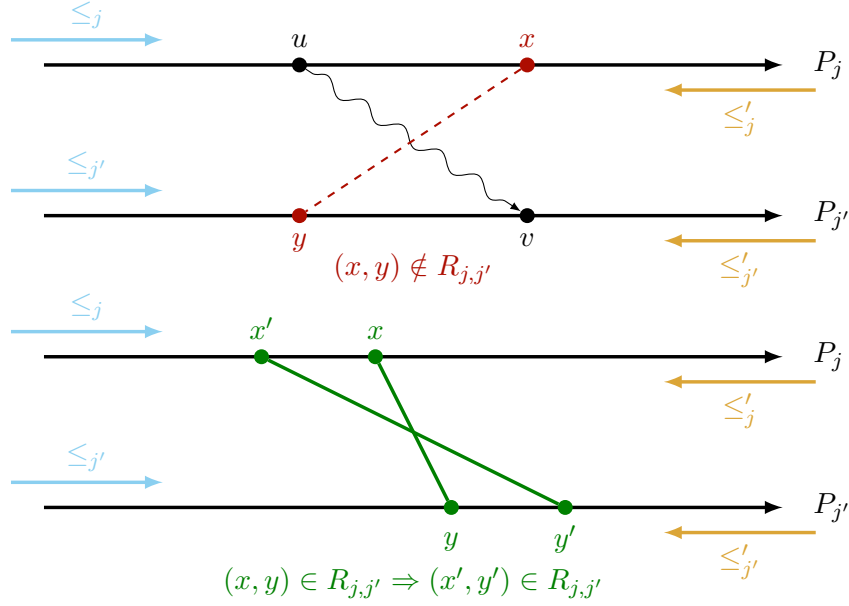
\begin{figure}[t]
	\centering
	\begin{tikzpicture}{scale=0.8}
		\colorlet{myGreen}{green!50!black}
		\definecolor{myRed}{rgb}{0.68, 0.05, 0.0}
		\colorlet{myBlue}{blue!90!black}
		\definecolor{myLightblue}{rgb}{0.54, 0.81, 0.94}
		\colorlet{myViolet}{myBlue!55!myRed}
		\colorlet{myOrange}{yellow!55!myRed}
		
		\tikzstyle{vertex} = [draw, circle, inner sep=.6mm, thick,fill=black]
		\tikzstyle{edge} = [draw=black,line width=1.3pt]
		\tikzstyle{nonedge} = [draw,thick,dashed]
		\tikzstyle{directededge} = [edge,->,-latex]
		\tikzstyle{path} = [decorate,decoration={snake,post length=1mm,pre length=1mm,segment length=17},->,-latex]
		
		\def\vdist{1}
		\def\hdist{5}
		\node (A) at (0,2) {\begin{tikzpicture}
				\node (center) at (0,0) {};
				\node (P-j-start) at ($(center)+(-\hdist,\vdist)$) {};
				\node[label=0:{$P_j$}] (P-j-end) at ($(center)+(\hdist,\vdist)$) {};
				\node (P-jj-start) at ($(center)+(-\hdist,-\vdist)$) {};
				\node[label=0:{$P_{j'}$}] (P-jj-end) at ($(center)+(\hdist,-\vdist)$) {};
				
				\draw[directededge,myLightblue] ($(P-j-start)+(-0.3,\vdist/3)$) -- ($(P-j-start)+(1.7,\vdist/3)$) node [midway, above] {$\leq_{j}$};
				\draw[directededge,myLightblue] ($(P-jj-start)+(-0.3,\vdist/3)$) -- ($(P-jj-start)+(1.7,\vdist/3)$) node [midway, above] {$\leq_{j'}$};
				
				\draw[directededge,myOrange] ($(P-j-end)+(0.3,-\vdist/3)$) -- ($(P-j-end)+(-1.7,-\vdist/3)$) node [midway, below] {$\leq'_{j}$};
				\draw[directededge,myOrange] ($(P-jj-end)+(0.3,-\vdist/3)$) -- ($(P-jj-end)+(-1.7,-\vdist/3)$) node [midway, below] {$\leq'_{j'}$};
				
				\draw[directededge] (P-j-start) to (P-j-end);
				\draw[directededge] (P-jj-start) to (P-jj-end);
				
				\node[vertex,label=90:{$u$}] (u) at ($(center)+(-1.5,\vdist)$) {};
				\node[vertex,myRed,label=90:{\textcolor{myRed}{$x$}}] (x) at ($(center)+(1.5,\vdist)$) {};
				\node[vertex,label=270:{$\vphantom{g}v$}] (v) at ($(center)+(1.5,-\vdist)$) {};
				\node[vertex,myRed,label=270:{\textcolor{myRed}{$\vphantom{g}y$}}] (y) at ($(center)+(-1.5,-\vdist)$) {};
				
				\draw[path] (u) to (v);
				\draw[nonedge,myRed] (x) to (y);
				
				\node at ($(center)+(0,-\vdist*1.7)$) {\textcolor{myRed}{$(x,y) \notin R_{j,j'}$}};
		\end{tikzpicture}};
		\node (B) at (0,-2) {\begin{tikzpicture}
				\node (center) at (0,0) {};
				\node (P-j-start) at ($(center)+(-\hdist,\vdist)$) {};
				\node[label=0:{$P_j$}] (P-j-end) at ($(center)+(\hdist,\vdist)$) {};
				\node (P-jj-start) at ($(center)+(-\hdist,-\vdist)$) {};
				\node[label=0:{$P_{j'}$}] (P-jj-end) at ($(center)+(\hdist,-\vdist)$) {};
				
				\draw[directededge,myLightblue] ($(P-j-start)+(-0.3,\vdist/3)$) -- ($(P-j-start)+(1.7,\vdist/3)$) node [midway, above] {$\leq_{j}$};
				\draw[directededge,myLightblue] ($(P-jj-start)+(-0.3,\vdist/3)$) -- ($(P-jj-start)+(1.7,\vdist/3)$) node [midway, above] {$\leq_{j'}$};
				
				\draw[directededge,myOrange] ($(P-j-end)+(0.3,-\vdist/3)$) -- ($(P-j-end)+(-1.7,-\vdist/3)$) node [midway, below] {$\leq'_{j}$};
				\draw[directededge,myOrange] ($(P-jj-end)+(0.3,-\vdist/3)$) -- ($(P-jj-end)+(-1.7,-\vdist/3)$) node [midway, below] {$\leq'_{j'}$};
				
				\draw[directededge] (P-j-start) to (P-j-end);
				\draw[directededge] (P-jj-start) to (P-jj-end);
				
				\node[vertex,myGreen,label=90:{\textcolor{myGreen}{$x'$}}] (xx) at ($(center)+(-2,\vdist)$) {};
				\node[vertex,myGreen,label=90:{\textcolor{myGreen}{$x$}}] (x) at ($(center)+(-0.5,\vdist)$) {};
				\node[vertex,myGreen,label=270:{\textcolor{myGreen}{$\vphantom{fg}y$}}] (y) at ($(center)+(0.5,-\vdist)$) {};
				\node[vertex,myGreen,label=270:{\textcolor{myGreen}{$\vphantom{fg}y'$}}] (yy) at ($(center)+(2,-\vdist)$) {};
				
				\draw[edge,myGreen] (xx) to (yy);
				\draw[edge,myGreen] (x) to (y);
				
				\node at ($(center)+(0,-\vdist*2)$) {\textcolor{myGreen}{$(x,y) \in R_{j,j'} \Rightarrow (x',y') \in R_{j,j'}$}};
		\end{tikzpicture}};
	\end{tikzpicture}
	\caption{The downwards-closed constraints introduced into the \permcsp\ instance~$\mathcal{C}_1$.
		All identifiers $\leq, \leq', R, P$ to be understood with an index $^i$.}
	\label{fig:downwards-closed-constraints}
\end{figure}

Next, we define the constraints that ensure that for each $i \in [3]$ the vertices selected by the variables form an $s_it_i$-separator.
For each $i \in [3]$ and each pair of variables $x_j^i \in X$, $x'^i_{j'} \in X'$ introduce a constraint $R_{j, j'}^i \subseteq D^i_j \times D'^i_{j'}$ as follows; refer to \cref{fig:downwards-closed-constraints} for an illustration.
At first, put $R_{j, j'}^i = D^i_j \times D'^{i}_{j'}$.
Then, for each $u \in V(P^i_j)$ and $v \in V(P^i_{j'})$ such that there is a path $Q$ in $G_i$ from $u$ to $v$ such that $Q$ is internally vertex-disjoint from each path $P_j^i \in \mathcal{P}_i$, remove from $R_{j, j'}^i$ all pairs $(x, y)$ such that $u \leq^i_j x$ and $v \leq'^i_{j'} y$.
(This means that $u$ occurs before $x$ on path $P^i_j$ and $v$ occurs after $y$ on path $P^i_{j'}$.
Intuitively, no solution may choose $x$ and $y$ because $Q$ bypasses the corresponding vertex set, showing it is not a separator.
Note that the path $Q$ may consist of a single edge and that one of its endpoints may be $s_i$ or $t_i$.)
This finishes the description of the constraint $R^i_{j, j'}$.
(Note that not necessarily $R^i_{j, j'} = R^i_{j', j}$.)
Note that if $(x, y) \in R_{j, j'}^i$ and $x' \leq_j x$ and $y' \leq'_{j'} y$ then $(x', y') \in R_{j, j'}^i$:
Otherwise, at the point where we have removed $(x', y')$ from $R_{j, j'}^i$ in the above construction, say due to some pair $(u, v)$, we have $u \leq_j x' \leq_j x$ and $v \leq'_j y' \leq'_j y$ and thus we would have removed $(x, y)$ as well, a contradiction.
Thus, $R^i_{j, j'}$ is downwards-closed.

We add further constraints to the \permcsp\ instance below.
But first we observe that already at this point, every solution to the \dmcthree\ instance $(G, k, (s_i, t_i)_{i \in [3]}, V^{\infty})$ induces a solution for the \permcsp\ instance and every solution to the \permcsp\ instance induces three separators between the three terminal pairs.
Let $\mathcal{C}_1$ denote the \permcsp\ instance constructed so far, that is, \[\mathcal{C}_1 = (X \cup X', (D^i_j, \leq^i_j, D'^i_j, \leq'^i_j)_{i \in [3], j \in [k_i]}, (\rho^i_j)_{i \in [3], j \in [k_i]}, (R^i_{j, j'})_{i \in [3], j, j' \in [k_i]}).\]
Observe that we can construct $\mathcal{C}_1$ in polynomial time because the algorithm of \cref{thm:dfl-upgrade-vertices} runs in polynomial time and the paths underlying the construction of $R^i_{j, j'}$ can be computed in polynomial time.

For proving the soundness of the algorithm, we need the following definition. 
Let $S$ be a solution to $(G, k, (s_i, t_i)_{i \in [3]}, V^{\infty})$.
Define the event \emph{safely augmented (wrt.~$S$)} as the intersection of the three events $E_i$, $i \in [3]$, that state that there exists a minimal $s_it_i$-separator $S_i$ that is contained in $S$ such that $(A_i, \mathcal{P}_i)$ is compatible with $S_i$.
We also say that the $S_i$ are the \emph{witnesses} to having safely augmented.
Note that, if we have safely augmented, then the flow value $k_i$ of each $\mathcal{P}_i$ is at most $k$.
In order to prove soundness, we make use of the following two statements.

\begin{lemma}\label{lem:dmc-to-csp-probab}
  Let $S$ be an arbitrary fixed solution to $(G, k, (s_i, t_i)_{i \in [3]}, V^{\infty})$.
  Then with probability at least $2^{-\Oh(k^{4} \log k)}$ we have safely augmented wrt.~$S$.
\end{lemma}
\begin{proof}
  Let $S_i$, $i \in [3]$, be a minimal $s_it_i$ separator (in $G$) contained in $S$.
  The three events $E_i$ are independent from each other.
  They each have probability lower bounded by the probability of the event that $(A_i, \mathcal{P}_i)$ is compatible with $S_i$.
  By \cref{thm:dfl-upgrade-vertices} they thus each have probability at least $2^{-\Oh(k^4 \log k)}$, which implies the desired bound.
\end{proof}

\begin{lemma}\label{lem:dmc-to-csp-1}
  Assume we have safely augmented and let $S_i$, $i \in [3]$, be witnesses to that fact.
  Then, for each $i \in [3]$ and for each $j \in [k_i]$ we have $|V(P^i_j) \cap S_i| = 1$.
  Define the mapping $\phi$ by defining $\phi(x^i_j)$ and $\phi(x'^i_j)$ both as the single vertex in $V(P^i_j) \cap S_i$ for each $i \in [3]$ and $j \in [k_i]$.
  Then $\phi$ is a solution to the \permcsp\ instance $\mathcal{C}_1$.
\end{lemma}
\begin{proof}
  Fix $i \in [3]$.
  Since we have safely augmented, $(A_i, \mathcal{P}_i)$ is compatible with $S_i$.
  Thus $\mathcal{P}_i$ is an $s_i t_i$-maxflow of value $|S_i|$, showing that for each $j \in [k_i]$ we have $|V(P^i_j) \cap S_i| = 1$.
  We define $\phi(x^i_j)$, $\phi(x'^i_j)$ as specified in the statement.
  Clearly the permutation constraints $\rho^i_j$ are satisfied.
  Assume towards a contradiction that some downwards-closed constraint $R^i_{j, j'}$ is not satisfied.
  Thus, at some point in the construction of $R^i_{j, j'}$ we have removed $(\phi(x^i_j), \phi(x'^i_{j'}))$ from $R^i_{j, j'}$, say due to some path $Q$ from the vertex $u \in V(P^i_j)$ to the vertex $v \in V(P^i_{j'})$.
  By the construction of $R^i_{j, j'}$, we have $u \leq^i_j \phi(x^i_j)$ and $v \leq'^i_j \phi(x'^i_{j'})$.
  We construct an $s_i t_i$-path $Q'$ in $G_i$ by following $P^i_j$ from $s_i$ to $u$, then adding $Q$, and finally following $P^i_{j'}$ from $v$ to $t_i$.
  Observe that $V(Q') \cap S_i = \emptyset$.
  Since $A_i$ is compatible with $S_i$, this implies that there is an $s_i t_i$-path in $G$, a contradiction to the fact that $S_i$ is an $s_i t_i$-separator.
\end{proof}

In order to prove completeness, we need the following.
\begin{lemma}\label{lem:csp-to-dmc-1}
  If $\phi$ is a solution to the \permcsp\ instance $\mathcal{C}_1$, then for each $i \in [3]$ the set $\{ \phi(x^i_j) \mid j \in [k_i]\}$ is an $s_it_i$-separator in $G$.
\end{lemma}
\begin{proof}
  Fix $i \in [3]$ and let $S_i = \{ \phi(x^i_j) \mid j \in [k_i]\}$.
  We show that $S_i$ is an $s_i t_i$-separator in $G_i$.
  Since $G_i$ is a supergraph of $G$ the statement then follows.
  For a contradiction, assume that there is an $s_i t_i$-path $Q$ in $G_i - S_i$.
  Observe that $Q$ has at least one internal vertex that is contained in a path $P^i_j$; otherwise, $\mathcal{P}^i$ would not be an $s_i t_i$-maxflow in $G_i$.
  Divide $Q$ into segments $Q_1, Q_2, \ldots$ such that each segment is of maximal length with respect to not having an internal vertex of a path $P^i_j$, $j \in [k_i]$.
  That is, the endpoints of the segments are either $s_i$, $t_i$, or an internal vertex of a path $P^i_j$.
  Label each segment $Q_p$ by a label in $\{-, +\} \times \{-, +\}$ as follows.
  Let $u, v$ be the first and the last vertex of $Q_p$, respectively, and observe that one of them occurs as internal vertex on a path in $\mathcal{P}_i$.
  The first part of the label of $Q_p$ is $-$ if $u = s_i$ or if $u$ occurs on a path $P^i_j$ before $\phi(x^i_j)$; otherwise the first part of the label is $+$.
  The second part of the label is $+$ if $v = t_i$ or if $v$ occurs on a path $P^i_{j'}$ after $\phi(x^i_{j'})$; otherwise the second part of the label is~$-$.
  Note that, since $Q$ is an $s_i t_i$-path, there is at least one segment with label $(- , +)$, say segment $Q_p$.
  Let $u, v$ be the first and the last vertex of $Q_p$, respectively.
  Without loss of generality, by symmetry, assume that $u \neq s_i$, that is, $u$ appears as internal vertex on a path in $\mathcal{P}_i$, say $P^i_j$.
  If $v$ appears as internal vertex on a path in $\mathcal{P}_i$ then let $P^i_{j'}$ be that path; otherwise take $P^i_{j'}$ to be an arbitrary path in $\mathcal{P}_i$ different from $P^i_j$.
  By construction of $Q_p$, we have $u \leq^i_j \phi(x^i_j)$ and $v \leq'^i_j \phi(x'^i_{j'})$.
  But then, by construction of $R^i_{j, j'}$ we would have removed $(\phi(x^i_j), \phi(x'^i_j))$ from $R^i_{j, j'}$ due to the path $Q_p$, a contradiction.
  Thus, indeed $S_i$ is an $s_i t_i$-separator in $G_i$.
\end{proof}

Next, we aim to add the consistency constraints between variables mentioned in the outline above.
To this end, we iterate over all possibilities of variables being assigned to the same vertex.
More precisely, we iterate over all of the $2^{\Oh(k \log k)}$ partitions $\mathcal{X}$ of the variable set $X$ into at most $k$ parts of size at most three such that no part contains two variables $x_\cdot^i$, $x_{\cdot}^{i'}$ with $i = i'$.
(Intuitively, these capture all possibilities because no two paths in $\mathcal{P}_{i}$ share a vertex and thus no cluster of pairwise equal variables exceeds size three.)
We call this the \emph{consistency iteration}.

Next, for each pair of variables $x_j^i, x_{j'}^{i'}$ in the \emph{same} part in $\mathcal{X}$, restrict their domains to the set of shared vertices of $P_j^i$ and $P_{j'}^{i'}$, that is, replace  both $D_j^i$ and $D_{j'}^{i'}$ by  their intersection $D_j^i \cap D_{j'}^{i'}$.
Perform analogous restrictions to the domains of the variables $x'^i_j, x'^{i'}_{j'}$.
Omit the thereby invalidated bindings from the constraints $\rho^i_j$ and~$R^i_{j, j'}$.

Finally, we introduce the permutation constraints enforcing the guessed consistency relation represented by~$\mathcal{X}$.
For every pair of variables $x^i_j$, $x^{i'}_{j'}$ contained in same part of $\mathcal{X}$, we introduce the constraint $\pi^{i, i'}_{j, j'} \colon D^i_j \to D^{i'}_{j'}$ mapping each $u \in D^i_j$ as $\pi^{i, i'}_{j, j'}(u) = u \in D^{i'}_{j'}$.
This concludes the description of the reduction to \permcsp.
Let $\mathcal{C}_2$ denote the resulting \permcsp\ instance \[(X \cup X', (D^i_j, \leq^i_j, D'^i_j, \leq'^i_j)_{i \in [3], j \in [k_i]}, (\rho^i_j)_{i \in [3], j \in [k]}, (R^i_{j, j'})_{i \in [3], j, j' \in [k_i]}, (\pi^{i, i'}_{j, j'})_{i \neq i' \in [3], j \in [k_i], j' \in [k_{i'}]}).\]
Note that iterating over all possibilities for $\mathcal{X}$ can be done in $2^{\Oh(k \log k)}$ time and thus constructing all the instances $\mathcal{C}_2$ takes $2^{\Oh(k \log k)} \cdot n^{O(1)}$ time.

We can now extend \cref{lem:dmc-to-csp-1,lem:csp-to-dmc-1} to $\mathcal{C}_2$.
To this end, let $S$ be a solution to $(G, k, \allowbreak (s_i, t_i)_{i \in [3]}, \allowbreak V^{\infty})$ and assume we have safely augmented with witnesses $S_i$.
We say that $\mathcal{X}$ \emph{complies (with $S$ and the witnesses $S_i$)} if for each pair $i, i' \in [3]$, each $j \in [k_i]$, and $j' \in [k_{i'}]$ we have that $V(P^i_j) \cap S_i = V(P^{i'}_{j'}) \cap S_{i'}$ if $x^i_j$ and $x^{i'}_{j'}$ are both contained in the same part of $\mathcal{X}$.

\begin{lemma}\label{lem:dmc-to-csp-correct}
  Let $S$ be a solution to $(G, k, (s_i, t_i)_{i \in [3]}, V^{\infty})$ and assume we have safely augmented with witnesses $S_i$.
  Then, one of the partitions considered in the consistency iteration complies.
\end{lemma}
\begin{proof}
  Since we have safely augmented and by \cref{lem:dmc-to-csp-1}, we have $|V(P^i_j) \cap S_i| = 1$ for each $i \in [3]$ and for each $j \in [k_i]$.
  For every vertex $v\in S$ define a set of variables $X_v \coloneqq \{x^i_j \mid V(P^i_j) \cap S_i = \{v\}\}$, which yields a complying partition $\mathcal{X} \coloneqq \{X_v \mid v \in S\}$ of $X$.
  To see that $\mathcal{X}$ is considered in the consistency iteration observe that $|\mathcal{X}| \leq k$.
  Furthermore, for each $i \in [3]$ there are no two variables $x^i_j$, $x^i_{j'}$ in the same part in $\mathcal{X}$, which proves the claim. 
\end{proof}

\begin{lemma}
  Let $S$ be a solution to $(G, k, (s_i, t_i)_{i \in [3]}, V^{\infty})$, assume we have safely augmented with witnesses $S_i$ and assume that $\mathcal{X}$ complies.
  Then, we have $|V(P^i_j) \cap S_i| = 1$ for each $i \in [3]$ and for each $j \in [k_i]$.
  Define the mapping $\phi$ by defining $\phi(x^i_j)$ and $\phi(x'^i_j)$ both as the single vertex in $V(P^i_j) \cap S_i$ for every $i \in [3]$ and $j \in [k_i]$.
  Then, $\phi$ is a solution to the \permcsp\ instance $\mathcal{C}_2$.
\end{lemma}
\begin{proof}
  By \cref{lem:dmc-to-csp-1}, $|V(P^i_j) \cap S_i| = 1$ and the mapping $\phi$ satisfies all constraints $\rho^i_j$ and $R^i_{j, j'}$.
  It remains to show that the values of $\phi$ have not been removed from the domains and that the permutation constraints $\pi^{i, i'}_{j, j'}$ are satisfied.
  
  For the claim about the domains, fix some variable $x^i_j$.
  Towards a contradiction, assume that $\phi(x^i_j)$ was removed in the domain-restriction step.
  Then, there is some other variable $x^{i'}_{j'}$ that is in the same part of $\mathcal{X}$ such that $\phi(x^i_j) \notin D^{i'}_{j'}$.
  In other words, $V(P^i_j) \cap S_i \neq V(P^{i'}_{j'}) \cap S_{i'}$.
  Since $\mathcal{X}$ complies, $x^i_j$ and $x^{i'}_{j'}$ are in different parts in $\mathcal{X}$, a contradiction.

  In order to show that the permutation constraints are satisfied, fix some constraint $\pi^{i, i'}_{j, j'}$.
  By construction of $\pi^{i, i'}_{j, j'}$, variables $x^i_j$ and $x^{i'}_{j'}$ are in the same part of $\mathcal{X}$.
  Since $\mathcal{X}$ complies, $V(P^i_j) \cap S_i = V(P^{i'}_{j'}) \cap S_{i'}$.
  Thus, by definition of $\phi$ we have $\phi(x^i_j) = \phi(x^{i'}_{j'})$, that is, $\pi^{i, i'}_{j, j'}$ is satisfied.
\end{proof}

Finally, we prove that solutions to $\mathcal{C}_2$ yield solutions to our \dmcthree\ instance.

\begin{lemma}\label{lem:csp-to-dmc-2}
  If $\phi$ is a solution to the \permcsp\ instance $\mathcal{C}_2$, then the set $\{\phi(x^i_j) \mid i \in [3], j \in [k_i]\}$ is a solution to $(G, k, (s_i, t_i)_{i \in [3]}, V^\infty)$.
\end{lemma}
\begin{proof}
  Let $S = \{\phi(x^i_j) \mid i \in [3], j \in [k_i]\}$ and for each $i \in [3]$ let $S_i = \{ \phi(x^i_j) \mid j \in [k_i]\}$.
  Note that $S = S_1 \cup S_2 \cup S_3$.
  In comparison to $\mathcal{C}_1$, instance $\mathcal{C}_2$ contains only smaller domains and more constraints.
  Thus, the conclusion of \cref{lem:csp-to-dmc-1} still holds for $\mathcal{C}_2$.
  Hence, each $S_i$ is an $s_i t_i$-separator in~$G$.
  Thus it remains to show that $|S| \leq k$.
  For this, observe that, due to the constraints $\pi^{i, i'}_{j, j'}$, we have that for each part $Y \in \mathcal{X}$ all the variables in $Y$ have the same value.
  Since $|\mathcal{X}| \leq k$ by construction of $\mathcal{X}$ it follows that $|S| \leq k$, as required.
\end{proof}

We now continue working with $\mathcal{C}_2$; first reducing the twin-width of its constraints and then applying the algorithm from \cref{sec:low-tw-csp-algo}.

\paragraph{Twin-width reduction by irrelevant vertices.}
We next show how to bound the twin-width of the constraints in $\mathcal{C}_2$.
First, observe that the twin-width of the constraints $\rho^i_j$ is already bounded.
Recall the definition of grid rank $\gridrank$ from \cref{sec:low-tw-csp-algo}.

\begin{lemma}
  For each $i \in [3]$ and $j \in [k_i]$ we have $\gridrank(\Adj(\rho^i_j)) \leq 1$.
\end{lemma}
\begin{proof}
  Let $M = \Adj(\rho^i_j)$.
  Observe that $M$ is an anti-diagonal matrix.
  Each cell of $M$ with at least two distinct rows or columns thus contains an entry of the anti-diagonal.
  Consider the upper-left cell $C$ in a rank-$k$ division of $M$.
  For a contradiction, assume that $k \geq 2$.
  Then, $C$ contains an entry of the anti-diagonal.
  Consider the cell $C'$ that is south east of~$C$.
  Cell $C'$ contains only zero entries, a contradiction.
\end{proof}

The crucial constraints are the constraints $\pi^{i, i'}_{j, j'}$ and they are indeed a priori not of bounded twin-width.
However, using an irrelevant vertex argument, we can bound their twin-width.

\begin{restatable}{lemma}{TwRedLemma}
	\label{lem:tw-red}
	There exists a computable function $h \colon \N \to \N$ and an algorithm that, given an instance $(G, k, (s_i, t_i)_{i \in [3]}, V^{\infty})$, 
        the constraint $\pi^{i, i'}_{j, j'}$ in $\mathcal{C}_2$ for $i, i' \in [3]$ distinct, $j \in [k_i]$, and $j' \in [k_{i'}]$ and the augmented paths $P^i_j$, $P^{i'}_{j'}$
together with the partitions $\mathcal{B}_{i}$ and $\mathcal{B}_{i'}$
(from \cref{thm:dfl-upgrade-vertices}),
  certifies that $\gridrank(\Adj(\pi^{i, i'}_{j, j'})) \leq h(k)$ or finds a vertex $v \in D^i_j \cap D^{i'}_{j'}$ such that 
  there is no shadowless solution $S$ with $v \in S$ and:
  \begin{itemize}
  \item all vertices before $v$ on $P^i_j$ do not reach $t_i$ and all vertices after $v$ on $P^i_j$ are not reachable from $s_i$ in $G-S$;
  \item all vertices before $v$ on $P^{i'}_{j'}$ do not reach $t_{i'}$ and all vertices after $v$ on $P^{i'}_{j'}$ are not reachable from $s_{i'}$ in $G-S$.
    \end{itemize}
    Moreover, the algorithm runs in fixed-parameter time with respect to~$k$.
\end{restatable}
The proof is given in \cref{sec:tw-red}.

Note that if the instance is safely augmented and $\mathcal{X}$ complies with $S$, then $P^i_j$ contains a unique vertex of $S_i$, $P^{i'}_{j'}$ contains a unique vertex of $S_{i'}$, and this is
the same vertex. \cref{lem:tw-red} returns $v$ that is \emph{guaranteed not to be the said vertex}, so we can remove it from the domains of the variables corresponding to $P^i_j$ and $P^{i'}_{j'}$.
Formally, we use \cref{lem:tw-red} as follows.
We iterate over all consistency constraints. 
That is, for each $i, i' \in [3]$, $j \in [k_i]$, and $j' \in [k_{i'}]$, we consider the constraint $\pi^{i, i'}_{j, j'}$ in $\mathcal{C}_2$.
We iteratively apply the algorithm of \cref{lem:tw-red} to it. 
If it returns a vertex $v$, 
then we remove $v$ from both the domains $D^i_j$ and $D^{i'}_{j'}$ of $x^i_j$ and $x^{i'}_{j'}$, and repeat.
If it returns that the grid rank is at most $h(k)$, we continue to the next constraint.
Since we may drop at most $n$ vertices from a domain, this eventually leads to an empty domain for a variable, that is, a no-instance, or a situation in which for every constraint $\pi^{i, i'}_{j, j'}$ in $\mathcal{C}_2$ the matrix has grid rank bounded by $h(k)$.
In the latter case, we apply \cref{thm:bounded-twin-width} to solve $\mathcal{C}_2$ in fixed-parameter time with respect to~$k$.

We are now able to prove \cref{thm:dmcthree-fpt}.

\DMCthreeFPTTheorem*
\begin{proof}%
  We claim that the algorithm described in this section solves $(G, k, (s_i, t_i)_{i \in [3]}, V^{\infty})$ in fixed-parameter time with respect to $k$ with large-enough probability.
  By the arguments given throughout the section, the algorithm indeed runs in fixed-parameter time.
  If the algorithm returns a vertex set $S$, then $S$ is a solution to $(G, k, (s_i, t_i)_{i \in [3]}, V^{\infty})$ by \cref{lem:csp-to-dmc-2}.

  Assume now that there is a solution $S$ to $(G, k, (s_i, t_i)_{i \in [3]}, V^{\infty})$.
  By \cref{thm:shadow-removal}, we may assume that $S$ is shadowless.
  By \cref{lem:dmc-to-csp-probab}, with probability at least $2^{-\Oh(k^4 \log k)}$ the algorithm safely augmented with respect to $S$.
  By \cref{lem:dmc-to-csp-correct}, one of the considered $\mathcal{X}$ complies with $S$.
  Thus, by \cref{lem:csp-to-dmc-2}, the \permcsp\ instance $\mathcal{C}_2$ has a solution.
  By \cref{lem:tw-red}, $\mathcal{C}_2$ maintains having a solution even after removing the vertices computed in \cref{lem:tw-red} from the respective domains.
  Thus, the algorithm of \cref{thm:bounded-twin-width} returns a solution, as required.
\end{proof}

\section{Irrelevant vertex rule---Proof of \cref{lem:tw-red}}
\label{sec:tw-red}

In this section, we give a tool that reduces the complexity of complicated permutation constraints in $\mathcal{C}_2$.
In fact, we argue that due to the shadow removal (as described in \cref{sec:shadow-removal}) and an improved version of the flow augmentation (\cref{thm:dfl-upgrade}), we obtain the following.
If the permutation constraint has a high grid rank ($\gridrank$), then we can explicitly find a vertex $v$ that cannot 
play the role of the solution vertex in the given permutation constraint.

First, we provide some very brief intuition.
Shadow removal grants us some additional information about the reachability relation with respect to the terminals:
For example, if some vertex cannot be reached from $s_1$ and $s_2$ after removal of a shadowless solution $S$, then it must be reachable from $s_3$.
Symmetrically, if a vertex cannot reach $t_1$ and $t_2$, it has to reach~$t_3$.
Using the large grid rank, this additional reachability can be easily extended first along flow path $P^{i'}_{j'}$ and then along flow path $P^i_j$, which gives a forbidden path $s_3t_3$-path.
However, this does not prove the existence of such an $s_3t_3$-path in the original graph, as the flow paths might contain augmented edges.
Therefore, we make use of the improved flow augmentation (\cref{thm:dfl-upgrade}) in order to still reach a contradiction.
In particular, if $S$ satisfies the additional assumption stated in \cref{lem:tw-red} (which still follows from the flow augmentation), the soybeans (introduced in \cref{ss:dfl}) provide the needed connectivity without using the augmented edges.
Unfortunately, their structure is (and has to be; see \cref{sec:soybean}) a bit more complicated than just a simple path, but using additional shadow removal arguments, we are able to derive the existence of an $s_3t_3$-path and therefore, the desired contradiction.

Note one more complication:
The soybeans may use vertices in other flow paths, but as they are vertex disjoint, at most $k$ of them can intersect $S$, so if we set the threshold high enough, we still reach the contradiction.

\TwRedLemma*

\begin{proof}%
  \newcommand{\splitP}{\zeta}
  \newcommand{\splitPv}{2q(k,q(k,k+1)+1+k)+1}
\newcommand{\division}{2\zeta}
\newcommand{\soybeansB}{q(k,q(k,k+1)+1+k)}
\newcommand{\soybeansRes}{q(k,k+1)+1+k}
\newcommand{\soybeansResTwo}{q(k,k+1) +1}

\newcommand{\df}{\coloneqq}

For simplicity, denote $P_j = P^i_j$ and $P_{j'} = P^{i'}_{j'}$.
Without loss of generality suppose that $i=1$ and $i'=2$.

For a (sub)matrix $A$ of $\Adj(\pi^{1,2}_{j,j'})$, we refer to the vertices that correspond to domains of $A$ as the \emph{vertices corresponding to $A$}.
Note that whenever there is a $1$ in the (sub)matrix $A$, it means that the corresponding two domain values are, in fact, the same vertex in $G$.

Recall functions $c$ and $q$ of \cref{thm:dfl-upgrade-vertices}.
We set $\splitP\df \splitPv$.
Let $\rho$ be the bipartite Ramsey number for $(c(k))^2$ colors and a
monochromatic biclique of size $\division$; that is,
every edge coloring of $K_{\rho,\rho}$ with $(c(k))^2$ colors contains
a monochromatic copy of $K_{\division,\division}$. 
We apply \cref{thm:grid-minor-apx}
to $\Adj(\pi^{1,2}_{j,j'})$ and $\rho$
to either obtain that $\gridrank(\Adj(\pi^{1,2}_{j,j'}))$ is bounded
by $2^{\Oh(\rho \log \rho)}$ (and a computable function of $k$), 
or obtain an $\rho$-division of $\Adj(\pi^{1,2}_{j,j'})$ that has at least one
$1$ in every cell.
In the first case, we return that $\gridrank(\Adj(\pi^{1,2}_{j,j'}))$ is bounded by a computable function of $k$,
   in the second case we proceed further; we leverage the obtained $\rho$-division to obtain the desired irrelevant vertex $v$.

We fix one such entry with value 1 for every $(p,q)$-cell of the said division, which corresponds to a vertex, and call it the \emph{$(p,q)$-representative}.
For every cell $(p,q)$, we color it with a pair $(B_1,B_2)$, where
for $\iota = 1,2$ the set $B_\iota \in \mathcal{B}_\iota$ is the set
containing the $(p,q)$-representative in the application of \cref{thm:dfl-upgrade-vertices}
for the pair $s_\iota, t_\iota$. 
Clearly, there are at most $(c(k))^2$ colors. 
By the choice of $\rho$, there is a coarser $\division$-division of $\Adj(\pi^{1,2}_{j,j'})$
such that we can choose a representative in every cell of this division
such that all representatives are of the same color $(B_1,B_2)$. 
In what follows, we only work with the latter division and representatives
and hence use the name $(p,q)$-representative for them.
For a (sub)matrix $A$ of $\Adj(\pi^{1,2}_{j,j'})$, we say \emph{representatives} for the set of all $(p,q)$-representatives within $A$. 

We choose $v$ as the $(\splitP,\splitP)$-representative.
Now, towards a contradiction, suppose that $v\in S$
and $v$ splits $P_j$ and $P_{j'}$ as in the lemma statement.

This split in particular applies to the representatives.
The first part of $P_j$ (resp.\ $P_{j'}$) contains the vertices not reaching $t_1$ (resp.\ $t_2$) up to the cell index $i=\splitP -1$ (resp.\ $j=\splitP -1$) and the second part contains vertices not reachable from $s_1$ (resp.\ $s_2$) in $G-S$.
Based on the above, we split the cells of the matrix $\Adj(\pi^{1, 2}_{j, j'})$ into four quadrants and take submatrices consisting only of their representatives: $A_{t12}$, $A_{t1s2}$, $A_{t2s1}$, and $A_{s12}$, where the subscript indicates the non-reachability of the quadrant.
More formally, $A_{t12}$ consists of the $(p,q)$-representatives for $1\le p,q \le \splitP -1$
(and they do not reach $t_1$ nor $t_2$ in $G-S$),
$A_{t1s2}$ consists of the $(p,q)$-representatives for $1\le p \le \splitP -1$ and $\splitP +1 \le q\le \division$ (and they do not reach $t_1$ and are not reachable from $s_2$ in $G-S$),
$A_{t2s1}$ consists of the $(p,q)$-representatives for $1\le q \le \splitP -1$ and $\splitP +1 \le p\le \division$ (and they do not reach $t_2$ and are not reachable from $s_1$ in $G-S$), and finally
$A_{s12}$ consists of the $(p,q)$-representatives for $\splitP +1 \le p,q \le \division$
(and they are not reachable from $s_1$ nor $s_2$ in $G-S$).

As $S$ is shadowless thus all vertices in $A_{t12}$ have to reach $t_3$ and all vertices in $A_{s12}$ have to be reachable from $s_3$ in $G-S$.
In order to make use of the soybeans given by the flow-augmentation (\cref{thm:dfl-upgrade-vertices}) we need to define suitable interlaced sets.
Let $\xi_o$ be the vertices of $P_{j'}$ corresponding to $(p,q)$-representatives of $A_{s12}$ for every odd $p\ge\splitP+1$ and any one fixed $q\ge\splitP +1$,
and let $\xi_e$ be the set of vertices of $P_{j'}$ corresponding to $(p,q)$-representatives of $A_{t1s2}$ for every even $p\ge\splitP+1$ and $q\df p-\splitP$.
It follows that $\xi_o$ and $\xi_e$ form interlaced sets on path $P_{j'}$ of size
	\[\lfloor\splitP/2\rfloor= \lfloor(\splitPv)/2\rfloor\ge \soybeansB.\]
	
  Hence we obtain $\soybeansRes$ $\xi_o,\xi_e$-soybeans by \cref{thm:dfl-upgrade-vertices} (note that the interlaced sets are within one set $B_2 \in \mathcal{B}_2$).
  Let $Q_1Q_2$ be one of these $\xi_o\xi_e$-soybean that does not intersect $S$.
  Pick two vertices $v_1 \in \xi_o$ and $v_2 \in \xi_e$.
  As $v_1$ corresponds to a vertex of $A_{s12}$, vertex $v_1$ is 
not reachable from $s_1$ and $s_2$ in $G-S$.
As $Q_1Q_2$ is disjoint with the solution $S$, 
$s(Q_1Q_2)$ is not reachable from $s_1$ nor $s_2$ in $G-S$.
Since $S$ is shadowless, $s(Q_1Q_2)$ is reachable from $s_3$ in $G-S$.
Again as $Q_1Q_2$ is disjoint with the solution $S$, $v_2$ is reachable from $s_3$ in $G-S$.

For an illustration, consider the left part of \cref{fig:tw-red-soybeans}.
Some of the $\xi_o,\xi_e$-soybeans may intersect $S$, but as they are pairwise vertex disjoint, at most $k$ of them can.
Hence, the above properties hold for $\soybeansResTwo$ of them.
We restrict the submatrix $A_{t1s2}$ only to entries in $\xi_e$ that correspond to $\xi_o,\xi_e$-soybeans that are disjoint from $S$.
We call the resulting submatrix $A_{t1s2}'$.
Observe that $A_{t1s2}'$ has at least $\soybeansResTwo$ columns with at least one $1$.

\begin{sidewaysfigure}
	\centering
  \includegraphics[scale=0.59]{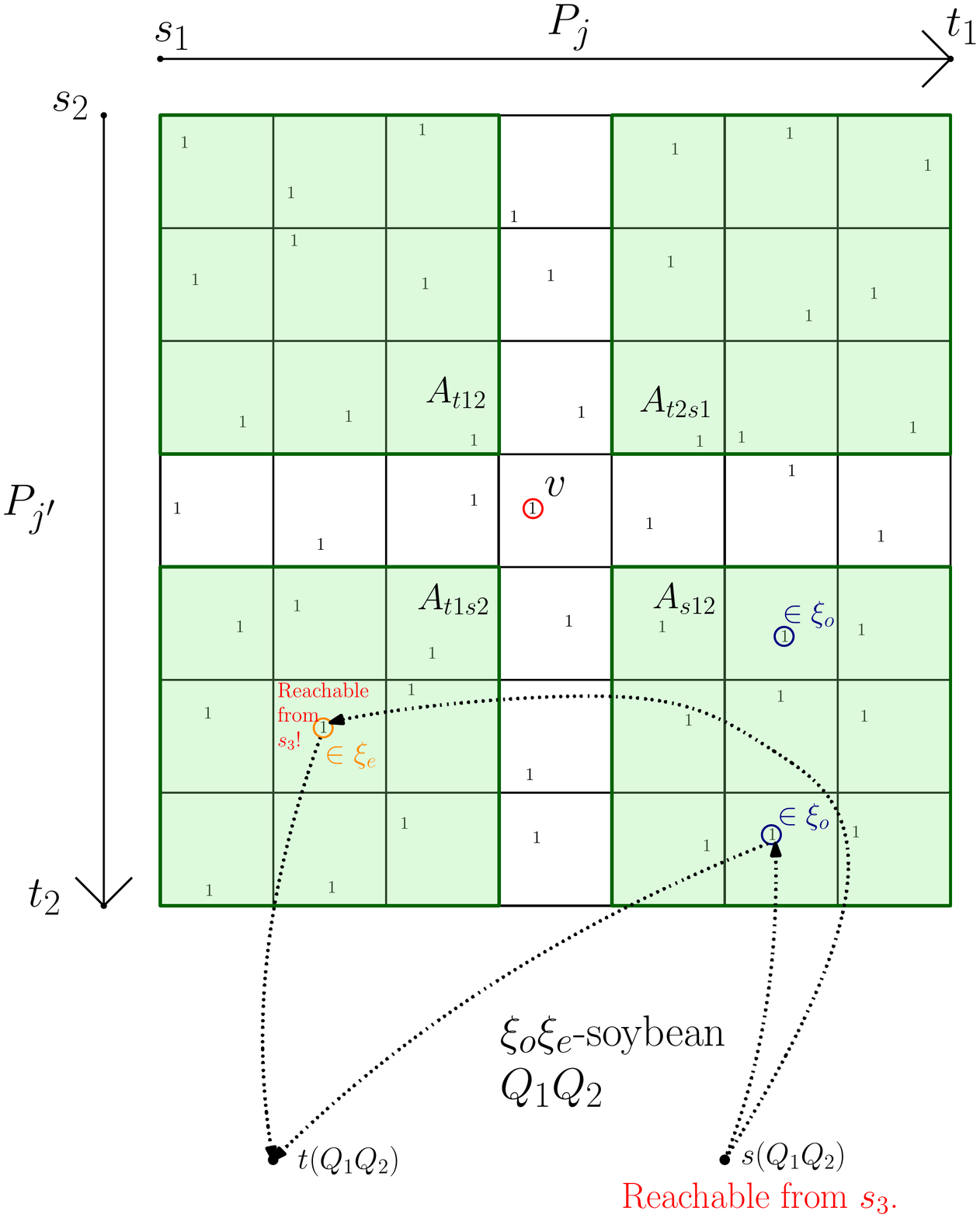}\qquad
  \includegraphics[scale=0.59]{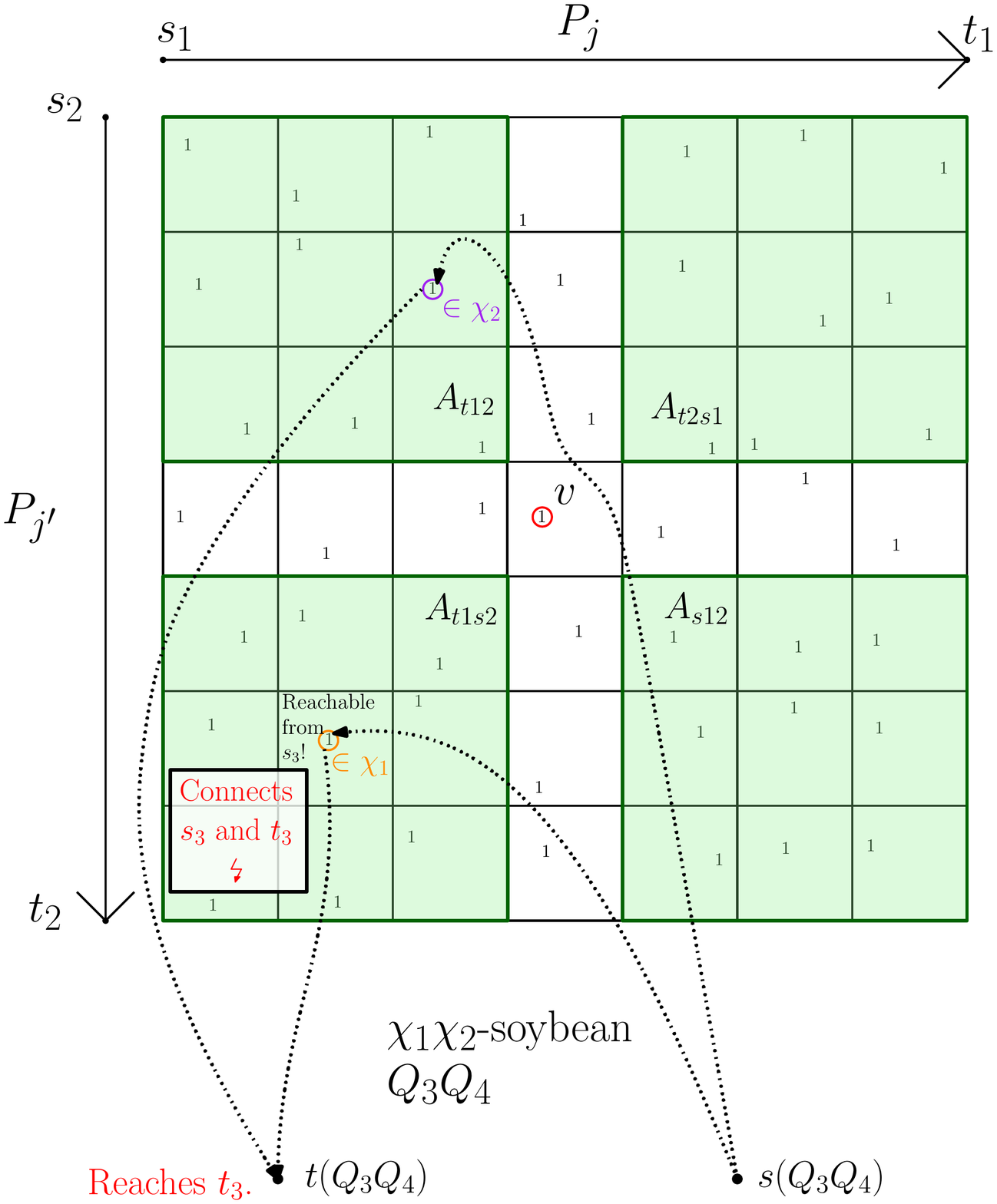}
	\caption{
    The \textbf{left} part shows a basic layout of the matrix $\Adj(\pi^{1, 2}_{j, j'})$. 
Each cell might contain multiple $1$ entries, but we show only the representatives in the picture.
The vertex $v$ is marked by a red circle.
The submatrices $A_{t12}$, $A_{t1s2}$, $A_{t2s1}$, and $A_{s12}$ have a light-green background.
The interlaced sets $\xi_o$ and $\xi_e$ on $P_{j'}$ are marked in blue and orange.
One $\xi_o\xi_e$-soybean is depicted using dotted lines.
If such a soybean does not intersect $S$, we derive the conclusions stated in red.
The \textbf{right} part shows the interlaced sets $\chi_1$ (marked orange) and $\chi_2$ (marked purple) on $P_{j}$.
Note that some filtering steps have to happen in the proof, so only vertices for which we derived the conclusion (Reachable from $s_3$!) are part of $\chi_1$.
One $\chi_1\chi_2$-soybean is depicted using dotted lines.
If such a $\chi_1\chi_2$-soybean does not intersect $S$, we derive the conclusions stated in red and, therefore, a contradiction.
        }\label{fig:tw-red-soybeans}
\end{sidewaysfigure}

Now, we construct another pair of interlaced sets.
Let $\chi_1$ be the set of vertices that are contained in $\xi_e$ and that correspond to the soybeans
selected in the previous step (i.e., to the representatives in $A_{t1s2}'$, except for the last one on the path $P_j$).
Note that by our previous construction, for $(p,q)$-representatives in $\chi_1$, all values of $q$ are pairwise different and of the same parity.
For every $(p,q)$-representative in $\chi_1$, we insert into $\chi_2$ a $(p',q')$-representative of $A_{t12}$ for $q' = q+1$ and any fixed $p \leq \splitP-1$.
It follows that $\chi_1,\chi_2$ are interlaced sets on path $P_{j}$ of size $q(k,k+1)$,
and they are within the same set $B_1 \in \mathcal{B}_1$. (We removed the last representative in $\chi_1$ as it could be the in the $(p,\splitP-2)$-cell of $\Adj(\pi^{1, 2}_{j, j'})$ for some $p\ge\splitP+1$.)
Hence we obtain $k+1$ vertex-disjoint $\chi_1,\chi_2$-soybeans by \cref{thm:dfl-upgrade-vertices}.
At most $k$ of them can intersect $S$.
Therefore let $Q_3Q_4$ be one $\chi_1,\chi_2$-soybean that does not.
We know that $v_3\in Q_3 \cap \chi_1$ is reachable by $s_3$ in $G-S$ by the arguments in the previous paragraph.

We know that $v_4\in Q_4\cap A_{t12}$ does not reach $t_1$ nor $t_2$ in $G-S$.
Hence, $t(Q_3Q_4)$ does not reach $t_1$ nor $t_2$.
Since $S$ is shadowless, $t(Q_3Q_4)$ reaches $t_3$ in $G-S$.
Since $Q_3Q_4$ is disjoint with the solution, $v_3$ reaches $t_3$ in $G-S$.
This is the desired contradiction, see the right part of \cref{fig:tw-red-soybeans} for this case.
\end{proof}

\section{Shadow removal}\label{sec:shadow-removal}
In this section, we prove \cref{thm:shadow-removal}.
Recall that $(G,k, (s_i,t_i)_{i\in [3]}, V^{\infty})$ is an instance of \dmcthree. 
A set $W \subseteq V(G)$ is called \emph{thin} if for every $v \in W$, $v \not \in \rr_G(W\setminus v)$.
For two sets $A,B \subseteq V(G)$ such that $A \cap B = \emptyset$, an \emph{$(A,B)$-separator} is a set $S \subseteq V(G) \setminus (A \cup B)$ such that $G-S$ has no path from any vertex of $A$ to any vertex of $B$.
Let $R^+_G(A)$ be the set of vertices that are reachable from some vertex of $A$ in $G$.
An \emph{$(A,B)$-important separator} is an inclusionwise minimal set $S \subseteq V(G) \setminus (A \cup B)$ such that $G-S$ has no path from $A$ to $B$ and there is no $(A,B)$-separator $S'$ such that $|S'| \leq |S|$ and $R^+_{G - S}(A) \subset R^+_{G-S'}(A)$.
To prove \cref{thm:shadow-removal}, we use the tool of random sampling of important separators from~\cite{marx:multicut,dir-mwc,dsfvs}, presented in its derandomized form and in the form that is convenient for us to use, as \cref{prop:sampling-shadow}.

\begin{proposition}[Theorem~$3.18$ by Chitnis et al.~\cite{dsfvs}]\label{prop:sampling-shadow}
Given an instance $(G,k,(s_i,t_i)_{i \in [3]}, V^{\infty})$,
there is a deterministic algorithm that runs in time $2^{\OO(k^2)} \cdot n^{\OO(1)}$ and outputs 
a family $\mathcal{Z} \subseteq 2^{V(G) \setminus V^{\infty}}$ of size $2^{\OO(k^2)} \log n$ %
such that the following holds.
Let $W \subseteq V(G)$ be a thin set of size at most $k$ and let $Y \subseteq V(G)$ such that for each $v \in Y$, there is an important $(v,T)$-separator $W' \subseteq W$. 
For every such pair $(W,Y)$, there exists $Z \in \mathcal{Z}$ such that $Z \cap W = \emptyset$ and $Y \subseteq Z$.
\end{proposition}

The idea is to use \cref{prop:sampling-shadow} to find a family of sets such that at least one of the sets in this collection contains the reverse and forward shadow of some solution, but does not contain the corresponding solution.
We apply \cref{prop:sampling-shadow} in two steps: first to cover the reverse shadow of a solution and then to also cover its forward shadow.
To apply \cref{prop:sampling-shadow} where the set $W$ corresponds to some solution and the set $Y$ corresponds to the reverse shadow of $W$, one needs to guarantee that there is a solution $W$ and its reverse shadow $Y$ that satisfies the properties of \cref{prop:sampling-shadow}.
In order to prove this, we define \emph{shadow-maximal solutions} (as in~\cite{dir-mwc,dsfvs}).

\begin{definition}[Shadow-maximal solution]
Let $\mathcal{I}=(G,k,(s_i,t_i)_{i \in [3]},V^{\infty})$ be an instance of \threedmc.
An inclusion-wise minimal solution $S$ for $\mathcal{I}$ is called a \emph{shadow-maximal} solution if $\rr_G(S) \cup \ff_G(S) \cup S$ is inclusion-wise maximal among all minimal solutions $S$.

A shadow-maximal solution $S$ is called a \emph{best shadow-maximal} solution, if it is shadow-maximal and amongst all shadow-maximal solutions $|\rr_G(S)|$ is maximum.
\end{definition}

We now show in \cref{lem:thin,lem:shadow-max} that if $W$ is a best shadow-maximal solution and $Y=\rr_G(W)$, then $W$ and $Y$ satisfy the properties of \cref{prop:sampling-shadow}.

\begin{lemma}\label{lem:thin}
	Every minimal solution of the instance $(G,k,(s_i,t_i)_{i \in [3]}, V^{\infty})$ of \threedmc~is thin.
\end{lemma}
\begin{proof}
Let $W$ be a minimal solution of the instance $(G,k,(s_i,t_i)_{i \in [3]}, V^{\infty})$.
We show that if 
$v \in W$ and $v \in \rr_G(W')$ for some $W' \subseteq W$, then $W\setminus v$ is also a solution for $(G,k,(s_i,t_i)_{i \in [3]},V^{\infty})$, contradicting the minimality of $W$.
Towards this, if $W \setminus v$ is not a solution, then there exists an $s_it_i$-path in $G - (W \setminus v)$ that contains $v$.
This implies the existence of a $vt_i$-path, call it $P$, in $G - (W \setminus v)$.
Since $v \in \rr_G(W')$, $P$ must contain a vertex of $W' \subseteq W \setminus v$, which yields a contradiction.
\end{proof} 

\begin{lemma}\label{lem:shadow-max}
Let $(G,k,(s_i,t_i)_{i \in [3]}, V^{\infty})$ be an instance of \threedmc\ and let $S$ be a shadow-maximal solution.
Then, either for every $v \in \rr_G(S)$, there exists $S_v \subseteq S$ such that $S_v$ is a $(v,\{t_1, t_2, t_3\})$-important separator, or there exists another shadow-maximal solution $S'$ such that $\rr_G(S) \subset \rr_G(S')$.
\end{lemma}
\begin{proof}
Fix $v \in \rr_G(S)$.
Then $S$ is a $(v,\{t_1,t_2,t_3\})$-separator in $G$. 
Let $S_v \subseteq S$ be a minimal $(v,\{t_1,t_2,t_3\})$-separator. If $S_v$ is an important $(v,\{t_1,t_2,t_3\})$-separator then we are done.
Otherwise there exists a minimal $(v,\{t_1,t_2,t_3\})$-separator $S'_v$ such that $|S'_v| \leq |S_v|$ and $R^+_{G - S_v}(v) \subset R^+_{G-S'_v}(v)$.
Let $S'=(S \setminus S_v) \cup S'_v$.
Clearly $|S'| \leq |S|$ since $|S'_v| \leq |S_v|$.
Also, since $S'_v \neq S_v$, we have $S \setminus S' \neq \emptyset$.
 We now show that $S'$ is a solution such that $(\rr_G(S) \cup \ff_G(S) \cup S) \subseteq (\rr_G(S') \cup \ff_G(S') \cup S')$ and 
 $\rr_G(S) \subset \rr_G(S')$, which proves the lemma.

 \begin{claim}\label{claim:reverse-shadow}
$\rr_G(S) \cup (S \setminus S') \subseteq \rr_G(S')$. 
 \end{claim}
 \begin{proof}

We first show that if $x \in S \setminus S' (= S_v \setminus S'_v)$, then $x \in \rr_G(S')$. Suppose not (that is $x \not \in \rr_G(S')$). Then there exists an $(x,\{t_1,t_2,t_3\})$-path in $G-S'$. Since $S_v$ is a minimal $(v,\{t_1,t_2,t_3\})$-separator and $x \in S_v$, there exists a $(v, \{t_1,t_2,t_3\})$-path that intersects $S_v$ exactly at $x$. Let $P$ denote the $vx$-subpath of this path. Then $V(P) \setminus x \subseteq R^+_{G-S_v}(v)$. Since $S'_v$ is a $(v,\{t_1,t_2,t_3\})$-important separator and $x \not \in S'_v$, $V(P) \subseteq R^+_{G-S'_v}(v)$. Thus, there is a $vx$-path in $G-S'$. This, together with the $(x,\{t_1,t_2,t_3\})$-path in $G-S'$, implies a $(v,\{t_1, t_2,t_3\})$-path in $G-S'$. Since $S'_v \subseteq S'$ and $S'_v$ is a $(v,\{t_1,t_2,t_3\})$-separator, this is not possible.

Now suppose, for the sake of contradiction, that $x \in \rr_G(S)$ but $x \not \in  \rr_G(S')$. Consider an $(x, \{t_1,t_2,t_3\})$-path $P$ in $G-S'$. Then there exists $y \in V(P)$ such that $y \in S \setminus S'$. From the claim in the previous paragraph, $y \in \rr_G(S')$. Thus, there exists a vertex of $S'$ on the $(y,\{t_1,t_2,t_3\})$-subpath of $P$, which is a contradiction.
 \end{proof}
 
  From \cref{claim:reverse-shadow} and since $|S \setminus S'| \geq 1$, $\rr_G(S) \subset \rr_G(S')$. We now show that $S'$ is a solution for the instance $(G,k,(s_i,t_i)_{i \in [3]},V^{\infty})$.

\begin{claim}\label{claim:sol}
$S'$ is a solution for the instance $(G,k,(s_i,t_i)_{i \in [3]},V^{\infty})$ of \dmcthree.
\end{claim}
\begin{proof}
For the sake of contradiction, say $S'$ is not a solution. Then there exists an $s_it_i$-path in $G-S'$ that uses a vertex $x \in S_v \setminus S'_v$. In particular, there exists an $xt_i$-path in $G-S'$, that is $x \not \in \rr_G(S')$.
 Since $S_v \setminus S'_v = S \setminus S'$, from \cref{claim:reverse-shadow}, $x \in \rr_G(S')$. This is a contradiction.
\end{proof}

\begin{claim}\label{claim:shadow-maximal}
$(\rr_G(S) \cup \ff_G(S) \cup S) \subseteq (\rr_G(S') \cup \ff_G(S') \cup S')$.
\end{claim}
\begin{proof}
From \cref{claim:reverse-shadow}, $(\rr_G(S) \cup S) \subseteq (\rr_G(S') \cap S')$.
We now show that for any $v \in \ff_G(S) \setminus \rr_G(S)$, $v \in \ff(S') \cup \rr_G(S')$.
Suppose not.
Then there exists an $s_iv$-path, say $P_1$, and a $vt_j$-path, say $P_2$, for some $i,j \in [3]$, in $G-S'$.  Since $v \in \ff_G(S)$, the path $P_1$ contains a vertex $x \in S \setminus S'$.
From \cref{claim:reverse-shadow}, $x \in \rr_G(S')$. But the $xv$-subpath of $P_1$, together with the $vt_j$-path $P_2$ gives an $xt_j$-path in $G-S'$, which is a contradiction because $x \in \rr_G(S')$.
\end{proof}

From \cref{claim:sol,claim:reverse-shadow} and since $|S \setminus S'|\geq 1$, we have shown that $S'$ has the properties stated earlier, which contradict that $S$ is a best shadow-maximal solution.
\end{proof}

The following lemma, uses \cref{lem:thin,lem:shadow-max} to show that appropriate applications of \cref{prop:sampling-shadow} result in a family of sets that cover the shadow of some solution.

\begin{lemma}[Covering the shadow]\label{lem:sample-reverse-shadow}
There is an algorithm that given an instance $\mathcal{I}=(G,k, \allowbreak (s_i,t_i)_{i \in [3]}, \allowbreak V^{\infty})$ of \dmcthree,
in time $2^{\OO(k^2)} \cdot n^{\OO(1)}$,
outputs a collection $\mathcal{Z} \subseteq 2^{V(G) \setminus V^{\infty}}$ of size $2^{\OO(k^2)} \log^2 n$, %
such that if $\mathcal{I}$ is a \yes-instance, then there exists a solution $S^*$ such that  
there exists $Z \in \mathcal{Z}$ for which $Z \cap S^* = \emptyset$ and $\rr_G(S^*) \cup \ff_G(S^*) \subseteq Z$.
\end{lemma}
\begin{proof}
Let $\mathcal{X}$ be the family returned by the algorithm of \cref{prop:sampling-shadow} on the instance $\mathcal{I}$.
Suppose $\mathcal{I}$ is a \yes-instance.
Fix a best shadow-maximal solution $S^*$ of $\mathcal{I}$.
From \cref{lem:thin,lem:shadow-max} follows that the pair of sets $(W,Y)=(S^*,\rr_G(S^*)$ satisfies the properties of the sets in \cref{prop:sampling-shadow} on input $\mathcal{I}$.
Thus, there exists $Z^{\rrr} \in \mathcal{X}$ such that $Z^{\rrr} \cap S^* =\emptyset$ and $\rr_G(S^*) \subseteq Z^{\rrr}$.

Let $\vec{G}$ be the graph obtained from $G$ after reversing all the edges of $G$. 
For each $Z \in \mathcal{X}$, create an instance $\mathcal{I}_Z=(\vec{G},k,(t_i,s_i)_{i \in [3]}, V_Z^{\infty})$, where $V_Z^{\infty} = V^{\infty} \cup Z$. 
Let $\mathcal{Y}_Z$ be the family returned by the algorithm of \cref{prop:sampling-shadow} on input $\mathcal{I}_Z$. 
Then output the family $\mathcal{Z} = \{\{Z_1 \cup Z_2 \}: Z_1 \in \mathcal{X}, Z_2 \in \mathcal{Y}_{Z_1}\}$.
We now show that $\mathcal{Z}$ is the desired family.
The size bound and the running time follow from \cref{prop:sampling-shadow}.

Observe that $S^*$ is also a solution for $\mathcal{I}_{Z^{\rrr}}$.
Further observe that for any set $S \subseteq V(G)$, the forward shadow of $S$ in $G$, with respect to $\{s_1,s_2,s_3\}$ is equal to the reverse shadow of $S$ in $\vec{G}$ with respect to $\{t_1, t_2,t_3\}$, and the reverse shadow of $S$ in $G$, with respect to $\{t_1,t_2,t_3\}$ is equal to the forward shadow of $S$ in $\vec{G}$ with respect to $\{s_1, s_2,s_3\}$.
That is, $\rr_G(S) = \ff_{\vec{G}}(S)$ and $\ff_G(S) = \rr_{\vec{G}}(S)$.
In particular, $S^*$ is a shadow-maximal solution of $\mathcal{I}_{Z^{\rrr}}$.
We now want to show that the pair $(W,Y)=(S^*,\rr_{\vec{G}}(S^*))$ satisfies the properties of \cref{prop:sampling-shadow} on input $\mathcal{I}_{Z^{\rrr}}$.
Towards this we prove the following claim.

\begin{claim}\label{claim:forward}
 For each 
$v \in \rr_{\vec{G}}(S^*)$, there exists an $S_v \subseteq S^*$ that is a $(v, \{s_1,s_2,s_3\})$-important separator in $\vec{G}$.
\end{claim}
\begin{proof}
Suppose the claim does not hold.
Then, by \cref{lem:shadow-max}, there exists a shadow-maximal solution $S'$ of $\mathcal{I}_{Z^{\rrr}}$ 
such that $\rr_{\vec{G}}(S^*) \subset \rr_{\vec{G}}(S')$.
Since $S^*$ is also a shadow-maximal solution of $\mathcal{I}_{Z^{\rrr}}$,
we conclude that $(\rr_{\vec{G}}(S^*) \cup \ff_{\vec{G}}(S^*) \cup S^*) = (\rr_{\vec{G}}(S') \cup \ff_{\vec{G}}(S') \cup S')$.
In particular, $S' \setminus S^* \subseteq (\rr_{\vec{G}}(S^*) \cup \ff_{\vec{G}}(S^*) \cup S^*)$.
That is, for any $v \in S^* \setminus S'$, either $v \in \rr_{\vec{G}}(S^*)$ or $v \in \ff_{\vec{G}}(S^*)$. 

Since $\ff_{\vec{G}}(S^*) = \rr_G(S^*) \subseteq Z^{\rrr} \subseteq V_{Z^{\rrr}}^{\infty}$ and $S^*$ is a solution of $\mathcal{I}_{Z^{\rrr}}$, $S^* \cap \ff_{\vec{G}}(S^*) =\emptyset$.
If $v \in \rr_{\vec{G}}(S^*)$, then, since 
$\rr_{\vec{G}}(S^*) \subset \rr_{\vec{G}}(S')$, $v  \in \rr_{\vec{G}}(S')$.
This is a contradiction as $v \in S' \setminus S^*$.
\end{proof}

Thus from \cref{claim:forward} and \cref{lem:thin} follows that the pair $(W,Y)=(S^*, \rr_{\vec{G}}(S^*))$ satisfies the properties of \cref{prop:sampling-shadow} on input $\mathcal{I}_{Z^{\rrr}}$.
Thus, there exists $Z^{\fff} \in \mathcal{Y}_{Z^{\rrr}}$ such that $Z^{\fff} \cap S^* = \emptyset$ and $\ff_G(S^*) \subseteq Z^{\fff}$.
Let $Z^* = Z^{\rrr} \cup Z^{\fff}$.
Then from the above arguments $Z^* \cap S^* = \emptyset$ and $(\rr_G(S^*) \cup \ff_G(S^*)) \subseteq Z^*$. Also $Z^* \in \mathcal{Z}$.
\end{proof}

Finally we use \cref{lem:sample-reverse-shadow} to prove \cref{thm:shadow-removal}.

\ShadowRemovalTheorem*

\begin{proof}%
Given an instance $(G,k,(s_i,t_i)_{i \in [3], V^{\infty}})$, let
$\mathcal{Z}$ be the collection returned by \cref{lem:sample-reverse-shadow}. 
From \cref{lem:sample-reverse-shadow}, if $\mathcal{I}$ is a \yes-instance, then there exists a solution $S$ and $Z \in \mathcal{Z}$ such that $Z \cap S = \emptyset$ and $(\rr_G(S) \cup \ff_G(S)) \subseteq Z$.
Let $G'$ be obtained from $G$ by bypassing $Z$.
From \cref{lem:torso}, $S$ is also a solution of $(G',k,(s_i,t_i)_{i \in [3]}, V^{\infty} \setminus Z)$.
We now show that is $S$ is a shadowless solution of $(G',k,(s_i,t_i)_{i \in [3], V^{\infty}})$, that is $\rr_{G'}(S) = \ff_{G'}(S) =\emptyset$. For the sake of contradiction, say $\rr_{G'}(S) \neq \emptyset$. Let $x \in \rr_{G'}(S)$. In particular, $x \in V(G') = V(G) \setminus Z$. Since $x \in \rr_{G'}(S)$, $G' -S$ has no $(x, \{t_1,t_2,t_3\})$-path. 
From \cref{lem:torso}, $G-S$ also has no $(x, \{t_1,t_2,t_3\})$-path, which implies that $x \in \rr_G(S)$. This is a contradiction, since then $x \in Z$ and hence $x \not \in V(G')$.
 \end{proof}

\section{Harvesting soybeans}\label{sec:soybean}

This section is devoted to the proof of \cref{thm:dfl-upgrade-vertices}. 
The proof revisits the whole proof of flow-augmentation of~\cite{dfl-arxiv} (recalled below as \cref{thm:dir-flow-augmentation})
and extracts the additional information along the way.
Furthermore, we need to slightly revise the behavior of the algorithm in the base
case to ensure the desired properties.
This section assumes that the reader is familiar with the proof of \cref{thm:dir-flow-augmentation} from~\cite{dfl-arxiv}.

\subsection{Back to the edge-deletion regime}

Since flow-augmentation in~\cite{dfl-arxiv} is defined on edge-cuts,
but \cref{thm:dfl-upgrade-vertices} is in the vertex-deletion regime, we first
go back to the edge-deletion regime. 
We need a few more definitions that closely follow~\cite{dfl-arxiv}.

In the edge-deletion regime, the edges of $G$ can be deletable or undeletable.
An (edge-based) $st$-flow, in this context, is a collection of $st$-paths that do not share a 
deletable edge and an $st$-cut is a set of deletable edges that intersects all $st$-paths.
It is always clear from the context whether we speak about the vertex- or the edge-deletion
regime, and hence we reuse names like $st$-flow or $st$-mincut
or the notation $\lambda_G(s,t)$ for both regimes.

An $st$-cut $Z$ is a \emph{star $st$-cut} if for every $(u,v) \in Z$ in the graph $G-Z$ there is a path from $s$ to $u$ but there is no path from $s$ to $v$;
observe that every minimal $st$-cut is a star $st$-cut.
For a star $st$-cut $Z$ in $G$, by $\corecutG{Z}{G} \subseteq Z$ we denote the set of arcs $(u,v) \in Z$ such that there exists a path from $v$ to $t$ in $G-Z$. 

Being compatible is slightly more complicated for star cuts:
a set of arcs $A \subseteq V(G) \times V(G)$ is \emph{compatible} with a star $st$-cut $Z$
if the set of vertices reachable from $s$ in $G-Z$ and $(G+A)-Z$ is the same.
We also need the notion of a witnessing flow:
for a star $st$-cut $Z$ in $G$, if $\corecut{Z}$ is an $st$-mincut, 
then an $st$-maxflow $\flow$ is a \emph{witnessing flow} if $E(\flow) \cap Z = \corecut{Z}$.
The pair $(A, \flow)$ is \emph{compatible} with $Z$ if $A$ is compatible with $Z$, $\corecutG{Z}{G+A}$ is an $st$-mincut in $G + A$, and $\flow$ is a witnessing flow for $Z$ in $G + A$.
Note that if $Z$ is a minimal $st$-cut, and $A$ is compatible with $Z$, then any $st$-maxflow
in $G+A$ is a witnessing flow for $Z$, but if $Z$ is a star $st$-cut, the notion
of a witnessing flow is more intricate.

The flow-augmentation technique is represented by the following statement. 
\begin{theorem}[Kim et al.~\cite{dfl-arxiv}]\label{thm:dir-flow-augmentation}
There exists a polynomial-time algorithm that, given a directed graph $G$,
vertices $s,t \in V(G)$, and an integer $k$,
returns a set $A \subseteq V(G) \times V(G)$
and an $st$-maxflow $\flow$ in $G+A$ such that 
for every star $st$-cut $Z$ of size at most $k$,
with probability $2^{-\Oh(k^4 \log k)}$ the pair $(A,\flow)$ is compatible with $Z$.
\end{theorem}

Most of this section is devoted to the proof of the following edge-deletion
variant of \cref{thm:dfl-upgrade-vertices}.

\begin{theorem}\label{thm:dfl-upgrade}
  There exist computable functions $c : \N \to \N$ and $q : \N \times \N \to \N$ such that the following holds.

  There exists a polynomial-time randomized algorithm that, given a directed graph $G$,
  vertices $s,t \in V(G)$, and an integer $k$,
  returns a set $A \subseteq V(G) \times V(G)$
  and an $st$-maxflow $\witnessflow$ in $G+A$ such that 
  for every star $st$-cut $Z$ of size at most $k$,
  with probability $2^{-\Oh(k^4 \log k)}$, the pair $(A,\witnessflow)$ is compatible with $Z$.
  
  Additionally, the algorithm returns a partition $\mathcal{B}$ of the deletable
  edges of $\bigcup_{P \in \witnessflow} E(P)$ into at most $c(k)$ sets
  such that for every $P \in \witnessflow$, every integer $p \in \N$, every $B \in \mathcal{B}$
  and every two disjoint sets $C,D$ of size at least $q(k,p)$,
  consisting of edges of $B \cap E(P)$ that
  are interlaced on $P$, the graph $G$ contains a family of $p$ pairwise vertex-disjoint $CD$-soybeans.
  
  Finally, one can take $c$ and $q$ such that 
  $c(k) = 2^{\Oh(k^3 \log k)}$ and $q(k,p) = 2^{\Oh(k^3 \log (kp))}$.
\end{theorem}

Observe that if $P \cap A = \emptyset$, that is, $P$ does not contain any augmentation edge, then
one can simply enumerate $C = \{c_1,\ldots,c_q\}$ and $D = \{d_1,\ldots,d_q\}$ for $q = q(k,p)$
along the path $P$, define $P_i$ to be the subpath of $P$ from $c_i$ to $d_i$, and use
soybeans $\{P_{2i}P_{2i}~|~1 \leq i \leq q/2\}$. (We skip every second such soybean in order to make them vertex-disjoint.) However, when $P$ contains augmentation edges, the situation
is more complex, as a soybean cannot use them.

We now formally show that \cref{thm:dfl-upgrade} implies \cref{thm:dfl-upgrade-vertices}, restated below:

\VertexSoybeanTheorem*

\begin{proof}
  We construct a new graph $G'$ as follows.
  For each vertex $v \in V(G)$, graph~$G'$ contains two vertices $v_1$ and $v_2$ and a deletable arc $(v_1, v_2)$; we call these arcs \emph{important}.
  If $v$ is undeletable or if $v = s$ or $v = t$, then the arc $(v_1,v_2)$ is undeletable instead.
  For each arc $(u, v) \in A(G)$, graph $G'$ contains an undeletable arc $(u_2, v_1)$.
  
  We apply \cref{thm:dfl-upgrade} to $G'$, $s_2$, $t_1$, and $k$, obtaining an arc set $A'$, an edge-based $s_2t_1$-maxflow $\mathcal{P}'$ in $G' + A'$, and a partition $\mathcal{B}'$
  of deletable arcs of $\mathcal{P}'$.

  We now perform the following cleanup step of $A'$ and $\mathcal{P}'$.
  First, for every $x,y \in V(G')$ such there is a path from $x$ to $y$
  in $G'+A'$ that uses undeletable arcs only (recall that every arc of $A'$ is considered
  undeletable), we add an undeletable arc $(x,y)$ to $A'$
  if it is not present already in $G'+A'$.
  Second, for every $P' \in \mathcal{P}'$ and every maximal subpath $Q$ of $P'$
  that uses undeletable arcs only and is of length at least $2$,
  we replace $Q$ with only an arc $(s(Q),t(Q))$
  (which is now in $G'+A'$ due to the previous step). 
  Finally, we restrict $A'$ to contain only arcs that are present on
  some flow path in $\mathcal{P}'$. 
  
  Observe that the first step does not change the space of $s_2t_1$-cuts in $G'+A'$
  (as any newly added arc $(x,y)$ cannot lead from the $s$-side to the $t$-side
   of any $s_2t_1$-cut in $G'+A'$ due to the assumed path of undeletable edges
   from $x$ to $y$)
  while $\mathcal{P}'$ remains an $s_2t_1$-maxflow in $G'+A'$ throughout the process
  (as we only reroute the paths through undeletable arcs)
  and every path $P' \in \mathcal{P}'$ still visits the same deletable arcs, in the same
  order (so the properties of $\mathcal{B}'$ are unharmed).
  The deletion of arcs in $A'$ in the last step only extends the space
  of $s_2t_1$-cuts in $G'+A'$, so if $(A',\mathcal{P}')$ was compatible
  with some $s_2t_1$-cut $Z$ before the process, then it is also compatible after the process.
  Furthermore, as the only deletable arcs in $G'+A'$ are the important arcs, 
  in the end we have the following property: every $(x,y) \in A'$
  is of the form $(x,y) = (u_2,v_1)$ for some $u,v \in V(G)$. 

  The above discussion
  implies that we can obtain the desired result $(A,\mathcal{P},\mathcal{B})$ in the natural manner:
  for every $(x,y)=(u_2,v_1) \in A'$, we add $(u,v)$ to $A$,
  for every $P' \in \mathcal{P}'$ we add to $\mathcal{P}$ a path $P$ being the 
  path $P'$ with all important edges on it contracted, 
  and for every $B' \in \mathcal{B}'$ we add to $\mathcal{B}$
  the set $B = \{v~|~(v_1,v_2) \in B'\}$. 
\end{proof} 

Thus, it remains to prove \cref{thm:dfl-upgrade}. This proof spans the rest of this section.

\subsection{Initial setup}

The algorithm of \cref{thm:dir-flow-augmentation} first filters out some 
trivial cases in which $A = \emptyset$ can be returned.
In these cases one can take $\mathcal{B}$ to be a singleton and construct the desired soybeans as in the comment after the statement of \cref{thm:dfl-upgrade}.
After treating these trivial cases, the algorithm of \cref{thm:dir-flow-augmentation} invokes a recursive subroutine.

The recursion has depth strictly less than $h_{\mathrm{max}}(k) = \Oh(k^3)$.
The input to the recursive call consists of a graph $G$ with distinguished vertices
$s,t \in V(G)$, an integer $k$, an $st$-flow $\flow$, and an integer $\kappa$. 
The output is a set $A \subseteq V(G) \times V(G)$ and an $st$-maxflow $\witnessflow$
such that $\lambda_{G+A}(s,t) \geq \kappa$ and for every star $st$-cut $Z$ with $|Z| \leq k$
and $|\corecutG{Z}{}| \geq \kappa$, $(A,\witnessflow)$ is compatible with $Z$ with high probability.

An important insight about the structure of the recursion is that for every recursive call
$\rho$ on a graph $G$ the following holds: For every subcall on a graph $G'$, $G'-\{s,t\}$
is a subgraph of $G-\{s,t\}$.
Furthermore, the graphs $G'-\{s,t\}$ over all recursive subcalls are pairwise vertex-disjoint (as subgraphs of $G$).

\paragraph{Constructing $\mathcal{B}$.}
Let $(A,\witnessflow)$ be the output of the algorithm of \cref{thm:dir-flow-augmentation}
and consider $P \in \witnessflow$. (We separately partition the deletable edges of $E(P)$
into sets of $\mathcal{B}$, so that every $B \in \mathcal{B}$ is contained in $E(P)$
for one $P \in \witnessflow$.)
If present, we put at most two deletable edges of $E(P)$ that are incident with $s$ or $t$
into separate singleton sets of $\mathcal{B}$ and do not worry about them further.

For every vertex $v \in V(G)$, we can consider all recursive calls on graphs $G'$
that contain $v$. By the previous observation on the structure of the recursion,
these calls form an upward path in the recursion tree.
For every call $\rho'$ in the recursion whose parent $\rho$ 
corresponds to the ``large $\ell$'' case, we assign
to $e$ a \emph{local signature} as follows. 
The recursive call is applied to a graph $G^\alpha$ for some $\alpha$ such that
$G^\alpha-\{s,t\}$ is a subgraph of $G$ between the $st$-mincuts $C_{b_\leftarrow^\alpha-1}$
and $C_{b_\rightarrow^\alpha+\lambda}$, the edges incident to $s$ and $t$ in $G^\alpha$
correspond to the edges of the two said mincuts, and every vertex in $G^\alpha$ is reachable
from $s$ and reaches $t$.
Pick a path $Q_1$ from $s$ to $v$ in $G^\alpha$ and let $i_1$
be the index of the path of the flow $\flow$ in the parent call $\rho$ that contains the edge
of $C_{b_\leftarrow^\alpha-1}$ corresponding to the first edge of $Q_1$. 
Pick a path $Q_2$ from $v$ to $t$ in $G^\alpha$ and let $i_2$
be the index of the path of the flow $\flow$ in the parent call $\rho$ that contains the edge
of $C_{b_\rightarrow^\alpha+\lambda}$ corresponding to the last edge of $Q_2$. 
Then $(i_1,i_2)$ is the local signature of $v$ at $\rho'$. 
The signature of $v$ is the sequence of all local signatures of $v$, in the top-to-bottom order
in the recursion tree.
The signature of an edge $e$ is the pair consisting of the signatures of its endpoints.
Finally, we define $\mathcal{B}$ as the partition of the deletable edges of $\witnessflow$
according to the path of $\witnessflow$ they belong to and according to their signatures.
Since there are $2^{\Oh(h_{\mathrm{max}}(k) \log k)} = 2^{\Oh(k^3 \log k)}$ signatures,
$|\mathcal{B}| \leq 2^{\Oh(k^3 \log k)}$. 

\medskip

Let $f_p(x) = 1+2px$, $q_0(p) = 2p$, $q_{i+1}(p) = f_p(q_i(p))$ for $0 \leq i < h_{\mathrm{max}}(k)$, 
and $q(k,p) = q_{h_{\mathrm{max}}(k)}(p)$; recall that $h_{\mathrm{max}}(k) = \Oh(k^3)$ is the maximum possible depth of the recursion for a fixed value of $k$.
We have $q(k,p) = 2^{\Oh(k^3 \log (kp))}$.
We prove, by the bottom-to-top induction over the recursion tree, that 
if the recursive call $\rho$ at depth $i$ applied to a tuple $(G,s,t,k,\flow,\kappa)$ returned
$(A,\witnessflow)$, $P \in \witnessflow$, and $C$ and $D$ are two disjoint sets
of size $q_{h-i}(p)$ consisting of deletable edges of $E(P)$ of the same signature
that are interlaced on $P$, then $G-\{s,t\}$ contains a family of
$p$ pairwise vertex-disjoint $CD$-soybeans.
Here $h \leq h_{\mathrm{max}}(k)$ is the actual depth of the recursion.

In the subsequent paragraphs we consider different cases the algorithm of \cref{thm:dir-flow-augmentation} can enter and in each of them prove the desired claim.

\subsection{Initial steps}

We first investigate the initial preprocessing steps.

These can be split into two types. The first type are leaves of the recursion: when $\lambda_{G}(s,t) = 0$ and when $\lambda > k$. In the first case, the algorithm returns $A = \emptyset$ and $\witnessflow=\emptyset$, so there is nothing to prove.
In the second case, the algorithm returns $A = \{(s,t)\}$ and $\witnessflow$ consisting
of a single path $P$ along the edge of $A$, and again there is nothing to prove.

The second type of steps
invoke one recursive call $\rho'$ on a modified graph $G'$, obtaining $(A',\witnessflow')$.
Recall that in all cases, $G'-\{s,t\}$ is a subgraph of $G-\{s,t\}$.
The algorithm returns $(A,\witnessflow)$ that is constructed from $(A',\witnessflow')$ by setting
$A$ to be $A'$ plus at most $2\lambda+1 \leq 2k+1$ additional edges, all incident with $s$ or $t$,
and $\witnessflow$
to be $\witnessflow'$ with possibly one additional one-edge path. 
The claim is again straightforward as $q_{h-i-1}(k,p) \leq q_{h-i}(k,p)$ and the 
requested soybeans cannot use vertices $s$ nor $t$.

\subsection{Base case}

In the base case of the algorithm, $\flow$ is an $st$-maxflow 
and for every $i \neq j$, $i,j \in [\lambda]$,
there is no path from $V(P_i)$ to $V(P_j)$ in $G-\{s,t\}$. 
This is a place where we need to slightly modify the behavior of the algorithm of
\cref{thm:dir-flow-augmentation}.

Let $B$ be the set of all bottleneck edges. 
For every $i \in [\lambda]$, let $(u_{i,1},v_{i,1}),\ldots,(u_{a_i},v_{a_i})$ be the 
bottleneck edges on $P_i$, in the order along $P_i$. Denote $v_{i,0} = s$ and $u_{i,a_i+1} = t$.
For $i \in [\lambda]$ and $0 \leq b \leq a_i$, let $G_{i,b}$ be the subgraph of $G$ induced
by all vertices that are reachable from $v_{i,b}$ and from where $u_{i,b+1}$ is reachable
in $G-B$. Let $G_i = \{(u_{i,b},v_{i,b}~|~1 \leq b \leq a_i\} \cup \bigcup_{b=0}^{a_u} G_{i,b}$. 
Note that the graphs $G_i$ intersect only in vertices $s$ and $t$
and $\lambda_{G_i}(s,t) = 1$.

For a star $st$-cut $Z$ of size at most $k$, let $Z_i = Z \cap E(G_i)$. 
The analysis of the base case shows that $Z_i$ is a star $st$-cut in $G_i$
and $\corecutG{Z}{G} = \bigcup_{i \in [\lambda]} \corecutG{Z_i}{G_i}$. 
Furthermore, either $\corecutG{Z_i}{G_i}$ consists of a single bottleneck edge
and $Z_i$ contains no other bottleneck edge on $P_i$, or $Z_i$ does not contain
any bottleneck edge of $P_i$ at all. 

For every $i \in [\lambda]$, we proceed as follows.
We guess integers $1 \leq \kappa_i \leq k_i \leq k$
such that $\sum_{i \in [\lambda]} \kappa_i \geq \kappa$ and $\sum_{i \in [\lambda]} k_i \leq k$.
We aim at $k_i = |Z_i|$ and $\kappa_i = |\corecutG{Z_i}{G_i}|$. This happens in total with 
probability $2^{-\Oh(k \log k)}$. 

If $\kappa_i = 1$, we aim at capturing star $st$-cuts $Z_i$ with $\corecutG{Z_i}{G_i}$
consisting of a single bottleneck edge. 
We set $A_i = \{(v_{i,b}, u_{i,b+1})~|~0 \leq b \leq a_i\}$ and $P_i'$ to be a path consisting of edges
$(s,u_{i,1}), (u_{i,1},v_{i,1}), (v_{i,1},u_{i,2}), \ldots, (u_{i,a_i}, v_{i,a_i}), (v_{i,a_i},t)$. 
The algorithm returns $A_i$ as part of the set $A$ and $P_i'$ as one of the flow
paths in $\witnessflow$. 

For our desired claim, observe that the deletable edges on $P_i'$ are only bottleneck edges.
Thus, given $C$ and $D$ interlaced on $P_i'$, each of size at least $q_0(k,p) = 2p$, one can
construct the desired soybeans as follows:
If $c_1,c_2,\ldots,c_{2p}$ and $d_1,d_2,\ldots,d_{2p}$ are the first $2p$ edges of $C$ and $D$, respectively, and $P_{i,j}$ is the subpath of $P_i$
from $c_{2j}$ to $d_{2j}$, then $\{P_{i,j}P_{i,j}~|~j \in [p]\}$ is the desired family of soybeans.
So we can put the whole $A_i$ as a single set in $\mathcal{B}$.

If $\kappa_i > 1$, we aim at capturing star $st$-cuts $Z_i$ that do not contain any bottleneck
edge on $P_i$. Let $A_{i}^\circ$ be the set of copies of all bottleneck edges
on $P_i$. 
For every $0 \leq b \leq a_i$, we recurse on $G_{i,b}$ with 
$v_{i,b}$ playing the role of $s$ and $u_{i,b+1}$ playing the role of $t$, 
parameters $k_i$, $\kappa_i$, and a flow consisting of a single flow path $P_i$ from $v_{i,b}$
to $u_{i,b+1}$. 
Let $(A_{i,b},\witnessflow_{i,b})$ be the returned pair.
The returned set $A$ consists of, for every $i \in [\lambda]$,
 the set $A_i = A_{i}^\circ \cup \bigcup_{b=0}^{a_i} A_{i,b}$.
The returned set $\witnessflow$ consists of, for every $i \in [\lambda]$, 
$\kappa_i$ flow paths, combined from
flow paths $\witnessflow_{i,b}$ (recall that each $\witnessflow_{i,b}$ is of size at least $\kappa_i$) concatenated using edges of $A_{i}^\circ$. 

For our desired claim, consider a returned path $P \in \witnessflow$ and interlaced sets
$C,D$ of $q_{h-i}(k,p)$ deletable edges on $P$. We have two cases.
First, there exists an integer $0 \leq b \leq a_i$ such that $G_{i,b}$ contains
at least $q_{h-i-1}(k,p)+1$ edges of $C$.
Then, the path of $\witnessflow_{i,b}$ used
to construct $P$ contains interlaced subsets of $C$ and $D$ of size $q_{h-i-1}(k,p)$.
The claim follows from the inductive hypothesis for the recursive call on $G_{i,b}$.

In the second case, there are at least $2p+1$ indices $b$ such that
$G_{i,b}$ contains an edge of $C$.
Consequently, there are indices $b_1 \leq b_1' < b_2 \leq b_2' < b_3 \leq b_3' < \ldots < b_p \leq b_p'$ such that for every $j \in [p]$, $G_{i,b_j}$ contains an edge of $C$ and $G_{i,b_j'}$ contains an edge of $D$. 
For $j \in [p]$, let $Q_j$ be a path from $v_{j,b_j}$ to $u_{j,b_j'+1}$ containing an
edge of $C$ from $G_{i,b_j}$ and let $Q_j'$ be a path from $v_{j,b_j}$ to $u_{j,b_j'+1}$ containing
an edge of $D$ from $G_{i,b_j'}$. Such paths exist by the construction of the graphs $G_{i,b}$.
Then, $\{Q_jQ_j'~|~j\in [p]\}$ is the desired soybean harvest.

\subsection{Small \boldmath$\ell$ case}

In the small $\ell$ case the situation is very similar to the second type of initial steps.

The algorithm always invokes one recursive call, on a graph $G'$
such that $G'-\{s,t\}$ is a subgraph of $G-\{s,t\}$, obtaining $(A',\witnessflow')$.
The returned set $A$ consists of $A'$ and additional edges $A_0$ with $|A_0| \leq 4\lambda \ell^{\mathrm{big}} + 2 \leq 16k^3+14$, all incident with $s$ or $t$.
The returned flow $\witnessflow$ consists of the paths $\witnessflow'$, possibly with an edge
of $A_0$ added at the beginning or end, and possibly one additional path that contains
at most one deletable edge.

Thus, the claim follows directly from the inductive hypothesis for the recursive subcall.

\subsection{Large \boldmath$\ell$ case}

In the large $\ell$ case the situation is quite similar to the base case, but a bit more complex.
Let $(G,s,t,k,\flow,\kappa)$ be the input to the recursive call in question.

The algorithm recurses on graphs $G^\alpha$ for all excellent indices $\alpha$, obtaining
pairs $(A^\alpha,\witnessflow^\alpha)$
Observe that in the returned flow $\witnessflow$, the only deletable edges are 
those in graphs $G^\alpha$ on paths $\witnessflow^\alpha$. 
Fix $P \in \witnessflow$; the path $P$ consists of edges of $A$ and some flow paths
from flows $\witnessflow^\alpha$.
Assume that we have interlaced sets $C,D$ of deletable
edges on $P$ of the same signature, each of size $q_{h-i}(k,p)$. 

As in the base case, there are two cases. First, there exists $\alpha$ such that
$C$ contains at least $q_{h-i-1}(k,p)+1$ edges in $G^\alpha$. Then, 
the flowpath $P' \in \witnessflow^\alpha$ contained in $P$ contains interlaced subsets $C',D'$
of $C$ and $D$ of size $q_{h-i-1}(k,p)$ each. The claim follows from the inductive hypothesis.

In the second case, there are at least $2p+1$ indices $\alpha$ for which there is an edge
of $C$ in $G^\alpha$. 
Thus, there are indices $\alpha_1 \leq \alpha_1' < \alpha_2 \leq \alpha_2' < \alpha_3 \leq \alpha_3' < \ldots < \alpha_p \leq \alpha_p'$ such that for every $j \in [p]$, 
there is an edge $c_j$ of $C$ in $G^{\alpha_j}$ and an edge $d_j$ of $D$ in $G^{\alpha_j'}$. 

Recall that all edges of $C$ and $D$ are of the same signature. That is, there are two paths
$P^1,P^2 \in \flow$ such that, for edge $e \in C \cup D$,
  if $e$ lies in $G^\alpha$ then 
there is a walk $W(e)$ in $G$ from the head of the edge of $E(P^1) \cap C_{b_{\leftarrow}^{\alpha}-1}$
to the tail of the edge of $E(P^2) \cap C_{b_{\rightarrow}^\alpha+\lambda}$ that contains $e$
and is completely contained between $C_{b_\leftarrow^\alpha-1}$ and $C_{b_\rightarrow^\alpha+\lambda}$. 
For every $j \in [p]$, let $W_j$ be a concatenation of $W(c_j)$ 
and a subpath of $P^2$ from the tail of the edge of
$E(P^2) \cap C_{b_{\rightarrow}^{\alpha_j}+\lambda}$ to the tail of the edge
of 
$E(P^2) \cap C_{b_{\rightarrow}^{\alpha_j'}+\lambda}$
and let $W_j'$ be a concatenation of
and a subpath of $P^1$ from the head of the edge of
$E(P^1) \cap C_{b_{\leftarrow}^{\alpha_j}-1}$ to the head of the edge
of 
$E(P^1) \cap C_{b_{\leftarrow}^{\alpha_j'}-1}$ and $W(d_j)$.
(Recall that $\alpha_j \leq \alpha_j'$.)
Then, $\{W_jW_j'~|~j \in [p]\}$ is the desired soybean family.

\newcommand{\psifull}{\textsc{Partitioned Subgraph Isomorphism}}

\newcommand{\psismall}{\textsc{PSI}}

\newcommand{\wt}{\texttt{wt}}
\newcommand{\ww}{W}

\newcommand{\wtdmcfull}{ \textsc{Weighted Directed Multicut} }

\newcommand{\wttwodmc}{\textsc{$2$-Wt-DMC}}

\newcommand{\degree}{\textnormal{\texttt{deg}}}

\newcommand{\XX}{\texttt{X}}

\newcommand{\YY}{\texttt{Y}}

\newcommand{\ZZ}{\texttt{Z}}

\section{Two-terminal-pair \textsc{Weighted Directed Multicut} is \W[1]-hard}\label{sec:hard}

In the \wtdmcfull\ problem, the input is a directed graph $G$, a set of terminal pairs $\{(s_i,t_i) : i \in [p]\}$, a weight function on the vertex set $\wt : V(G) \to \N$ and positive integers $k$ and $\ww$.
For a subset $S \subseteq V(G)$ we define $\wt(S) \coloneqq  \sum_{v \in S} \wt(v)$.
The goal is to determine whether there exists a set $S \subseteq V(G)$ such that $|S| \leq k$, $\wt(S) \leq \ww$ and $G-S$ has no $s_it_i$-path for each $i \in [p]$.
In this section, we show that \wtdmcfull\ is \W[1]-hard parameterized by $k$, even with two terminal-pairs (that is when $p=2$).
In fact, we show that it does not admit an $f(k) \cdot n^{o(k/ \log k)}$ algorithm under the \ETH.
We denote this problem with two terminal pairs by \wttwodmc.
The hardness proof we provide is essentially a simplification of the reduction given by Pilipczuk and Wahlstr\"om~\cite{PilipczukW18a} for proving the \W[1]-hardness of the \textsc{Directed Multicut} problem with four terminal-pairs.
Our reduction essentially demonstrates that the synchronization of some gadgets achieved in the reduction in~\cite{PilipczukW18a} using two additional terminal-pairs can also be achieved if the vertices are allowed polynomial (in the input size) weights.
This helps us to eliminate two terminal-pairs in the reduction of~\cite{PilipczukW18a} at the cost of adding polynomial weights.

\begin{theorem}\label{thm:weighted-hardness}
\textsc{Weighted Directed Multicut} is \W[1]-hard even for two terminal pairs. Furthermore, assuming the \ETH, the problem cannot be solved in $f(k) \cdot n^{o(k / \log k)}$ time, where $n$ is the number of vertices of the input graph.
\end{theorem}

To prove \cref{thm:weighted-hardness}, we give a reduction from \psifull\ (\psismall), parameterized by the number of edges in the pattern graph.
In the \psifull\ problem, given two undirected graphs $G,H$ such that $V(G) = \biguplus_{i \in V(H)} V^i$, the goal is to determine if there exists a homomorphism $\xi : V(H) \to V(G)$ such that $\xi(i) \in V^i$ for each $i \in V(H)$.
This problem has been shown to be \W[1]-hard parameterized by $|E(H)|$ by Marx~\cite[Corollary~$6.3$]{marx2007can}.
 In fact, the authors show that there is no $f(k) \cdot n^{o(k/ \log k)}$ algorithm for \psismall, where $k=|E(H)|$ and $n$ is the number of vertices in the input graph, unless the \ETH\ fails.

\begin{proof}[Proof of \cref{thm:weighted-hardness}]
Let $(G,H)$ be an instance of \psismall\ where $V(G) = \biguplus_{i\in V(H)} V^i$.
Without loss of generality, let $V(H)=\{1, \ldots,h\}$, let $|V^i| = |V^j|=n$ for each $i,j \in V(H)$ and assume that $H$ has no isolated vertices. 
Let $V^i=\{v^i_1, \ldots, v^i_n\}$. %
Let $k=|E(H)|$, then $|V(H)| \leq 2k$ since $H$ has no isolated vertices. 
We now construct an instance $(D,(s_i,t_i)_{i \in [2]},\wt,k',\ww)$ of \wttwodmc.
Set $k' = 5k +h$ and $\ww = M (2k (n+1) +h) +k$, where $M=k+1$.

\paragraph*{Construction of $D$:}

\newcommand{\xmin}{1}
\newcommand{\xmax}{15}
\newcommand{\ymin}{1}
\newcommand{\ymax}{\ymin +15}
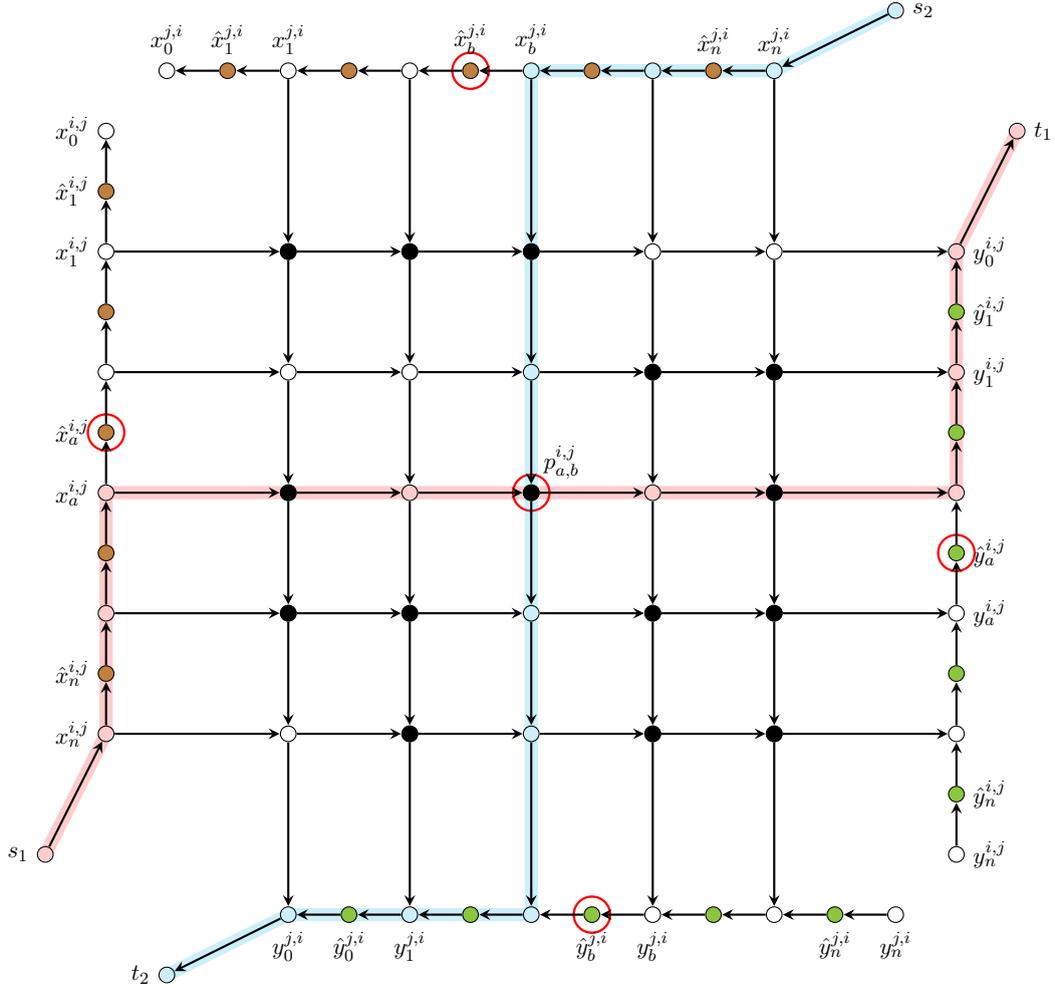
\begin{figure}
\centering
\begin{tikzpicture}[%
scale=.8,
transform shape,
myvertex/.style={circ,scale=2}
]

		\colorlet{myGreen}{green!50!black}
\definecolor{myRed}{rgb}{0.68, 0.05, 0.0}
\colorlet{myBlue}{blue!90!black}
\definecolor{myLightblue}{rgb}{0.54, 0.81, 0.94}
\colorlet{myViolet}{myBlue!55!myRed}
\colorlet{myOrange}{yellow!55!myRed}

\pgfdeclarelayer{background}
\pgfdeclarelayer{foreground}
\pgfsetlayers{background,main,foreground}

\node[myvertex,fill=none,label=above:{$x^{j,i}_0$}] (315) at (\ymin+2,\xmax) {};
\node[myvertex, fill=none,label=above:{$x^{j,i}_1$}] (515) at (\ymin+4,\xmax) {};
\node[myvertex, fill=none] (715) at (\ymin+6,\xmax) {};
\node[myvertex, fill=none,label=above:{$x^{j,i}_b$}] (915) at (\ymin+8,\xmax) {};
\node[myvertex, fill=none] (1115) at (\ymin+10,\xmax) {};
\node[myvertex, fill=none,label=above:{$x^{j,i}_n$}] (1315) at (\ymin+12,\xmax) {};

\node[myvertex, fill=brown,label=above:{$\hat{x}^{j,i}_1$}] (415) at (\ymin+3,\xmax) {};
\node[myvertex, fill=brown] (615) at (\ymin+5,\xmax) {};
\node[myvertex, fill=brown,label=above:{$\hat{x}^{j,i}_b$}] (815) at (\ymin+7,\xmax) {};
\node[myvertex, fill=brown] (1015) at (\ymin+9,\xmax) {};
\node[myvertex, fill=brown,label=above:{$\hat{x}^{j,i}_n$}] (1215) at (\ymin+11,\xmax) {};

\foreach \i in {3,...,12}
{
        		{\pgfmathtruncatemacro{\y}{\i+1}
        		\pgfmathtruncatemacro{\z}{\xmax}
                 \draw[<-,>=stealth,thick] (\i\z) -- (\y\z);}
}

\node[myvertex, fill=none,label=left:{$x^{i,j}_0$}] (214) at (\ymin+1,\xmax-1) {};
\node[myvertex, fill=none,label=left:{$x^{i,j}_1$}] (212) at (\ymin+1,\xmax-3) {};
\node[myvertex, fill=none] (210) at (\ymin+1,\xmax-5) {};
\node[myvertex, fill=none,label=left:{$x^{i,j}_a$}] (28) at (\ymin+1,\xmax-7) {};
\node[myvertex, fill=none] (26) at (\ymin+1,\xmax-9) {};
\node[myvertex, fill=none,label=left:{$x^{i,j}_n$}] (24)  at (\ymin+1,\xmax-11) {};

\node[myvertex, fill=brown,label=left:{$\hat{x}^{i,j}_1$}] (213)  at (\ymin+1,\xmax-2) {};
\node[myvertex, fill=brown] (211) at (\ymin+1,\xmax-4) {};
\node[myvertex, fill=brown,label=left:{$\hat{x}^{i,j}_a$}] (29)  at (\ymin+1,\xmax-6) {};
\node[myvertex, fill=brown] (27) at (\ymin+1,\xmax-8) {};
\node[myvertex, fill=brown,label=left:{$\hat{x}^{i,j}_n$}] (25) at (\ymin+1,\xmax-10) {};

\foreach \i in {4,...,13}
{
        		{\pgfmathtruncatemacro{\y}{\i+1}
        		\pgfmathtruncatemacro{\z}{2}
                 \draw[->,>=stealth,thick] (\z\i) -- (\z\y);}
}

\node[myvertex, fill=none,label=right:{$y^{i,j}_0$}] (1612)  at (\ymin+15,\xmax-3) {};
\node[myvertex, fill=none,label=right:{$y^{i,j}_1$}]  (1610) at (\ymin+15,\xmax-5) {};
\node[myvertex, fill=none] (168) at (\ymin+15,\xmax-7) {};
\node[myvertex, fill=none,label=right:{$y^{i,j}_a$}] (166) at (\ymin+15,\xmax-9) {};
\node[myvertex, fill=none] (164) at (\ymin+15,\xmax-11) {};
\node[myvertex, fill=none,label=right:{$y^{i,j}_n$}] (162) at (\ymin+15,\xmax-13) {};

\node[myvertex, fill=LimeGreen,label=right:{$\hat{y}^{i,j}_1$}] (1611) at (\ymin+15,\xmax-4) {};
\node[myvertex, fill=LimeGreen] (169) at (\ymin+15,\xmax-6) {};
\node[myvertex, fill=LimeGreen,label=right:{$\hat{y}^{i,j}_a$}] (167)  at (\ymin+15,\xmax-8) {};
\node[myvertex, fill=LimeGreen] (165) at (\ymin+15,\xmax-10) {};
\node[myvertex, fill=LimeGreen,label=right:{$\hat{y}^{i,j}_n$}] (163) at (\ymin+15,\xmax-12) {};

\foreach \i in {2,...,11}
{
        		{\pgfmathtruncatemacro{\y}{\i+1}
        		\pgfmathtruncatemacro{\z}{16}
                 \draw[->,>=stealth,thick] (\z\i) -- (\z\y);}
}

\node[myvertex, fill=none,label=below:{$y^{j,i}_0$}] (51)  at (\ymin+4,\xmax-14) {};
\node[myvertex, fill=none,label=below:{$y^{j,i}_1$}] (71) at (\ymin+6,\xmax-14) {};
\node[myvertex, fill=none] (91) at (\ymin+8,\xmax-14) {};
\node[myvertex, fill=none,label=below:{$y^{j,i}_b$}] (111) at (\ymin+10,\xmax-14) {};
\node[myvertex, fill=none] (131) at (\ymin+12,\xmax-14) {};
\node[myvertex, fill=none,label=below:{$y^{j,i}_n$}] (151) at (\ymin+14,\xmax-14) {};

\node[myvertex, fill=LimeGreen,label=below:{$\hat{y}^{j,i}_0$}] (61) at (\ymin+5,\xmax-14) {};
\node[myvertex, fill=LimeGreen] (81) at (\ymin+7,\xmax-14) {};
\node[myvertex, fill=LimeGreen,label=below:{$\hat{y}^{j,i}_b$}] (101) at (\ymin+9,\xmax-14) {};
\node[myvertex, fill=LimeGreen] (121) at (\ymin+11,\xmax-14) {};
\node[myvertex, fill=LimeGreen,label=below:{$\hat{y}^{j,i}_n$}] (141) at (\ymin+13,\xmax-14) {};

\foreach \i in {5,...,14}
{
        		{\pgfmathtruncatemacro{\y}{\i+1}
        		\pgfmathtruncatemacro{\z}{1}
                 \draw[<-,>=stealth,thick] (\i\z) -- (\y\z);}
}

\foreach \i in {5,7,9,11,13}
{
        \foreach \j in {12,10,8,6,4}
        {
                \node[myvertex,fill=none] (\i\j) at (\i,\j) {};
        }
}

\node[myvertex,fill] at (512) {};
\node[myvertex,fill] at (58) {};
\node[myvertex,fill] at (56) {};
\node[myvertex,fill] at (76) {};
\node[myvertex,fill] at (74) {};
\node[myvertex,fill] at (712) {};
\node[myvertex,fill] at (912) {};
\node[myvertex,fill] at (98) {};
\node[myvertex,fill] at (116) {};
\node[myvertex,fill] at (114) {};
\node[myvertex,fill] at (1110) {};
\node[myvertex,fill] at (1310) {};
\node[myvertex,fill] at (134) {};
\node[myvertex,fill] at (136) {};
\node[myvertex,fill] at (138) {};

\node[] at (9.5,8.5) {$p^{i,j}_{a,b}$};

\foreach \i in {5,7,9,11,13}
{
        \foreach \j in {12,10,8,6}
        {
        		{\pgfmathtruncatemacro{\y}{\j-2}
                 \draw[->,>=stealth,thick] (\i\j) -- (\i\y);}
        }
}

\foreach \i in {5,7,9,11}
{
        \foreach \j in {12,10,8,6,4}
        {
        		{\pgfmathtruncatemacro{\y}{\i+2}
                 \draw[->,>=stealth,thick] (\i\j) -- (\y\j);}
        }
}

\draw[->,>=stealth,thick] (515) -- (512);
\draw[->,>=stealth,thick] (715) -- (712);
\draw[->,>=stealth,thick] (915) -- (912);
\draw[->,>=stealth,thick] (1115) -- (1112);
\draw[->,>=stealth,thick] (1315) -- (1312);

\draw[->,>=stealth,thick] (54) -- (51);
\draw[->,>=stealth,thick] (74) -- (71);
\draw[->,>=stealth,thick] (94) -- (91);
\draw[->,>=stealth,thick] (114) -- (111);
\draw[->,>=stealth,thick] (134) -- (131);

\draw[->,>=stealth,thick] (212) -- (512);
\draw[->,>=stealth,thick] (210) -- (510);
\draw[->,>=stealth,thick] (28) -- (58);
\draw[->,>=stealth,thick] (26) -- (56);
\draw[->,>=stealth,thick] (24) -- (54);

\draw[->,>=stealth,thick] (1312) -- (1612);
\draw[->,>=stealth,thick] (1310) -- (1610);
\draw[->,>=stealth,thick] (138) -- (168);
\draw[->,>=stealth,thick] (136) -- (166);
\draw[->,>=stealth,thick] (134) -- (164);

\node[myvertex,fill=none,label=right:{$s_2$}] (1516) at (15,16) {}; %
\node[myvertex,fill=none,label=right:{$t_1$}] (1714) at (17,14) {}; %
\node[myvertex,fill=none,label=left:{$s_1$}] (12) at (1,2) {}; %
\node[myvertex,fill=none,label=left:{$t_2$}] (30) at (3,0) {}; %

\draw[->,>=stealth,thick] (1516) -- (1315); %
\draw[->,>=stealth,thick] (1612) -- (1714); %
\draw[->,>=stealth,thick] (12) -- (24); %
\draw[->,>=stealth,thick] (51) -- (30);%

\begin{pgfonlayer}{background}
\draw[draw=red,line width=0.3mm,fill=none] (9,8) circle[radius=3mm];

\draw[draw=red,line width=0.3mm,fill=none] (10,1) circle[radius=3mm];

\draw[draw=red,line width=0.3mm,fill=none] (8,15) circle[radius=3mm];

\draw[draw=red,line width=0.3mm,fill=none] (16,7) circle[radius=3mm];

\draw[draw=red,line width=0.3mm,fill=none] (2,9) circle[radius=3mm];

\draw[ProcessBlue,opacity=0.2,line width=5pt,line cap=round] (1516.center) -- (1315.center)
	-- (915.center)
	-- ($(101)-(1,0)$)
	-- (51.center)
	-- (30.center); 
	
\draw[red,opacity=0.2,line width=5pt,line cap=round] (12.center) -- (24.center)
	-- (28.center)
	-- ($(167)+(0,1)$)
	-- (1612.center)
	-- (1714.center); 
\end{pgfonlayer}

\end{tikzpicture}
\caption{An illustration of the construction of the paths $X^{i,j}$, $Y^{i,j}$, $X^{j,i}$, $Y^{j,i}$ and the grid $P^{i,j}$, where $i<j$ and $n=5$.
	The brown vertices whose labels have subscript $a$ have weight $Ma$.
	The green vertices whose labels have subscript $a$ have weight $M(n+1-a)$.
	The black vertices have weight $1$ and the white vertices are undeletable.
	The red, and the blue, highlighted paths are the $s_1t_1$-path, and $s_2t_2$-path respectively, that survive after deleting the solution vertices from $X^{i,j}$, $Y^{i,j}$, $X^{j,i}$ and $Y^{j,i}$.
	The unique common intersection point of the two highlighted paths is $p^{i,j}_{a,b}$, which is picked by the solution.
	The vertices picked by the solution are encircled in red.}
\label{fig:full-construction}
\end{figure}

For each $i \in V(H)$,
let $Z^i$ be the path $(z^i_n,\hat{z}^i_n, z^i_{n-1}, \hat{z}^i_{n-1}, \ldots, z^i_1, \hat{z}^i_1, z^i_0)$ on $2n+1$ vertices.
These paths are called the $\ZZ$-paths and are added to $D$.
For each ordered pair $(i,j) \in V(H) \times V(H)$, such that $\{i,j\} \in E(H)$,
let $X^{i,j} =(x^{i,j}_n,  \hat{x}^{i,j}_n, \allowbreak x^{i,j}_{n-1}, \hat{x}^{i,j}_{n-1}, \ldots , \allowbreak x^{i,j}_1, \hat{x}^{i,j}_1, x^{i,j}_0)$ and 
$Y^{i,j} =(y^{i,j}_n,  \hat{y}^{i,j}_n, y^{i,j}_{n-1}, \hat{y}^{i,j}_{n-1}, \ldots , \allowbreak y^{i,j}_1, \hat{y}^{i,j}_1, y^{i,j}_0)$ be paths on $2n +1$ vertices each.
These paths are called $\XX$-paths and $\YY$-paths, respectively, and are also added to $D$.
For each $a \in \{0,\dots,n\}$, set $\wt(z^i_a) = \wt(x^{i,j}_a) = \wt(y^{i,j}_a) = \ww +1$, that is, these vertices are undeletable.
Further set $\wt(\hat{z}^i_a) = M$, %
 $\wt(\hat{x}^{i,j}_a) =M a$ and
$\wt(\hat{y}^{i,j}_a) =M(n+1-a)$. 
Observe that $\wt(\hat{x}^{i,j}_a) + \wt(\hat{y}^{i,j}_a)=M(n+1)$.

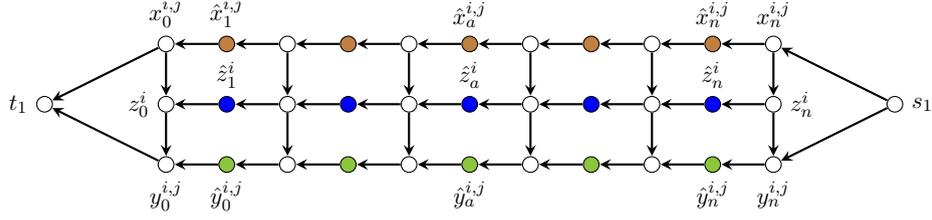
\begin{figure}
\centering
\begin{tikzpicture}[%
scale=.8,
transform shape,
myvertex/.style={circ,scale=2}
]

\node[myvertex,fill=none,label=above:{$x^{i,j}_0$}] (315) at (\ymin+2,\xmax) {};
\node[myvertex, fill=none] (515) at (\ymin+4,\xmax) {};
\node[myvertex, fill=none] (715) at (\ymin+6,\xmax) {};
\node[myvertex, fill=none] (915) at (\ymin+8,\xmax) {};
\node[myvertex, fill=none] (1115) at (\ymin+10,\xmax) {};
\node[myvertex, fill=none,label=above:{$x^{i,j}_n$}] (1315) at (\ymin+12,\xmax) {};

\node[myvertex, fill=brown,label=above:{$\hat{x}^{i,j}_1$}] (415) at (\ymin+3,\xmax) {};
\node[myvertex, fill=brown] (615) at (\ymin+5,\xmax) {};
\node[myvertex, fill=brown,label=above:{$\hat{x}^{i,j}_a$}] (815) at (\ymin+7,\xmax) {};
\node[myvertex, fill=brown] (1015) at (\ymin+9,\xmax) {};
\node[myvertex, fill=brown,label=above:{$\hat{x}^{i,j}_n$}] (1215) at (\ymin+11,\xmax) {};

\foreach \i in {3,...,12}
{
        		{\pgfmathtruncatemacro{\y}{\i+1}
        		\pgfmathtruncatemacro{\z}{\xmax}
                 \draw[<-,>=stealth,thick] (\i\z) -- (\y\z);}
}

\node[myvertex,fill=none,label=left:{$z^{i}_0$}] (314) at (\ymin+2,\xmax-1) {};
\node[myvertex, fill=none] (514) at (\ymin+4,\xmax-1) {};
\node[myvertex, fill=none] (714) at (\ymin+6,\xmax-1) {};
\node[myvertex, fill=none] (914) at (\ymin+8,\xmax-1) {};
\node[myvertex, fill=none] (1114) at (\ymin+10,\xmax-1) {};
\node[myvertex, fill=none,label=right:{$z^{i}_n$}] (1314) at (\ymin+12,\xmax-1) {};

\node[myvertex, fill=blue, label=above:{$\hat{z}^{i}_1$}] (414) at (\ymin+3,\xmax-1) {};
\node[myvertex, fill=blue] (614) at (\ymin+5,\xmax-1) {};
\node[myvertex, fill=blue,label=above:{$\hat{z}^{i}_a$}] (814) at (\ymin+7,\xmax-1) {};
\node[myvertex, fill=blue] (1014) at (\ymin+9,\xmax-1) {};
\node[myvertex, fill=blue,label=above:{$\hat{z}^i_n$}] (1214) at (\ymin+11,\xmax-1) {};

\foreach \i in {3,...,12}
{
        		{\pgfmathtruncatemacro{\y}{\i+1}
        		\pgfmathtruncatemacro{\z}{\xmax-1}
                 \draw[<-,>=stealth,thick] (\i\z) -- (\y\z);}
}

\node[myvertex, fill=none,label=below:{$y^{i,j}_0$}] (313)  at (\ymin+2,\xmax-2) {};
\node[myvertex, fill=none] (513)  at (\ymin+4,\xmax-2) {};
\node[myvertex, fill=none] (713) at (\ymin+6,\xmax-2) {};
\node[myvertex, fill=none] (913) at (\ymin+8,\xmax-2) {};
\node[myvertex, fill=none] (1113) at (\ymin+10,\xmax-2) {};
\node[myvertex, fill=none,label=below:{$y^{i,j}_n$}] (1313) at (\ymin+12,\xmax-2) {};

\node[myvertex, fill=LimeGreen,label=below:{$\hat{y}^{i,j}_0$}] (413) at (\ymin+3,\xmax-2) {};
\node[myvertex, fill=LimeGreen] (613) at (\ymin+5,\xmax-2) {};
\node[myvertex, fill=LimeGreen,label=below:{$\hat{y}^{i,j}_a$}] (813) at (\ymin+7,\xmax-2) {};
\node[myvertex, fill=LimeGreen] (1013) at (\ymin+9,\xmax-2) {};
\node[myvertex, fill=LimeGreen,label=below:{$\hat{y}^{i,j}_n$}] (1213) at (\ymin+11,\xmax-2) {};

\foreach \i in {3,...,12}
{
        		{\pgfmathtruncatemacro{\y}{\i+1}
        		\pgfmathtruncatemacro{\z}{\xmax-2}
                 \draw[<-,>=stealth,thick] (\i\z) -- (\y\z);}
}

\draw[->,>=stealth,thick] (315) -- (314);
\draw[->,>=stealth,thick] (515) -- (514);
\draw[->,>=stealth,thick] (715) -- (714);
\draw[->,>=stealth,thick] (915) -- (914);
\draw[->,>=stealth,thick] (1115) -- (1114);
\draw[->,>=stealth,thick] (1315) -- (1314);

\draw[->,>=stealth,thick] (314) -- (313);
\draw[->,>=stealth,thick] (514) -- (513);
\draw[->,>=stealth,thick] (714) -- (713);
\draw[->,>=stealth,thick] (914) -- (913);
\draw[->,>=stealth,thick] (1114) -- (1113);
\draw[->,>=stealth,thick] (1314) -- (1313);
         
\node[myvertex,fill=none,label=left:{ $t_1$}] (114) at (1,14) {}; %
\node[myvertex,fill=none,label=right:{$s_1$}] (1414) at (15,14) {}; %

\draw[->,>=stealth,thick] (1414) -- (1315); %
\draw[->,>=stealth,thick] (1414) -- (1313); %
\draw[->,>=stealth,thick] (315) -- (114); %
\draw[->,>=stealth,thick] (313) -- (114);%

\end{tikzpicture}
\caption{The paths $X^{i,j}$, $Z_i$ and $Y^{i,j}$ are shown.
	If the index of a vertex is $a$, then its weight is $M$ if it the vertex is blue, $Ma$ if the vertex is brown, $M(n+1-a)$ if the vertex is green and $\ww+1$ otherwise.
	This gadget together with the definition of the target weight $\ww$ of the solution, ensures that a solution picks exactly one vertex from each of the three paths.
	Moreover, the three vertices chosen from these three paths have the same index.}
\label{fig:synchronization-hardness}
\end{figure}

Next, we describe the gadget that synchronizes the vertices that a solution for \wttwodmc\ picks from the $\XX$-, $\YY$- and $\ZZ$-paths.
For each ordered pair $(i,j) \in V(H) \times V(H)$ such that $\{i,j\} \in E(H)$, and for every $a \in \{0,1, \ldots, n\}$, we add the edges $(x^{i,j}_a,z^i_a)$ and $(z^i_a, y^{i,j}_a)$  %
to~$D$ (see \cref{fig:synchronization-hardness} for an illustration). 
Further we add two terminal-pairs $(s_1,t_1)$ and $(s_2,t_2)$ together with the following incident edges.
For each $\{i,j\} \in E(H)$, such that $i < j$, the edges $(s_1,x^{i,j}_n), (x^{i,j}_0,t_1), (s_1,y^{i,j}_n), (y^{i,j}_0,t_1)$ are added to $D$.
Further the edges $(s_2,x^{j,i}_n), \allowbreak (x^{j,i}_0,t_2), \allowbreak (s_2,y^{j,i}_n), \allowbreak (y^{j,i}_0,t_2)$ are added to $D$.
The edges described here ensure that 
for each $(i,j)\in V(H) \times V(H)$ such that $(i,j) \in E(H)$, there exists an $a \in \{0,1,\ldots,n\}$ such that the solution for \wttwodmc\ picks $\hat{z}^i_a,\hat{x}^{i,j}_a$, and $\hat{y}^{i,j}_a$

The next gadget ensures that vertices of the \wttwodmc\ solution in the paths described above are the map of a \emph{valid homomorphism}.
For every edge $\{i,j\} \in E(H)$ such that $i<j$, we add a grid $P^{i,j}$ on the vertex set $\{p^{i,j}_{a,b}\mid a,b \in \{1, \ldots, n\} \}$ to $D$.
The edge set contains the \emph{column edges} $\{(p^{i,j}_{a,b},p^{i,j}_{a+1,b}) \mid a \in \{1, \ldots, n-1\}, b \in \{1, \ldots, n\}\}$, and the \emph{row edges} $\{(p^{i,j}_{a,b},p^{i,j}_{a,b+1}) \mid a \in \{1, \ldots, n\}, b \in \{1, \ldots, n-1\}\}$.
The weight function of the vertices of this grid is defined as follows: $\wt(p^{i,j}_{a,b}) =1$ if $(v^i_a,v^i_b) \in E(G)$, and $\wt(p^{i,j}_{a,b}) = \ww$ otherwise.

Further we add edges to $D$ connecting the paths $X^{i,j},Y^{i,j}, X^{j,i}$ and $Y^{j,i}$ to $P^{i,j}$.
For every $a \in \{1,\ldots,n\}$, the following edges are added:
$(x^{i,j}_a, p^{i,j}_{a,1}), \allowbreak  (p^{i,j}_{a,n}, y^{i,j}_{a-1}), \allowbreak (x^{j,i}_a, p^{i,j}_{1,a})$, and $(p^{i,j}_{n,a}, y^{j,i}_{a-1})$.
This finishes the construction of $D$.
See \cref{fig:full-construction} for an illustration of the construction.

We next show that $(G,H)$ is a \yes-instance of \psismall\ if and only if $(D,k',(s_i,t_i)_{i \in [2]}, \wt,\ww)$ is a \yes-instance of \wttwodmc.

\paragraph*{Forward direction.} Let $\phi: V(H) \to \{1, \ldots, n\}$ be a solution to the \psismall-instance $(G,H)$, that is for every $\{i,j\} \in E(H)$, we have $(v^i_{\phi(i)},v^j_{\phi(j)}) \in E(G)$. 
We choose $S \coloneqq \{\hat{x}^{i,j}_{\phi(i)}, \hat{y}^{i,j}_{\phi(i)}: (i,j) \in V(H) \times V(H) \text{ and } \{i,j\} \in E(H)\} \cup \{\hat{z}^i_{\phi(i)} : i \in V(H)\} \cup \{p^{i,j}_{\phi(i),\phi(j)} : (i,j) \in V(H) \times V(H) \text{ and }\{i,j\} \in E(H), i <j\}$ 
and claim that it is a solution to the constructed \wttwodmc-instance $(D,k',(s_i,t_i)_{i \in [2]}, \wt,\ww)$.
Observe that the set $S$ contains a vertex from $Z^i$, $X^{i,j}$, $Y^{i,j}$ and the grid $P^{i,j}$ for every $(i,j) \in V(H) \times V(H)$ with $\{i,j\} \in E(H)$.
Further observe that $|S| = 5k+h$ and $\wt(S) = M (2k (n+1) +h) +k$ since the weight of the vertices on the $\XX$- and $\YY$-paths in $S$ is $2kM(n+1)$, the weight of the vertices of $\ZZ$-paths in $S$ is $Mh$ and the weight of the grid vertices in $S$ is $k$. %

We claim that $D-S$ has no $s_1t_1$-path and no $s_2t_2$-path.
Here, we show that $D-S$ has no $s_1t_1$-path.
That there is also no $s_2t_2$-path follows by symmetric arguments.
Observe from the construction that every $(s_1,t_1)$ path contains $x^{i,j}_n$ or $y^{i,j}_n$,
for some $i,j$, as its first internal vertex.
There are four kinds of $s_1t_1$-paths in $D$.
The first one traverses $X^{i,j}$ fully until $x^{i,j}_0$ (which is an in-neighbour of $t_1$).
The second one traverses $Y^{i,j}$ fully.
Both, paths of the first and the second kind, are hit by $S$ as $S$ contains a vertex from $X^{i,j}$ and a vertex from $Y^{i,j}$.

The third kind of $s_1t_1$-path traverses some subpath of $X^{i,j}$, say until the vertex $x^{i,j}_a$, 
jumps to $Z^i$ at vertex $z^i_a$,
traverses $Z^i$ until say $z^i_b$ where $b \leq a$, jumps to $Y^{i,j}$ at vertex $y^{i,j}_b$, and then traverses $Y^{i,j}$ until $y^{i,j}_0$.  
Since $\hat{x}^{i,j}_{\phi(i)}, \hat{y}^{i,j}_{\phi(i)}, \hat{z}^i_{\phi(i)} \in S$, if such a path exists then $a \geq  \phi(i)$ and $b \leq \phi(i)$, thus the $y^{i,j}_by^{i,j}_0$-subpath of $Y^{i,j}$ traversed by such an $s_1,t_1$-path contains a vertex of $S$, namely $\hat{y}^{i,j}_{\phi(i)}$.

The last kind of $s_1,t_1$-path passes through the grid $P_{i,j}$.
Such a path traverses $X^{i,j}$ until $x^{i,j}_a$,
 jumps to the first column of the grid at the vertex $p^{i,j}_{a,1}$,
 traverses the grid $P^{i,j}$ to get to a vertex $p^{i,j}_{b,n}$ %
 (there are potentially many ways to reach this vertex inside the grid), for some $b \geq a$, 
 then jumps to $Y^{i,j}$ at the vertex $y^{i,j}_{b-1}$, and then traverses $Y^{i,j}$ until $y^{i,j}_0$. 
Since $\hat{x}^{i,j}_{\phi(i)} \in S$, if $a < \phi(i)$, the path is hit by $S$.
Since $\hat{y}^{i,j}_{\phi(i)} \in S$, if $\phi(i) < b$, the path is hit by $S$.
 Otherwise $a=b=\phi(i)$.
 In this case the $s_1t_1$-path under consideration uses exactly the vertices of the $p^{i,j}_{\phi(i),i}p^{i,j}_{\phi(i),n}$-subpath of the grid.
 Since $p^{i,j}_{\phi(i),\phi(j)} \in S$, this kind of $s_1t_1$-path is again hit by $S$.
 
 Thus, $S$ is a solution of the \wttwodmc-instance $(D,k',(s_i,t_i)_{i \in [2]}, \wt,\ww)$.

\paragraph*{Reverse direction.} Let $S \subseteq V(D)$ be a  solution for the \wttwodmc-instance $(D,k', \allowbreak (s_i,t_i)_{i \in [2]},\allowbreak \wt, \allowbreak \ww,)$.
We construct a function $\phi: V(H) \to \{1, \ldots, n\}$ such that for each $\{i,j\} \in E(H)$, $(v^i_{\phi(i)},v^j_{\phi(j)}) \in E(G)$.

Note that, since for every path $X^{i,j}$ and for every path $Y^{i,j}$, there is an $s_1t_1$-path or an $s_2t_2$-path containing only this path, $S$ contains at least one vertex from each $X^{i,j}$ and each $Y^{i,j}$.
Also, since the vertices $x^{i,j}_n$ and $y^{i,j}_0$ are undeletable, the set $S$ contains a deletable vertex of $Z^i$, as otherwise there is an $s_1t_1$-path or an $s_2t_2$-path starting from $x^{i,j}_n$ and then jumping to the path $Z^i$ at the vertex $z^i_n$, traversing the path $Z^i$ until $z^i_0$, then jumping to $y^{i,j}_0$ (which is an in-neighbour of $t_1$).

Fix $(i,j) \in V(H) \times V(H)$ such that $\{i,j\} \in E(H)$ and $i <j$.
The case when $j >i$ is symmetric.
We eventually show that $S$ intersects $X^{i,j}$, $Y^{i,j}$ and $Z^i$, each at exactly one vertex. Towards this, let $a \in \{1, \ldots,  n\}$ be the largest index such that $\hat{x}^{i,j}_a \in S$, and let $b \in \{1, \ldots, n\}$ be the smallest index such that $\hat{y}^{i,j}_b \in S$.
We first claim that $a \geq b$.
Suppose $a < b$, then consider the following $s_1t_1$-path in $G-S$.
The path first visits $x^{i,j}_n$, traverses $X^{i,j}$ until reaching $x^{i,j}_a$ (note that none of the vertices on this subpath belong to $S$ so far, either due to the choice of $a$ or because they are undeletable), jumps to $z^{i}_a$ (note that this vertex is undeletable), and then jumps to $y^{i,j}_a$ and finally traverses $Y^{i,j}$ until reaching $y^{i,j}_0$ (again, none of the vertices of this subpath belongs to $S$ either due to the choice of $b$ or because they are undeletable).
This is a contradiction to the fact that $S$ is a solution.

From the above paragraph, $\wt(S) \cap (X^{i,j} \cup Y^{i,j}) \geq M \cdot a + M\cdot (n+1 -b) \geq M \cdot (n+1)$.
Note that if $a >b$, then $\wt(S) \cap (X^{i,j} \cup Y^{i,j}) > M \cdot (n+1) +M$.
Thus, if $S$ picks exactly one vertex from each $\XX$-, $\YY$- and $\ZZ$-paths, then the weight of $S$ restricted to the vertices in these paths is $4k M \cdot (n+1) + M h$.
Since the total weight of $S$ is $4k M \cdot (n+1) + M h +k$, and the weight of each vertex on the paths $X^{i,j}$, $Y^{i,j}$ or $Z^i$ is at least $M = k +1$, we conclude that $S$ indeed picks exactly one vertex from each of the $\XX$-, $\YY$- and $\ZZ$-paths.
Moreover, $a=b$. That is, 
$S \cap X^{i,j} = \hat{x}^{i,j}_{a}$ and $S \cap Y^{i,j}=\hat{y}^{i,j}_{a}$.

We now show that $S \cap Z^i = \hat{z}^i_{a}$.
Suppose not, then consider the following $s_1t_1$-path in $G-S$.
It first visits $x^{i,j}_n$, traverses $X^{i,j}$ until reaching $x^{i,j}_{a}$, jumps to $Z^i$ at the vertex $z^i_{a}$, traverses $Z^i$ until reaching $z^i_{a-1}$ (note that the only deletable vertex on this subpath is $\hat{z}^i_{a}$, which by our assumption is not in $S$), jumps to the path $Y^{i,j}$ at the vertex $y^{i,j}_{a-1}$, and then traverses $Y^{i,j}$ until reaching $y^{i,j}_0$.
Again, this yields a contradiction to $S$ being a solution.

So far, we have concluded that $S$ intersects each of the paths $X^{i,j}$, $Y^{i,j}$ and $Z^i$ in exactly one vertex.
In fact, there exists $a \in \{1, \ldots, n\}$ such that 
$S \cap X^{i,j} = \hat{x}^{i,j}_{a}$, $S \cap Y^{i,j}=\hat{y}^{i,j}_{a}$ and $S \cap Z^i = \hat{z}^i_{a}$. 
Hence, to construct $\phi$, define $\phi(i)=a$.

Fix $(i,j) \in V(H) \times V(H)$ and $i <j$.
Consider the $s_1t_1$-path, we call it $P_1$, that first visits $x^{i,j}_n$,then traverses $X^{i,j}$ until $x^{i,j}_{\phi(i)}$, jumps to the first column of the grid $P^{i,j}$ at vertex $p^{i,j}_{\phi(i),1}$, traverses the path in the $\phi(i)$-th row of $P^{i,j}$, that is the subpath from $(p^{i,j}_{\phi(i),1}$ to $p^{i,j}_{\phi(i),n})$, jumps to the path $Y^{i,j}$ at the vertex $y^{i,j}_{\phi(i)}$ and then traverses $Y^{i,j}$ until reaching $y^{i,j}_0$. 
Similarly, consider the $s_2t_2$-path, we call it $P_2$, that first visits $x^{j,i}_n$, traverses the path $X^{j,i}$ until reaching $x^{j,i}_{\phi(j)}$, jumps to the first row of the grid $P^{i,j}$ at the vertex %
$p^{i,j}_{1, \phi(j)}$
traverses the path in the $\phi(j)$-th column of $P^{i,j}$, that is the subpath from $p^{i,j}_{1, \phi(j)}$ to $p^{i,j}_{n, \phi(j)}$, jumps to the path $Y^{j,i}$ at $y^{j,i}_{\phi(j)-1}$ and then traverses $Y^{j,i}$ until reaching $y^{j,i}_0$.
Observe that the paths $P_1$ and $P_2$ do not intersect any of the vertices of $S$ that are on some $\XX$-, $\YY$- or $\ZZ$-path.
Also, the paths $P_1$ and $P_2$ intersect at exactly one vertex of the grid $P_{i,j}$, which is $p^{i,j}_{\phi(i), \phi(j)}$.
Since $S$ has at most $k$ vertices that are in none of the $\XX$-, $\YY$- or $\ZZ$-path, $S$ contains at most one vertex from each $P^{i,j}$.
Thus, $p^{i,j}_{\phi(i), \phi(j)} \in S$.
As $\wt(S) \leq W$, we can deduce from the construction of $D$ that $(v^i_{\phi(i)},v^j_{\phi(j)}) \in E(G)$. 
\end{proof}

\bibliographystyle{plainurl}
\bibliography{references}

\appendix
\section{Hardness of arbitrary CSPs with permutation constraints}\label{app:edge-choice}

In this section we explain in bigger detail the claim from the introduction
that without any control on the complexity of permutation constraints, the obtained CSP instance is W[1]-hard when parameterized by the number of variables and permutation constraints.

Recall the classic W[1]-hard \textsc{Multicolored Clique} problem: the input consists of an integer $k$ being the parameter, a graph $G$, and a partition of $V(G)$ into $k$
independent sets $V_1,V_2,\ldots,V_k$; the goal is to find a $k$-clique in $G$, which necessarily needs to contain exactly one vertex from each set $V_i$. 
By a padding argument, we can assume $|V_1| = |V_2| = \ldots = |V_k| = n$. For each $1 \leq i \leq k$, enumerate the vertices of $V_i$ as $v_{i,0}, v_{i,1}, \ldots, v_{i,n-1}$.

For every $1 \leq i \leq k$, create a variable $x_i$ with domain $\{0,1,\ldots,n-1\}$; setting the value $x_i = a$ corresponds to choosing a vertex $v_{i,a}$ to our clique.
For every $1 \leq i,j \leq k$, $i \neq j$, create a variable $y_{i,j}$ with domain $\{0,1,\ldots,n-1\} \times \{0,1,\ldots,n-1\}$ ordered lexicographically;
setting the value $y_{i,j} = (a,b)$ corresponds to choosing vertices $v_{i,a}$ and $v_{j,b}$ to our clique.
Bind the variables using the following constraints:
\begin{itemize}
\item For every $1 \leq i,j \leq k$, $i \neq j$, we express a constraint 
\[ \bigwedge_{0 \leq a,b < n} y_{i,j} = (a,b) \Longrightarrow x_i=a \]
as the following conjunction:
\[ \left(\bigwedge_{0 \leq a < n-1} (x_i \leq a) \vee (y_{i,j} \geq (a+1,0))\right) \wedge \left(\bigwedge_{1 < a < n} (y_{i,j} \leq (a-1,n-1)) \vee (x_i \geq a)\right). \]
\item For every $1 \leq i < j \leq k$, we introduce a permutation constraint
\[ y_{j,i} = \pi(y_{i,j}), \]
where $\pi : \{0,1,\ldots,n-1\} \times \{0,1,\ldots,n-1\} \to \{0,1,\ldots,n-1\} \times \{0,1,\ldots,n-1\} $ is defined as $\pi(a,b) = (b,a)$.
\end{itemize}
Finally, for every $1 \leq i < j \leq k$, we restrict the permutation constraint between $y_{i,j}$ and $y_{j,i}$ to only those values $y_{i,j} = (a,b)$ where
$v_{i,a} v_{j,b} \in E(G)$. 

The above is an encoding of the input \textsc{Multicolored Clique} instance a CSP instance with $k + k(k-1) = k^2$ variables and $\binom{k}{2}$ permutation constraints. 
Hence, we cannot hope for an FPT algorithm for our CSP instances, with only the number of variables and permutation constraints as parameters; we need some structural parameter
of the obtained permutation constraints. Note that the permutation $\pi$ used above has a grid minor of size $n$ in its permutation matrix.

\section{Twin-width to grid-rank}\label{app:tw:to:gr}

\begin{theorem}[Marcus and Tardos~\cite{MarcusTardos}]
	\label{thm:MarcusTardos}
	For every integer $k$, there is some $c_k \le 2^{\Oh(k \log k)}$ 
	such that every $n\times m$ 0-1 matrix $\bfA$
	with at least $c_k\max(n, m)$ 1-entries has a $k$-grid minor.
	Moreover, if it exists, such a grid minor can be found in $2^{\Oh(k \log k)}n^{\Oh(1)}$ time.
\end{theorem}

\begin{proposition}\label{prop:int:sequence}
	Let $n \geq 2$ and $a_1, \ldots, a_n$ be non-negative integers,
	and let $s = \sum_{i=1}^n a_i/n$.
	Then, there exists an $i \in [n-1]$ such that $a_i + a_{i+1} \le 4s-1$.
\end{proposition}
\begin{proof}
	Suppose for a contradiction that for all $i \in [n-1]$, $a_i + a_{i+1} \ge 4s$.
	First suppose that $n$ is even. Then,
	\begin{align*}
		\sum_{i=1}^n a_i = a_1 + a_2 + \sum_{i = 3}^n a_i \ge 4s + \sum_{i=3}^n a_i \ge \cdots \ge \frac{n}{2}4s > ns,
	\end{align*}
	a contradiction.
	Next suppose that $n$ is odd which means that $n \ge 3$. In this case,
	\begin{align*}
		\sum_{i=1}^n a_i = a_1 + \sum_{i = 2}^n a_i \ge a_1 + \frac{n-1}{2}4s \ge 2(n-1)s = ns\left(2 - \frac{2}{n}\right) > ns,
	\end{align*}
	a contradiction since $2/n < 1$.
\end{proof}

\begin{proposition}\label{prop:gridminor:or:contraction}
	Let $\bfA$ be a 0-1 matrix, let $k$ be a positive integer, 
	and let $c_k$ be the constant from \cref{thm:MarcusTardos}.
	One can in $2^{\Oh(k \log k)}n^{\Oh(1)}$ time find either
	\begin{itemize}
		\item a $k$-grid minor in $\bfA$, or
		\item a $4c_k$-contraction sequence of $\bfA$, respecting the order of the rows and columns of $\bfA$.
	\end{itemize}
\end{proposition}
\begin{proof}
	We do the following greedily.
	Assume the number of rows in $\bfA$ is at least the number of its columns, 
	otherwise we swap the roles of rows and columns.
	Let $o$ be the average number of 1-entries in each row of $\bfA$.
	If $o > c_k$, then by \cref{thm:MarcusTardos},
	we can find a $k$-grid minor of $\bfA$ in $2^{\Oh(k \log k)}n^{\Oh(1)}$ time and we are done.
	Otherwise, $o \le c_k$, which implies by \cref{prop:int:sequence} that
	$\bfA$ has two consecutive rows which together have at most $4c_k$ ones.
	We contract these two rows and repeat.
	In the following iterations, the error, or \emph{red} entries are treated like $1$s by the Marcus Tardos theorem.
	This is because such entries point to the existence of a $1$-entry 
	in the original matrix $\bfA$ which could be used to form a $k$-grid minor.
	Note that no row or column ever exceeds 
	$4c_k$ red entries if the algorithm succeeds to give a contraction sequence.
\end{proof}

\thmgminorgrank*
\begin{proof}
	If $\bfA$ does not have a $k$-grid minor, then by \cref{prop:gridminor:or:contraction} there is a 
	$4c_k$-contraction sequence of $\bfA$ that respects the order of the rows and columns of $\bfA$.
	In the words of~\cite{BonnetKTW22}, this means that $\bfA$ is $4c_k$-twin-ordered,
	and a theorem in~\cite{BonnetKTW22} asserts that $\bfA$ cannot have a $(8c_k + 2)$-mixed minor,
	which in turn implies that $\gridrank(\bfA) \le (8c_k+2) = 2^{\Oh(k \log k)}$.
\end{proof}

\end{document}